\documentclass[11pt]{article}
\usepackage{hyperref}
\usepackage{amssymb}
\usepackage{amsmath}
\usepackage{mdframed}
\usepackage{subfig}
\usepackage{cancel}
\usepackage{enumerate}
\usepackage{color}
\usepackage{mathtools}
\usepackage{pdfpages}
\usepackage{mathrsfs}
\usepackage{enumitem}
\usepackage{mathtools}
\usepackage{graphicx}
\usepackage{mdframed}
\usepackage{amsfonts}
\usepackage{cancel}
\usepackage{placeins}
\usepackage{amsthm}
\usepackage{tikz}
\usetikzlibrary{shapes,arrows,automata}
\usepackage{xcolor}
\usetikzlibrary{arrows.meta}

\newtheorem{definition}{Definition}[section]
\newtheorem{theorem}{Theorem}[section]

\newtheorem{lemma}{Lemma}[section]
\newtheorem{proposition}{Proposition}[section]
\newtheorem{corollary}{Corollary}[section]

\let\Item\item

\begin{document}
\vspace*{-2.5cm}

\centerline{{\huge Integral feedback in synthetic biology:}}

 \medskip

 \centerline{{\huge Negative-equilibrium catastrophe}}
 
\medskip
\bigskip

\centerline{
\renewcommand{\thefootnote}{$1$}
{\Large Tomislav Plesa,\footnote{
Department of Bioengineering, Imperial College London,
Exhibition Road, London, SW7 2AZ, UK}\renewcommand{\thefootnote}{$*$}\footnote{
Corresponding author and lead contact. E-mail: t.plesa@ic.ac.uk}
\qquad 
Alex Dack$^1$,
\qquad 
Thomas E. Ouldridge$^1$
}}

\medskip
\bigskip

\noindent
{\bf Abstract}: A central goal of synthetic biology is the design of molecular
controllers that can manipulate the dynamics of intracellular networks
in a stable and accurate manner. To address the fact that detailed knowledge 
about intracellular networks is unavailable, integral-feedback controllers (IFCs) 
 have been put forward for controlling molecular abundances. 
These controllers can maintain accuracy in spite of the uncertainties in the controlled networks.
However, this desirable feature is achieved only if stability is also maintained. 
In this paper, we show that molecular IFCs can suffer from a hazardous instability called 
\emph{negative-equilibrium catastrophe} (NEC), whereby
all nonnegative equilibria vanish under the action of the controllers, 
and some of the molecular abundances blow up.
We show that unimolecular IFCs do not exist due to a NEC.
We then derive a family of bimolecular IFCs that are safeguarded against NECs
when uncertain unimolecular networks, with any number of molecular species, are controlled. 
However, when IFCs are applied on uncertain bimolecular (and hence most intracellular) networks, 
we show that preventing NECs generally becomes an intractable problem
as the number of interacting molecular species increases.

\section{Introduction} \label{sec:intro}
A main objective in synthetic biology is to control
living cells~\cite{SynthBio1,Control1,Toggle,Repressilator,PreTranscription,PostTranscription}
- a challenging problem that requires addressing a number of complicating factors displayed 
by intracellular networks:
\begin{enumerate}
\item[{\rm (N)}] \textbf{Nonlinearity}. Intracellular networks are \emph{bimolecular} (nonlinear), 
i.e. they include reactions involving two reacting molecules.
\item[{\rm (HD)}] \textbf{Higher-dimensionality}.
Intracellular networks are \emph{higher-dimensional}, 
i.e. they contain larger number of coupled molecular species.
\item[{\rm (U)}] \textbf{Uncertainty}. 
The experimental information about the structure,
rate coefficients and initial conditions of intracellular networks is 
\emph{uncertain}/incomplete.
\end{enumerate} 
When embedded into an \emph{input} network satisfying properties (N), (HD) and (U),
an ideal molecular \emph{controller} network would ensure that the resulting 
\emph{output} network autonomously traces a predefined
dynamics in a \emph{stable} and \emph{accurate} manner over a desired time-interval.
Controllers that maintain accuracy in spite of suitable uncertainties
are said to achieve \emph{robust adaptation} (homeostasis) - 
a fundamental design principle of living 
systems~\cite{Adaptation,CellSignal,Glycolic,Chemotaxis1,Chemotaxis2,Chemotaxis3}.
Control can be sought over deterministic dynamics when all of the molecular
species are in higher-abundance~\cite{Feinberg,Kurtz}, or over 
 stochastic dynamics when some species are 
present at lower copy-numbers~\cite{RadekBook,CellCycle,Circadian}.
See also Figure~\ref{fig:Control_Theory}, and 
Appendices~\ref{app:background} and~\ref{app:biochemical_control}. 

\begin{figure}[!htbp]
\vskip  -2.7cm
\centerline{
\includegraphics[width=0.8\columnwidth]{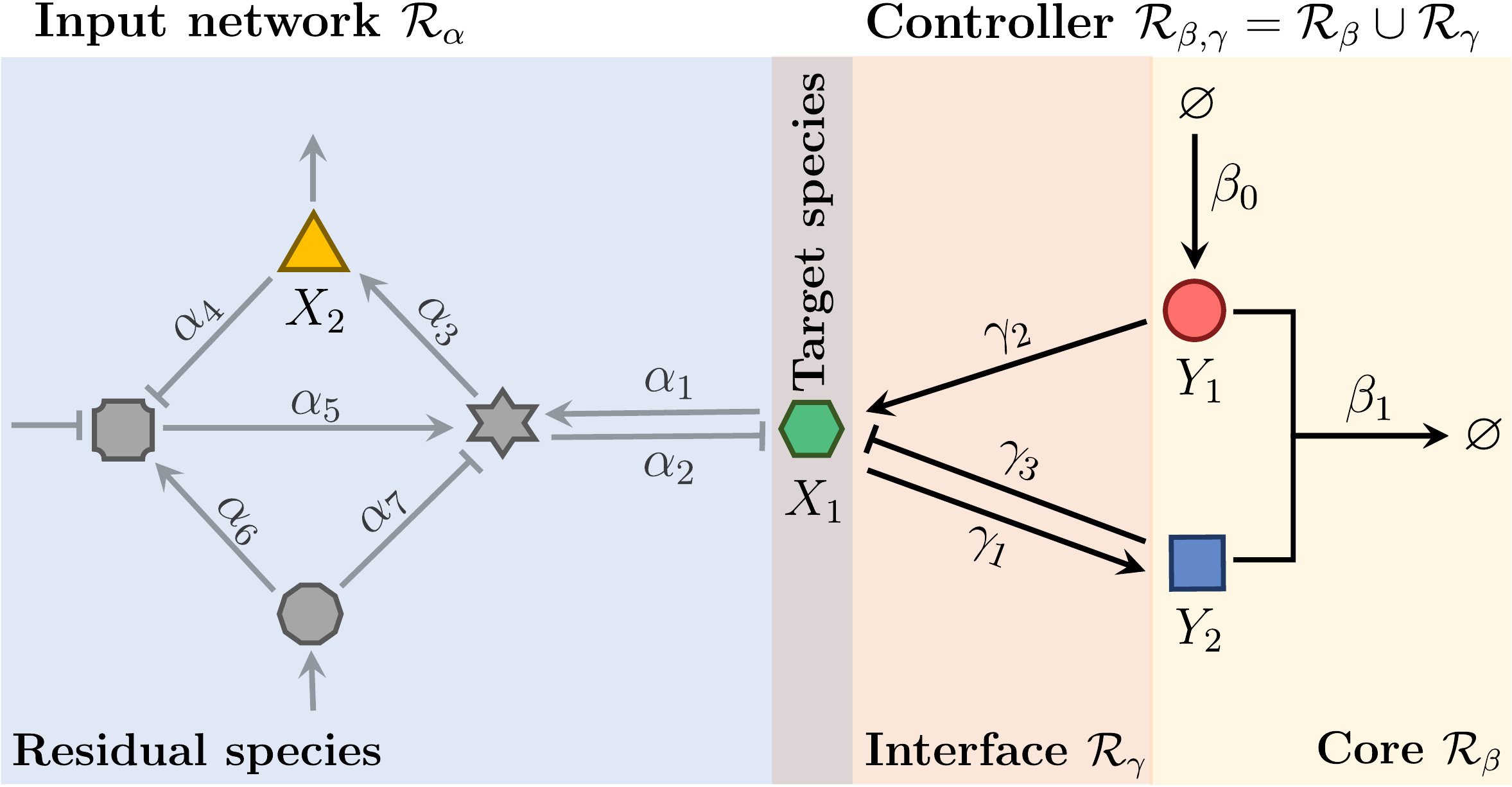}
}
\vskip -0.2cm
\caption{\it{\emph{Schematic representation of biochemical control}.
A black-box \emph{input} network $\mathcal{R}_{\alpha}$ is displayed, 
consisting of unknown biochemical interactions shown in grey, 
where $X_1 \to X_2$ (respectively, $X_1 \dashv X_2$) indicates 
that species $X_1$ influences species $X_2$ positively (respectively, negatively). 
The particular input network contains one \emph{target} species $X_1$, 
shown as a green hexagon, that can be interfaced with a controller. 
The rest of the input species, that cannot be interfaced with a controller,
are called \emph{residual} species; one of the residual species, $X_2$, 
is highlighted as a yellow triangle, while some of the other ones are displayed
in grey. 
A \emph{controller} network $\mathcal{R}_{\beta, \gamma}
= \mathcal{R}_{\beta} \cup \mathcal{R}_{\gamma}$ is also shown,
consisting of \emph{controlling} species $Y_1$ and $Y_2$
shown as a red circle and blue square, respectively,
whose predefined biochemical interactions are shown in black.
The controller consists of the \emph{core} 
$\mathcal{R}_{\beta} = \mathcal{R}_{\beta}(Y_1, Y_2)$, 
specifying how the controlling species interact among themselves, 
and the \emph{interface} $\mathcal{R}_{\gamma} = \mathcal{R}_{\gamma}(X_1, Y_1, Y_2)$, 
specifying how the controlling species $Y_1$ and $Y_2$ interact with the 
target species $X_1$. The composite network 
$\mathcal{R}_{\alpha,\beta,\gamma} = \mathcal{R}_{\alpha} \cup \mathcal{R}_{\beta, \gamma}$
is called an \emph{output} network. The particular controller $\mathcal{R}_{\beta, \gamma}$
displayed corresponds to the network~{\rm(\ref{eq:IFCnetapp})},
with $i = j = 1$, from {\rm Section~\ref{sec:nonlinear}}.}}  \label{fig:Control_Theory}
\end{figure}  

In context of electro-mechanical systems, accuracy robust to some uncertainties 
can be achieved via so-called \emph{integral-feedback controllers} (IFCs)~\cite{Control_theory}. 
Loosely speaking, IFCs dynamically calculate a time-integral of a difference (error) between
the target and actual values of the controlled variable. 
The error is then used to decrease (respectively, increase) the controlled variable
when it deviates above (respectively, below) its target value via a negative-feedback loop. 
However, IFCs implementable with electro-mechanical systems
are not necessarily implementable with biochemical reactions~\cite{Me_Homoclinic}. 
Central to this problem is the fact that the error takes both positive and negative values
and, therefore, cannot be directly represented as a nonnegative molecular abundance. 
In this context, a linear non-biochemical IFC has been mapped to 
a bimolecular one in~\cite{Biochemical_IFC}, which has been adapted in~\cite{Khammash} and called
the antithetic integral-feedback controller (AIFC). 

Performance of the AIFC has been largely studied
when unimolecular and/or lower-dimensional 
input networks are controlled~\cite{Khammash,Khammash2,AIFC_1,AIFC_2}; 
in contrast, intracellular networks are generally bimolecular and higher-dimensional
(challenges (N) and (HD) stated above).
For example, authors from~\cite{Khammash} analyze performance of the 
AIFC in context of controlling average copy-numbers of intracellular species at the stochastic level.
In this setting, in~\cite[Theorem~2]{Khammash}, the authors specify a class
of unimolecular input networks that can be controlled with the AIFC; 
in particular, to ensure stability, these input networks have to satisfy an algebraic constraint
given as~\cite[Equation~7]{Khammash}. This technical condition cannot generally be guaranteed to hold
as it not only depends on the rate coefficients of the controller, but also 
on the uncertain rate coefficients of a given input network. More precisely, \emph{affine} input networks, 
i.e. unimolecular input networks that contain one or more basal productions (zero-order reactions),
can violate condition~\cite[Equation~7]{Khammash}. In contrast, \emph{linear} input networks,
i.e. unimolecular networks with no basal production, always satisfy this condition.
To showcase the performance of the AIFC, the authors from~\cite{Khammash} 
put forward a gene-expression system as the input network, 
given by~\cite[Equation~9]{Khammash}, and demonstrate that the AIFC can 
arbitrarily control the average protein copy-number, and mitigate the
uncertainties in the input rate coefficients (challenge (U) stated above). 
However, this unconditional success 
arises because basal transcription is not included in the 
linear gene-expression input network~\cite[Equation~9]{Khammash},
which ensures that condition~\cite[Equation~7]{Khammash} always holds.
A similar choice of an input network without basal production
is put forward in~\cite{Khammash2}, where the AIFC is experimentally implemented.
The AIFC has also been analyzed in context of 
controlling species concentrations  at the deterministic level in~\cite{AIFC_1,AIFC_2};
however, the results are derived only for a restricted class of 
linear input networks that, due to lacking basal production,
unconditionally satisfy~\cite[Theorem~2]{Khammash}.

Questions of critical importance arise in context of controlling unimolecular networks:
When the AIFC is applied on affine input networks (unimolecular networks with basal production),  
how likely is control to fail? Are the consequences of control failures biochemically safe or hazardous~\cite{Burden}? 
Does there exist a molecular IFC with a better stability performance than the AIFC?
Such questions are of great importance when intracellular networks are controlled. 
In particular, for a fixed affine model of an intracellular network,
due to the uncertainties in experimental measurements of the underlying rate coefficients 
(challenge (U)), it is not possible to a-priori guarantee that the stability condition~\cite[Equation~7]{Khammash} holds. 
Furthermore, due to the uncertainties in the structure of intracellular networks, 
only approximate models are available that are obtained by neglecting a number
of underlying coupled molecular species and processes. Hence, even if~\cite[Equation~7]{Khammash}
holds for a less-detailed model of an intracellular network, there is no guarantee that 
this will remain true when a more-detailed model is used - a phenomenon we call 
\emph{phantom control}. In fact, in light of the challenges (N) and (HD), 
more-detailed models of most intracellular networks are nonlinear, 
so that condition~\cite[Equation~7]{Khammash} is inapplicable, and 
a question of fundamental importance to intracellular control is: 
How do molecular IFCs perform when applied to bimolecular and 
higher-dimensional input networks?

The objective of this paper is to address these questions. 
We show that at the center of all these issues are 
equilibria - stationary solutions of the reaction-rate equations (RREs)
that govern the deterministic dynamics of biochemical networks~\cite{Feinberg}. 
In particular, molecular concentrations can reach only equilibria that are nonnegative.
In this context, we show that IFCs can destroy all nonnegative equilibria of
the controlled system and lead to a control failure; furthermore, 
this failure can be catastrophic, as some of the molecular concentrations 
can then experience an unbounded increase with time (blow-up).
We call this hazardous phenomenon, involving absence of nonnegative
equilibria and blow-up of some of the underlying species abundances, 
a \emph{negative-equilibrium catastrophe} (NEC), which we outline in Figure~\ref{fig:cells}. 
To the best of our knowledge, NECs and the related challenges, which are the focus of this paper,
have not been previously analyzed in the literature. For example,
the only form of instability presented in~\cite{Khammash,AIFC_1,AIFC_2}
are bounded deterministic oscillations, which average out at the stochastic level 
and do not correspond to violation of condition~\cite[Equation~7]{Khammash}; 
in contrast, we show that this condition is violated when NECs occur.

The paper is organized as follows. 
In Section~\ref{sec:linear}, we prove that unimolecular IFCs do not exist due to a NEC. 
We then derive a class of bimolecular IFCs given by~(\ref{eq:IFCnetapp}) in Section~\ref{sec:nonlinear};
as a consequence of demanding in the derivation that the controlling variables are positive, 
we obtain IFCs that influence the target species both positively and negatively, 
in contrast to the AIFC that acts only positively. 
In Section~\ref{sec:first_order}, we apply different variants of 
these controllers on a unimolecular gene-expression network~(\ref{eq:input_1}), 
and then generalize the results. In particular, 
we show that the AIFC can lead to a NEC when applied to~(\ref{eq:input_1}),
both deterministically and stochastically; more broadly, we show that the AIFC does 
not generically operate safely when applied to unimolecular networks. 
Furthermore, we prove that there exists a two-dimensional 
(two-species) IFC that eliminates NECs when applied to any (arbitrarily large) 
stable unimolecular input network.
However, in Section~\ref{sec:second_order} we demonstrate that, without detailed information about the input systems,
 NECs generally cannot be prevented when bimolecular networks are controlled. 
In particular, we show that, as opposed to 
dimension-independent control of unimolecular networks, control of bimolecular networks suffers 
from the \emph{curse of dimensionality} - the problem becomes more challenging as the 
dimension of the input network increases. 
We conclude the paper by presenting a summary and discussion in Section~\ref{sec:discussion}.
Notation and background theory are introduced as needed in the paper, 
and are summarized in Appendices~\ref{app:background} and~\ref{app:biochemical_control}.
Rigorous proofs of the results presented in Sections~\ref{sec:linear}, \ref{sec:first_order}
and~\ref{sec:second_order} are provided in Appendices~\ref{app:nonexistence}--\ref{app:kinetictrans}, 
\ref{app:proof} and~\ref{app:biproof}, respectively.

\begin{figure}[!htbp]
\vskip  0cm
\centerline{
\hskip 1.0cm
\includegraphics[scale=0.5]{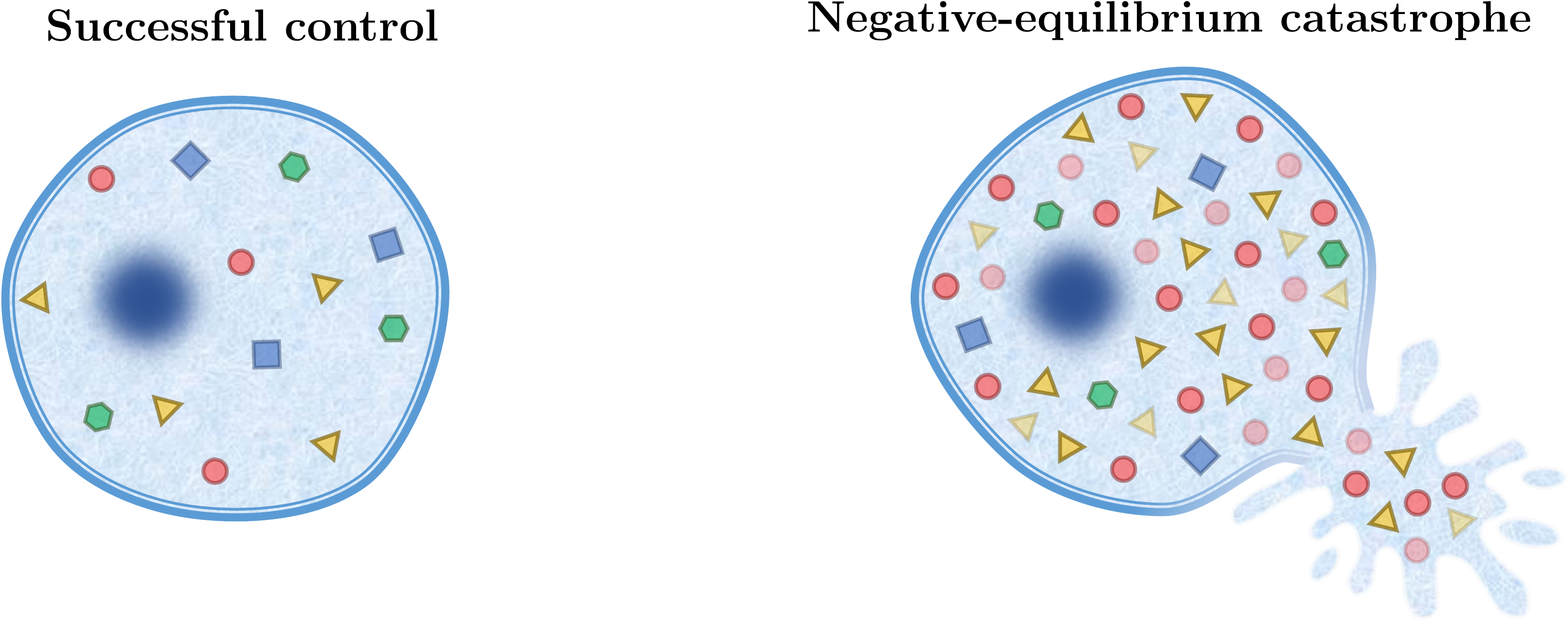}
}
\vskip -5.5cm
\leftline{\hskip -0.1cm (a) \hskip 9.0cm (c)}
\vskip 5.2cm
\centerline{
\includegraphics[width=0.55\columnwidth]{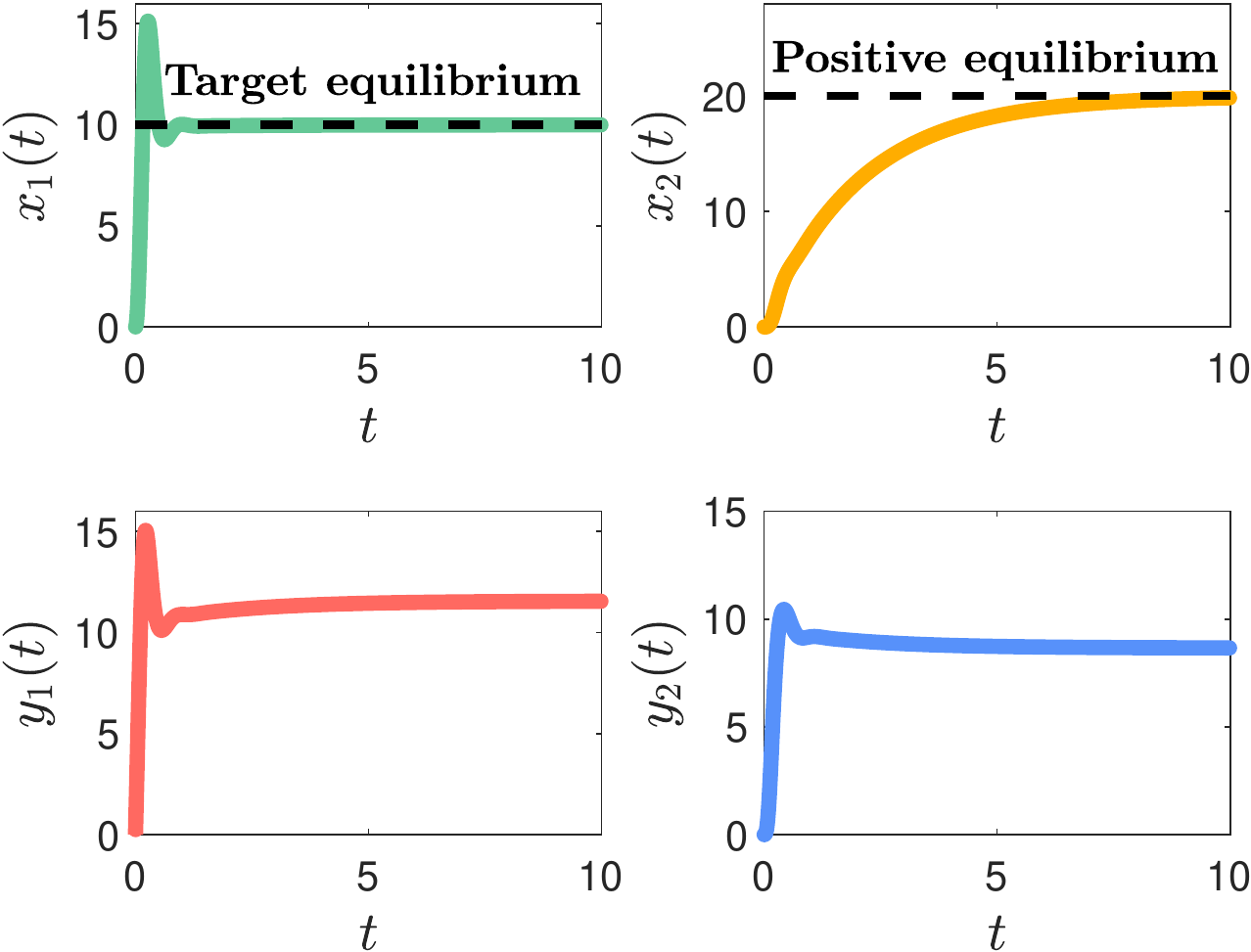}
\hskip 0.3cm
\includegraphics[width=0.55\columnwidth]{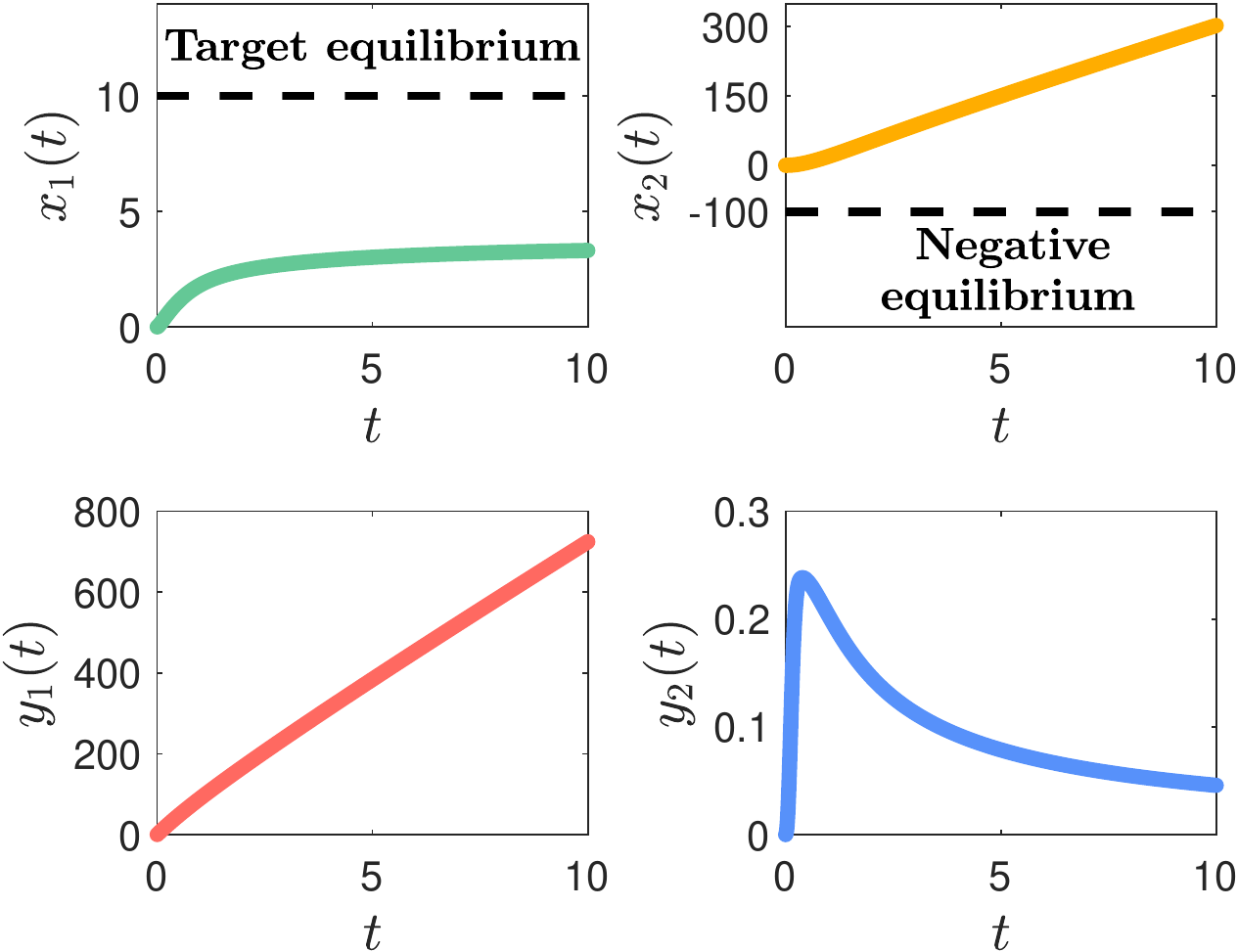}
}
\vskip -7.8cm
\leftline{\hskip -0.1cm (b) \hskip 9.0cm (d)}
\vskip 7.0cm
\caption{\it{\emph{Caricature representation of a successful and catastrophically failed
intracellular control}.
Panel~{\rm (a)} displays a cell successfully controlled with 
the {\rm IFC} in the setup shown in {\rm Figure~\ref{fig:Control_Theory}}. 
The time-evolution of the underlying species concentrations 
are shown in panel~{\rm (b)}. In particular, the target species $X_1$ approaches
a desired equilibrium, shown as a black dashed line, 
and the equilibrium for the residual species $X_2$ is positive. 
Panel~{\rm (c)} displays a cell that has taken lethal damage due to 
a failure of the {\rm IFC}. In particular, as shown in panel~{\rm (d)}, 
the target equilibrium for $X_1$ enforces a negative equilibrium for the residual 
species $X_2$. However, since molecular concentrations
are nonnegative, this equilibrium cannot be reached and, therefore, 
control fails. Furthermore, the failure is catastrophic, as
concentrations of some of the underlying species (in this example, species
$X_2$ and $Y_1$) blow up, placing a lethal burden on the cell.
Panels~{\rm (b)}, and~{\rm (d)}, are obtained by solving the reaction-rate equations
for the output network~{\rm (\ref{eq:IFC_positive_negative})$\cup$(\ref{eq:input_2})}
from {\rm Section~\ref{sec:second_order}} with the dimensionless coefficients 
$(\alpha_0, \alpha_1, \alpha_2, \alpha_3) = (1, 1,1/10, 3/2)$, 
$(\beta_0, \beta_1, \gamma_1, \gamma_2, \gamma_3) = (100, 1, 10, 10, 1)$, and with 
$(\alpha_0, \alpha_1, \alpha_2, \alpha_3) = (1, 25,2/5,3/2)$, 
$(\beta_0, \beta_1, \gamma_1, \gamma_2, \gamma_3) = (100, 1, 10,2/3, 1)$, 
respectively.}}  \label{fig:cells}
\end{figure}  

\newpage
\section{Nonexistence of unimolecular IFCs} \label{sec:linear}
In this section, we consider an arbitrary one-dimensional black-box input network 
$\mathcal{R}_{\alpha} = \mathcal{R}_{\alpha}(X_1)$, 
where $X_1$ is a single target species and there are no residual species, 
see also Figure~\ref{fig:Control_Theory} for a more general setup. 
In this paper, we assume all reaction networks are under
mass-action kinetics~\cite{Feinberg} with 
positive dimensionless rate coefficients, which
are displayed above or below the reaction arrows; we denote
the rate coefficients of $\mathcal{R}_{\alpha}$ by 
$\boldsymbol{\alpha} \in \mathbb{R}_{>}^a$, where $\mathbb{R}_{>}$
is the space of positive real numbers. 
In what follows, we say that a reaction network is unimolecular
(respectively, bimolecular) if it contains a reaction with 
one (respectively, two), but not more, reactants.

\textbf{Linear non-biochemical controller}. Let us consider a controller formally described by the network
$\mathcal{\bar{R}}_{\beta, \gamma} = \mathcal{\bar{R}}_{\beta}(\bar{Y}_1) \cup 
\mathcal{\bar{R}}_{\gamma}(X_1, \bar{Y}_1)$, given by
\begin{align}
\mathcal{\bar{R}}_{\beta}: \hspace{0.5cm} 
\varnothing & \xrightarrow[]{\beta_{0}} \bar{Y}_1, \nonumber \\
\mathcal{\bar{R}}_{\gamma}: \hspace{0.5cm} 
X_1 & \xrightarrow[]{\gamma_1} X_1 - \bar{Y}_1, \nonumber \\
\bar{Y}_1 & \xrightarrow[]{\gamma_2} \bar{Y}_1 + X_1.
\label{eq:controller_bar}
\end{align}
Here, $\mathcal{\bar{R}}_{\beta} = \mathcal{\bar{R}}_{\beta}(\bar{Y}_1)$ is the controller core, 
describing the internal dynamics of the controlling species $\bar{Y}_1$,
where the source $\varnothing$ denotes some species that are not explicitly modelled,
while $\mathcal{\bar{R}}_{\gamma} = \mathcal{\bar{R}}_{\gamma}(X_1, \bar{Y}_1)$
is the controller interface, specifying interactions between $\bar{Y}_1$ and the target
species $X_1$ from the input network, see also Figure~\ref{fig:Control_Theory}.
Let us denote abundances of species $\{X_1, \bar{Y}_1\}$ from the output network 
$\mathcal{R}_{\alpha} \cup \mathcal{\bar{R}}_{\beta, \gamma}$ at time
 $t \ge 0$ by $(x_1, \bar{y}_1) = (x_1(t), \bar{y}_1(t)) \in \mathbb{R}^2$.
At the deterministic level, formal reaction-rate equations (RREs)~\cite{Feinberg} read
\begin{align}
\frac{\mathrm{d} x_1}{\mathrm{d} t} & = f_1(x_1; \, \boldsymbol{\alpha}) + \gamma_2 \bar{y}_1, 
\hspace{1.0cm} x_1^* = \frac{\beta_0}{\gamma_1}, \nonumber \\
\frac{\mathrm{d} \bar{y}_1}{\mathrm{d} t} & = \beta_0 - \gamma_{1} x_1,
\hspace{2.2cm}
\bar{y}_1^*= - \gamma_2^{-1} f_1 \left(\frac{\beta_0}{\gamma_1}; \, \boldsymbol{\alpha}\right),
 \label{eq:RREs_bar}
\end{align}
where $f_1(x_1; \, \boldsymbol{\alpha})$ is an unknown function describing the 
dynamics of $\mathcal{R}_{\alpha}$, and
 $(x_1^*, \bar{y}_1^*) \in \mathbb{R}^2$ is the unique equilibrium of
the output network, obtained by solving the RREs with zero left-hand sides.
Assuming that $(x_1^*, \bar{y}_1^*)$ is globally stable, 
network~(\ref{eq:controller_bar}) is an IFC;
in particular, in this case, $x_1^* = (\beta_0/\gamma_1)$ is independent 
of the initial conditions and the input coefficients $\boldsymbol{\alpha}$.
However, controller~(\ref{eq:RREs_bar}) cannot be interpreted
as a biochemical reaction network. In particular, the term $(-\gamma_{1} x_1)$
in~(\ref{eq:RREs_bar}) induces a process graphically described by 
$X_1 \xrightarrow[]{\gamma_1} X_1 - \bar{Y}_1$ in~(\ref{eq:controller_bar}), 
which consumes species $\bar{Y}_1$ even when its abundance is zero.
Consequently, variables $(x_1, \bar{y}_1)$ may take negative values
and, therefore, cannot be interpreted as molecular concentrations~\cite{Me_Homoclinic}.

\textbf{Unimolecular controllers}. The only unimolecular analogue of the 
IFC~(\ref{eq:controller_bar}), that contains only one controlling species $Y_1$, is of the form
\begin{align}
\mathcal{R}_{\beta}: \hspace{0.5cm} 
\varnothing & \xrightarrow[]{\beta_{0}} Y_1, \nonumber \\
\mathcal{R}_{\gamma}: \hspace{0.5cm} 
X_1 & \xrightarrow[]{\gamma_1} X_1 + Y_1, \nonumber \\
Y_1 & \xrightarrow[]{\gamma_2} Y_1 + X_1.
\label{eq:controller_blowup}
\end{align}
The RREs and the equilibrium for the output network 
$\mathcal{R}_{\alpha} \cup \mathcal{R}_{\beta, \gamma}$ are given by 
\begin{align}
\frac{\mathrm{d} x_1}{\mathrm{d} t} & = f_1(x_1; \, \boldsymbol{\alpha}) + \gamma_2 y_1, 
\hspace{1.0cm} x_1^* = -\frac{\beta_0}{\gamma_1}, 
\nonumber \\
\frac{\mathrm{d} y_1}{\mathrm{d} t} & = \beta_0 + \gamma_{1} x_1,
\hspace{2.2cm}
y_1^*=  - \gamma_2^{-1} f_1 \left(-\frac{\beta_0}{\gamma_1}; \, \boldsymbol{\alpha}\right).
 \label{eq:RREs_blowup}
\end{align}
Given nonnegative initial conditions, variables $(x_1, y_1)$ from~(\ref{eq:RREs_blowup}) are confined 
to the nonnegative quadrant $\mathbb{R}_{\ge}^2$, and represent biochemical concentrations.
However, the $x_1$-component of the equilibrium from~(\ref{eq:RREs_blowup}) is negative
and, therefore, not reachable by the controlled system. Furthermore, 
$y_1$ is a monotonically increasing 
function of time, $\mathrm{d} y_1/\mathrm{d} t > 0$, 
i.e. $y_1$ blows up. We call this phenomenon 
a deterministic \emph{negative-equilibrium catastrophe} (NEC),
see also Appendix~\ref{app:background}.
Network~(\ref{eq:controller_blowup})
not only fails to achieve control, but it introduces an unstable
species and is, hence, biochemically hazardous. 
In Appendix~\ref{app:nonexistence}, we prove that a NEC
occurs at both deterministic and stochastic levels 
for any candidate unimolecular IFC, which we state as the following theorem.

\begin{theorem}\label{theorem:linear_nonexist}
There does not exist a unimolecular integral-feedback controller.
\end{theorem}

\begin{proof}
See Appendix~\ref{app:nonexistence}.
\end{proof}

To the best of our knowledge, Theorem~\ref{theorem:linear_nonexist}
has not been previously reported in the literature. 
A related result is presented in~\cite[Proposition S2.7]{Khammash2}
and states that a molecular controller $\mathcal{R}_{\beta} \cup \mathcal{R}_{\gamma}$, 
satisfying a set of assumptions, including the assumption that the interface 
$\mathcal{R}_{\gamma}$ contains only catalytic reactions,
is a molecular IFC only if the core $\mathcal{R}_{\beta}$ contains a 
bimolecular degradation. No such assumptions
have been made in Theorem~\ref{theorem:linear_nonexist}, 
which holds for all unimolecular networks; 
in particular, we allow interface $\mathcal{R}_{\gamma}$ 
to contain non-catalytic reactions, such as $Y_i \to X_j$
and $X_i \to Y_j$.

\section{Design of bimolecular IFCs} \label{sec:nonlinear}
Theorem~\ref{theorem:linear_nonexist} implies that only 
bimolecular (and higher-molecular) biochemical networks 
may exert integral-feedback control. An approach to finding such 
networks is to map non-biochemical IFCs into biochemical networks,
while preserving the underlying integral-feedback structure.
This task can be achieved using special mappings
called \emph{kinetic transformations}~\cite{Me_Homoclinic}.
Let us consider the non-biochemical system~(\ref{eq:RREs_bar}).
The first step in bio-transforming~(\ref{eq:RREs_bar}) is to translate
relevant trajectories $(x_1, \bar{y}_1)$ into the nonnegative quadrant. However, 
since $\mathcal{R}_{\alpha}(X_1)$ is a black-box network,
i.e. $f_1(x_1; \, \boldsymbol{\alpha})$ is unknown and unalterable, 
only $\bar{y}_1$ can be translated; to this end, 
we define a new variable $y_1 \equiv (\bar{y}_1 + T)$, 
with translation $T > 0$, under which~(\ref{eq:RREs_bar}) becomes
\begin{align}
\frac{\mathrm{d} x_1}{\mathrm{d} t} & = f_1(x_1; \, \boldsymbol{\alpha}) + \gamma_2 y_1 - \gamma_2 T, 
\hspace{1.0cm} x_1^* = \frac{\beta_0}{\gamma_1}, \nonumber \\
\frac{\mathrm{d} y_1}{\mathrm{d} t} & = \beta_0 - \gamma_{1} x_1, 
\hspace{3.3cm}
y_1^*= - \gamma_2^{-1} f_1 \left(\frac{\beta_0}{\gamma_1}; \, \boldsymbol{\alpha}\right) + T.
 \label{eq:RREs_debar1}
\end{align} 
Terms $(-\gamma_1 x_1)$ and $(-\gamma_2 T)$, called \emph{cross-negative} terms~\cite{Me_Homoclinic}, 
do not correspond to biochemical reactions and, therefore, must be eliminated. 
Let us note that cross-negative terms also play a central role
in the questions of existence of other fundamental phenomena in biochemistry, such as
oscillations, multistability and chaos~\cite{Me_Homoclinic,Me_Limitcycles}.
Term $(-\gamma_1 x_1)$ can be eliminated with the so-called hyperbolic kinetic transformation,
presented in Appendix~\ref{app:kinetictrans}, which involves introducing an additional 
controlling species $Y_2$ and extending system~(\ref{eq:RREs_debar1}) into
\begin{align}
\frac{\mathrm{d} x_1}{\mathrm{d} t} & = f_1(x_1; \, \boldsymbol{\alpha}) + \gamma_2 y_1 - \gamma_2 T,
\hspace{1.0cm} x_1^* = \frac{\beta_0}{\gamma_1}, \nonumber \\
\frac{\mathrm{d} y_1}{\mathrm{d} t} & = \beta_0 - \beta_{1} y_1 y_2, 
\hspace{2.9cm} y_1^*= - \gamma_2^{-1} f_1 \left(\frac{\beta_0}{\gamma_1}; \, \boldsymbol{\alpha}\right) + T, \nonumber \\
\frac{\mathrm{d} y_2}{\mathrm{d} t} & =  \gamma_{1} x_1 - \beta_{1} y_1 y_2, 
\hspace{2.5cm} y_2^*= \frac{\beta_0}{\beta_1}(y_1^*)^{-1}.
 \label{eq:RREs_debar2}
\end{align}
Note that~(\ref{eq:RREs_debar1}) and~(\ref{eq:RREs_debar2}) have identical
equilibria (time-independent solutions) for the species $X_1$ and $Y_1$, and
that the equilibria for the species $Y_1$ and $Y_2$ have a hyperbolic relationship; 
furthermore, provided $\beta_1$ is sufficiently large, time-dependent solutions of~(\ref{eq:RREs_debar1}) 
and~(\ref{eq:RREs_debar2}) are close as well, see Appendix~\ref{app:kinetictrans}.
On the other hand, cross-negative term $(-\gamma_2 T)$ can be eliminated via
multiplication with $x_1$ and any other desired factor; such operations
do not influence the $x_1$-equilibrium, which is determined solely 
by the RREs for $y_1$ and $y_2$, and which we want to preserve.
One option is to simply map $(-\gamma_2 T)$ to $(-\gamma_2 T x_1)$,
and take $T$ large enough to ensure that the $y_1^*$-equilibrium is positive;
however, this approach requires the knowledge of $f_1(x_1; \, \boldsymbol{\alpha})$.
A more robust approach is to map $(-\gamma_2 T)$ to $(-\gamma_2 T x_1 y_2)$, 
under which, defining $\gamma_3 \equiv \gamma_2 T$, one obtains 
\begin{align}
\frac{\mathrm{d} x_1}{\mathrm{d} t} & = f_1(x_1; \, \boldsymbol{\alpha}) + \gamma_2 y_1 - \gamma_3 x_1 y_2, 
\hspace{1.0cm} x_1^* = \frac{\beta_0}{\gamma_1}, \nonumber \\
\frac{\mathrm{d} y_1}{\mathrm{d} t} & = \beta_0 - \beta_{1} y_1 y_2, 
\hspace{3.5cm} 0 = (y_1^*)^2 + \left[\gamma_2^{-1} 
f_1 \left(\frac{\beta_0}{\gamma_1}; \, \boldsymbol{\alpha}\right) \right] y_1^*
- \left(\frac{\gamma_3}{\gamma_1 \gamma_2} \frac{\beta_0^2}{\beta_1} \right),
\nonumber \\
\frac{\mathrm{d} y_2}{\mathrm{d} t} & =  \gamma_{1} x_1 - \beta_{1} y_1 y_2, 
\hspace{3.0cm} y_2^*= \frac{\beta_0}{\beta_1}(y_1^*)^{-1}.
 \label{eq:RREs_debar3}
\end{align}
The quadratic equation for $y_1^*$ from~(\ref{eq:RREs_debar3}) always has one positive solution;
therefore, there always exists an equilibrium with positive $y_1$- and $y_2$-components.

In what follows, we largely consider input networks $\mathcal{R}_{\alpha}(\mathcal{X})$ with 
at most two target species $\{X_1, X_2\}$, and focus on controlling $X_1$ with the bimolecular 
controllers induced by~(\ref{eq:RREs_debar3}), given by
\begin{align}
\mathcal{R}_{\beta}(Y_1, Y_2): \;
& & \varnothing & \xrightarrow[]{\beta_0} Y_1, \nonumber \\
& & Y_1 + Y_2 & \xrightarrow[]{\beta_1} \varnothing, \nonumber \\
\mathcal{R}_{\gamma}^{0}(Y_2; \, X_1): \;
& & X_1 & \xrightarrow[]{\gamma_{1}} X_1 + Y_2, \nonumber \\
\mathcal{R}_{\gamma}^{+}(X_i; \, Y_1): \;
& & Y_1 & \xrightarrow[]{\gamma_{2}} X_i + Y_1, \; \; \; \; \; \; \textrm{for some } i \in \{1, 2\}, \nonumber \\
\mathcal{R}_{\gamma}^{-}(X_j; \, Y_2): \;
& & X_j + Y_2 & \xrightarrow[]{\gamma_{3}} Y_2, \hspace{1.5cm} \textrm{for some } j \in \{1, 2\}.
\label{eq:IFCnetapp}
\end{align}
In particular, the controller core $\mathcal{R}_{\beta}(Y_1, Y_2)$ consists
of a production of $Y_1$ from a source, and a bimolecular degradation of $Y_1$ and $Y_2$. 
On the other hand, the controller interface consists of the unimolecular reactions
 $\mathcal{R}_{\gamma}^0(Y_2; \, X_1)$, and $\mathcal{R}_{\gamma}^{+}(X_i; \, Y_1)$, that 
produce $Y_2$ catalytically in $X_1$, and $X_i$ catalytically in $Y_1$, respectively, and 
the bimolecular reaction  $\mathcal{R}_{\gamma}^{-}(X_j; \, Y_2)$
that degrades a target species $X_j$ catalytically in $Y_2$. 
We call reactions $\mathcal{R}_{\gamma}^{+}(X_i; \, Y_1)$
and $\mathcal{R}_{\gamma}^{-}(X_j, Y_2)$ \emph{positive}
and \emph{negative} interfacing, respectively. Furthermore, we say that 
positive (respectively, negative) interfacing is \emph{direct} 
if $i = 1$ (respectively, if $j = 1$), i.e. if it is applied directly to the controlled
species $X_1$; otherwise, the interfacing is said to be \emph{indirect}. 
In Figure~\ref{fig:Control_Theory}, we display controller~(\ref{eq:IFCnetapp}) with direct positive and 
negative interfacing applied to an input network with a single target species $X_1$.

As shown in this section, positive and negative interfacing arise naturally when molecular IFCs 
are designed using the theoretical framework from~\cite{Me_Homoclinic}.
It is interesting to note that the ``housekeeping" sigma/anti-sigma system in \emph{E. coli}, 
proposed to implement integral control~\cite{Khammash}, has been experimentally shown to be capable 
of exhibiting both positive and negative transcriptional control, 
at least when hijacked by bacteriophage~\cite{Sigma}.
Let us note that the AIFC from~\cite{Khammash} is of the form~(\ref{eq:IFCnetapp}), 
but it lacks negative interfacing $\mathcal{R}_{\gamma}^{-}(X_j; \, Y_2)$.
In view of the derivation from this section, the AIFC is missing a key designing step, 
namely the translation from~(\ref{eq:RREs_debar1});
consequently, NECs may occur due to $y_1^*$- and $y_2^*$-equilibria being negative.
Let us also note that the negative interfacing $\mathcal{R}_{\gamma}^{-}(X_j; \, Y_2)$
has also been considered in~\cite{Khammash3}, where this reaction is shown to be capable of eliminating
oscillations at the deterministic level, and reducing variance at the stochastic level, 
for a particular gene-expression input network. In contrast, in this section, 
we have systematically derived reaction $\mathcal{R}_{\gamma}^{-}(X_j; \, Y_2)$ in order to ensure that a positive
equilibrium for $Y_1$ and $Y_2$ exists. Such matters are not discussed
 in~\cite{Khammash3}, where basal transcription is set to zero in the gene-expression input network considered
and, therefore, negative equilibria are not encountered.

\section{Control of unimolecular input networks} \label{sec:first_order}
In this section, we study performance of the IFCs~(\ref{eq:IFCnetapp})
when applied on unimolecular input networks. To this end, let us consider the 
input network $\mathcal{R}_{\alpha}^1 = \mathcal{R}_{\alpha}^1(X_1, X_2)$, given by
\begin{align}
\mathcal{R}_{\alpha}^1(X_1, X_2):
& & \varnothing & \xrightleftharpoons[\alpha_{1}]{\alpha_{0}} X_2, \hspace{0.3cm}
X_2 \xrightarrow[]{\alpha_{2}} X_1 + X_2, \hspace{0.3cm}
X_1 \xrightarrow[]{\alpha_{3}} \varnothing. 
\label{eq:input_1}
\end{align}
We interpret~(\ref{eq:input_1}) as a two-dimensional reduced (simplified) model of
a higher-dimensional gene-expression network.
In this context, $X_1$ is a degradable protein species
that is produced via translation from 
a degradable mRNA species $X_2$, which is transcribed from a gene; 
some of the ``hidden" species (dimensions), that are not explicitly modelled, 
such as genes, transcription factors and waste molecules,
 are denoted by $\varnothing$. See also Figure~\ref{fig:genetic}(a)
for a schematic representation of network~(\ref{eq:input_1}).
The RREs of~(\ref{eq:input_1}) have 
a unique globally stable equilibrium given by
\begin{align}
x_1^{**} & = \frac{\alpha_0 \alpha_2}{\alpha_1 \alpha_3}, 
\hspace{0.5cm}
x_2^{**} = \frac{\alpha_0}{\alpha_1}. \label{eq:input_eq}
\end{align}
The goal in this section is to control the equilibrium concentration 
of the protein species $X_1$ at the deterministic level, and its average
copy-number at the stochastic level. To this end, we embed different variants of the 
controller~(\ref{eq:IFCnetapp}) into~(\ref{eq:input_1}).

\begin{figure}[!htbp]
\vskip  0.4cm
\centerline{
\includegraphics[width=0.47\columnwidth]{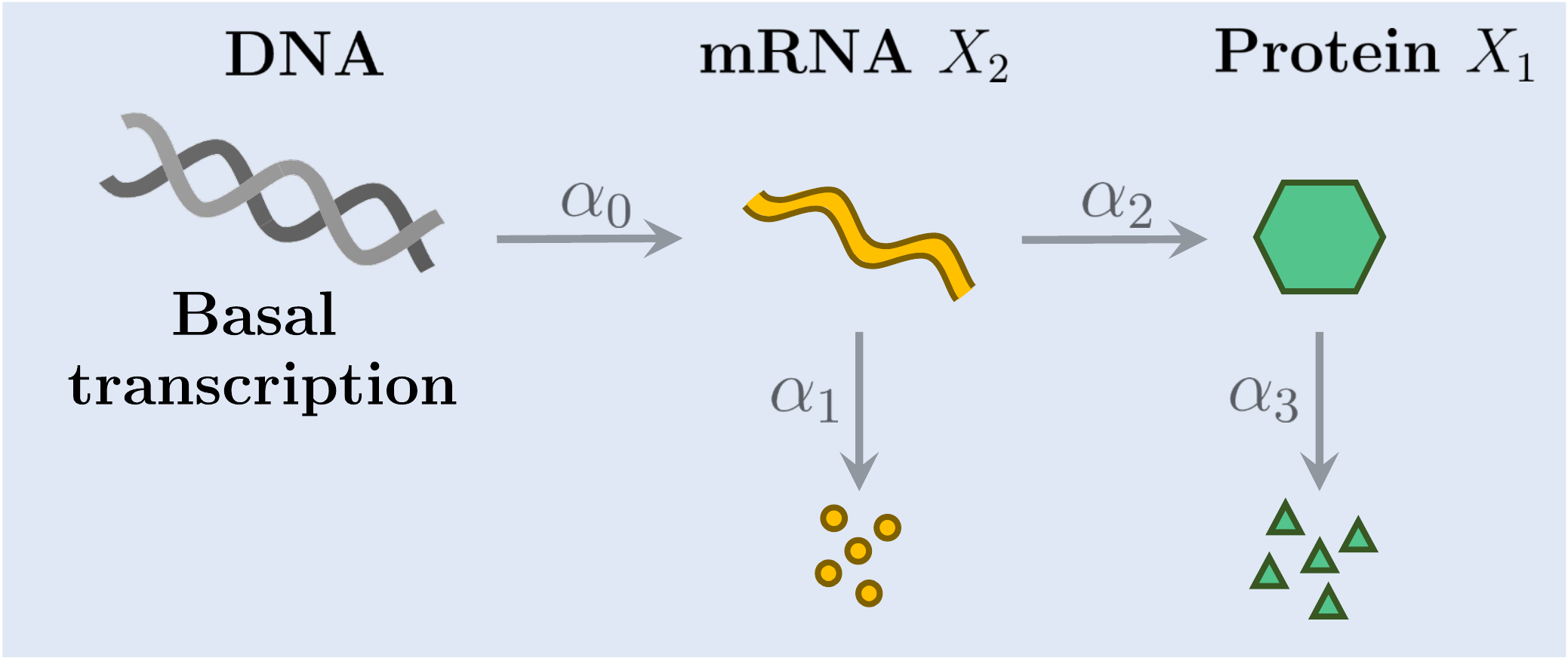}
\hskip 1.3cm
\includegraphics[width=0.47
\columnwidth]{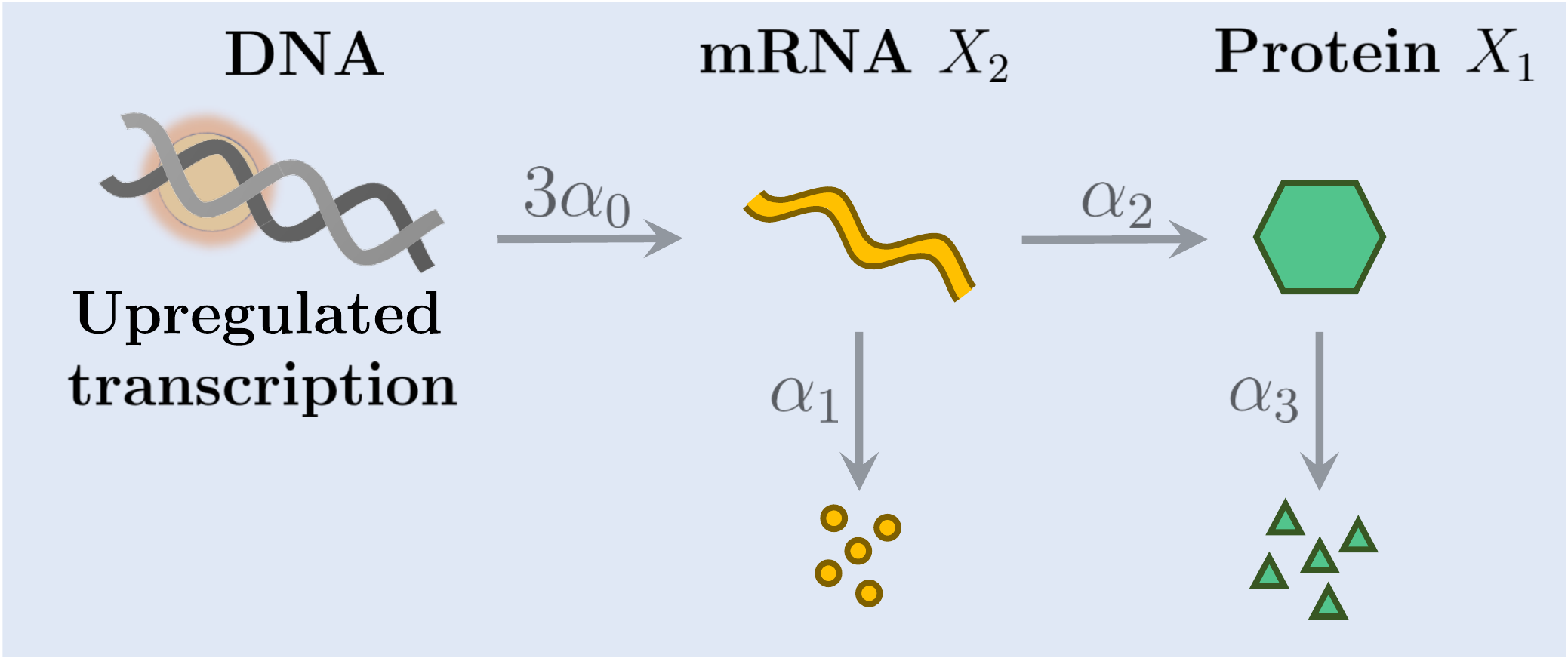}
}
\vskip -3.9cm
\leftline{\hskip -0.15cm (a) \hskip 8.55cm (b)}
\vskip 3.2cm
\caption{\it{\emph{Schematic representation of the gene-expression input network~{\rm(\ref{eq:input_1})}}.
Panel~{\rm (a)} displays~{\rm(\ref{eq:input_1})} with basal transcription rate
$\alpha_0$. Panel~{\rm (b)} displays network~{\rm(\ref{eq:input_1})} with tripled 
effective transcription rate, $3 \alpha_0$, arising when an activating transcription factor binds
to the underlying gene promoter.}}  \label{fig:genetic}
\end{figure}  

\textbf{Pure positive interfacing}. Let us first consider controller~(\ref{eq:IFCnetapp}) 
with only positive interfacing, i.e. the AIFC from~\cite{Khammash}.
We denote the controller by $\mathcal{R}_{\beta, \gamma}^{+} \equiv 
\mathcal{R}_{\beta} \cup \mathcal{R}_{\gamma}^0 \cup \mathcal{R}_{\gamma}^{+}$
and, for simplicity, assume that interfacing is direct:
\begin{align}
\mathcal{R}_{\beta}(Y_1, Y_2): \;
& & \varnothing & \xrightarrow[]{\beta_0} Y_1, \nonumber \\
& & Y_1 + Y_2 & \xrightarrow[]{\beta_1} \varnothing, \nonumber \\
\mathcal{R}_{\gamma}^{0}(Y_2; \, X_1): \;
& & X_1 & \xrightarrow[]{\gamma_{1}} X_1 + Y_2, \nonumber \\
\mathcal{R}_{\gamma}^{+}(X_1; \, Y_1): \;
& & Y_1 & \xrightarrow[]{\gamma_{2}} X_1 + Y_1.
\label{eq:IFC_positive}
\end{align}
The RREs for the output network~(\ref{eq:input_1})$\cup$(\ref{eq:IFC_positive})
 are given by 
\begin{align}
\frac{\mathrm{d} x_1}{\mathrm{d} t} & = \left(\alpha_2 x_2 - \alpha_3 x_1 \right) + \gamma_2 y_1, 
\hspace{0.8cm}
\frac{\mathrm{d} x_2}{\mathrm{d} t} = \alpha_0 - \alpha_1 x_2, \nonumber \\
\frac{\mathrm{d} y_1}{\mathrm{d} t} & = \beta_0 - \beta_{1} y_1 y_2, 
\hspace{2.4cm}
\frac{\mathrm{d} y_2}{\mathrm{d} t} =  \gamma_{1} x_1 - \beta_{1} y_1 y_2,
 \label{eq:RREs_plus}
\end{align}
with the unique equilibrium
\begin{align}
x_1^* & = \frac{\beta_0}{\gamma_1}, 
\hspace{0.5cm}
x_2^* = \frac{\alpha_0}{\alpha_1}, 
\hspace{0.5cm}
y_1^* = \frac{\alpha_3}{\gamma_2} \left( \frac{\beta_0}{\gamma_1} 
-  \frac{\alpha_0 \alpha_2}{\alpha_1 \alpha_3}\right), 
\hspace{0.5cm}
y_2^* = \frac{\beta_0}{\beta_1} (y_1^*)^{-1}.
\label{eq:output_eq_1}
\end{align}
As anticipated in Section~\ref{sec:nonlinear}, the AIFC can lead to
equilibria with negative $y_1$- and $y_2$-components. In particular,
equation~(\ref{eq:output_eq_1}) implies that the output nonnegative equilibrium
is destroyed when $x_1^* < x_1^{**}$ (equivalently, 
when $\beta_0/\gamma_1 < \alpha_0 \alpha_2/(\alpha_1 \alpha_3)$).
Hence, using only positive interfacing, it is not possible to achieve 
an output equilibrium below the input one.
To determine the dynamical behavior of~(\ref{eq:input_1})$\cup$(\ref{eq:IFC_positive})
when the nonnegative equilibrium ceases to exist, let us consider the linear combination 
of species concentration $(\alpha_3^{-1} x_1 + \alpha_1^{-1} \alpha_2 \alpha_3^{-1} x_2 
+ \gamma_1^{-1} (y_2 - y_1))$ that, using~(\ref{eq:RREs_plus}), satisfies
\begin{align}
\frac{\mathrm{d}}{\mathrm{d} t} \left(\frac{1}{\alpha_3} x_1 + \frac{\alpha_2}{\alpha_1 \alpha_3} x_2 
+ \frac{1}{\gamma_1} (y_2 - y_1) \right) & = - \left(  \frac{\beta_0}{\gamma_1} 
-  \frac{\alpha_0 \alpha_2}{\alpha_1 \alpha_3}\right) + \frac{\gamma_2}{\alpha_3} y_1
\ge - \left(  \frac{\beta_0}{\gamma_1} 
-  \frac{\alpha_0 \alpha_2}{\alpha_1 \alpha_3}\right). 
 \label{eq:blowup_plus}
\end{align}
When $\beta_0/\gamma_1 < \alpha_0 \alpha_2/(\alpha_1 \alpha_3)$, 
equation~(\ref{eq:blowup_plus}) implies that a species concentration 
blows up for all nonnegative initial conditions, i.e. the output network 
displays a deterministic NEC; by applying identical argument to the first-moment equations, 
it follows that a stochastic NEC occurs as well. This result 
is summarized as a bifurcation diagram in Figure~\ref{fig:linear}(a).

\begin{figure}[!htbp]
\vskip  -2.5cm
\centerline{
\hskip 1mm
\includegraphics[width=0.4\columnwidth]{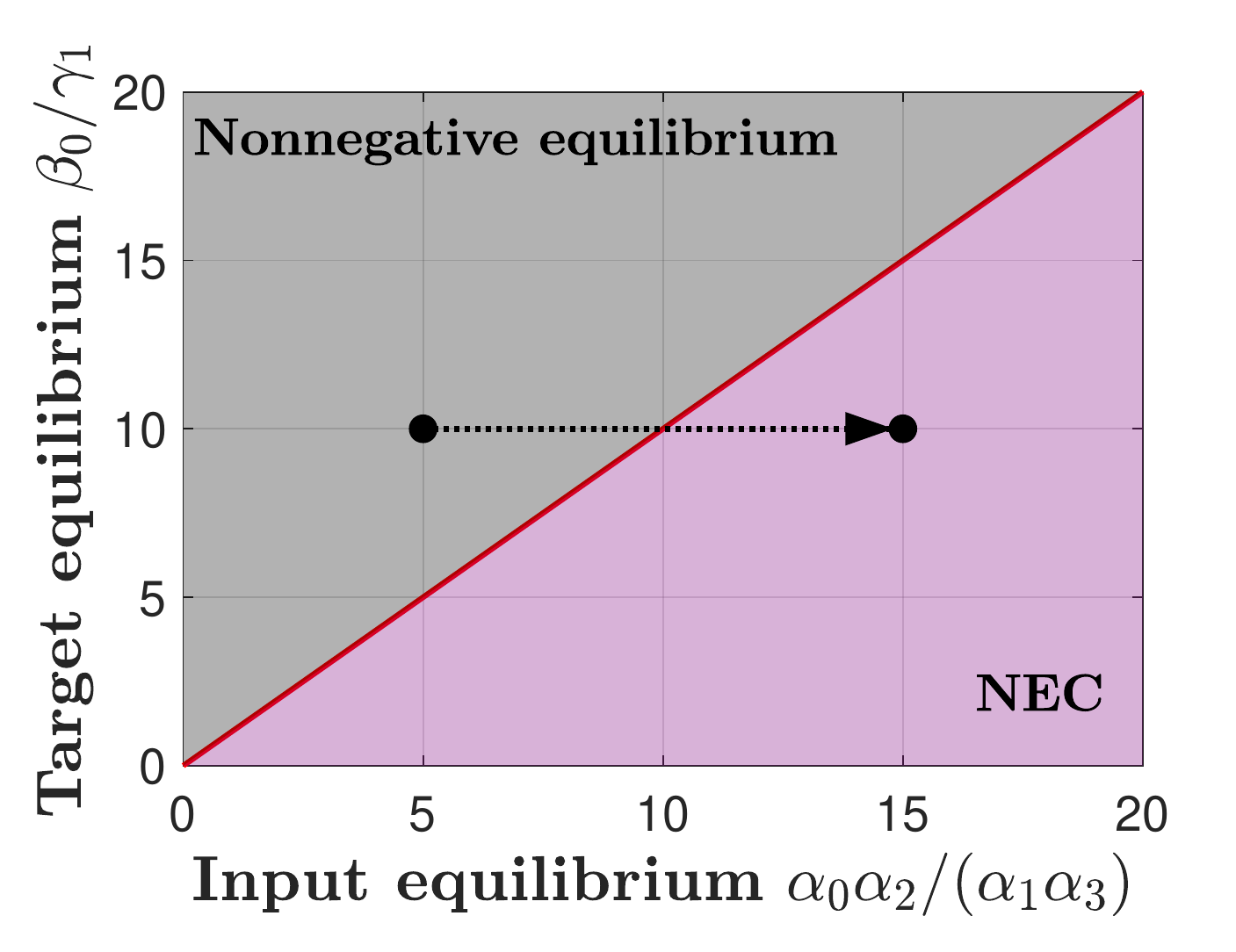}
\hskip -0.3cm
\includegraphics[width=0.4\columnwidth]{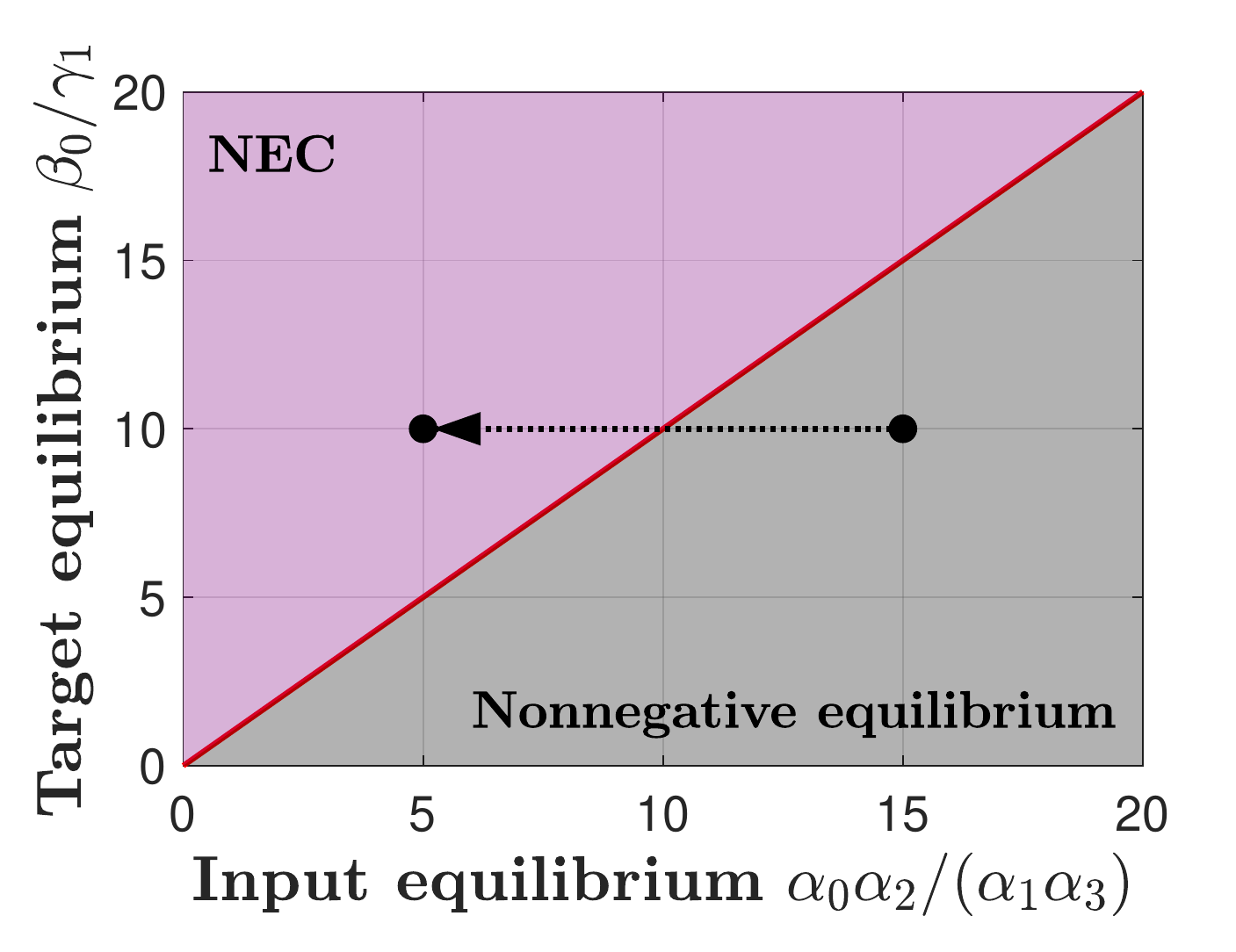}
\hskip -0.3cm
\includegraphics[width=0.4\columnwidth]{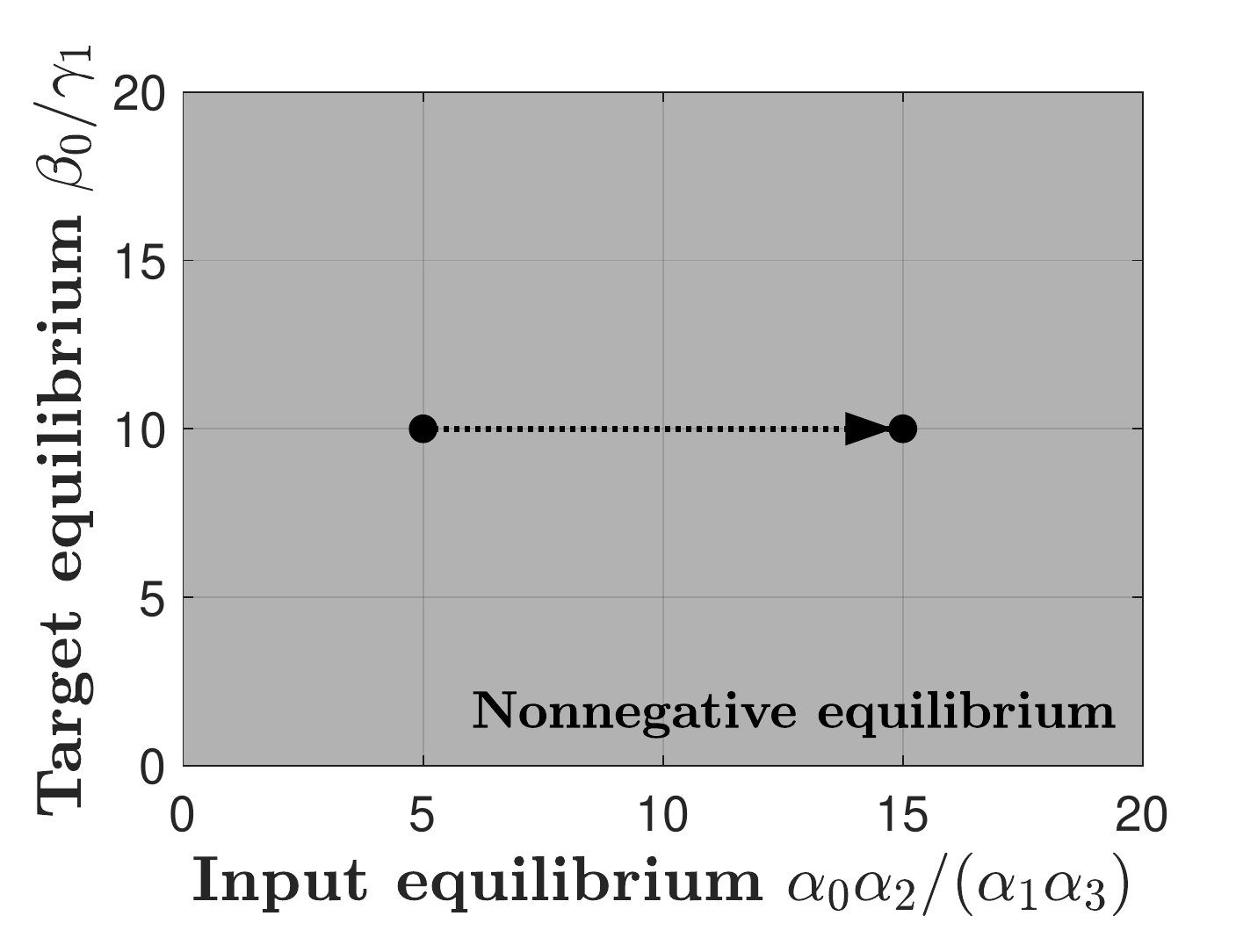}
}
\vskip -5.4cm
\leftline{\hskip -0.8cm (a) \hskip 5.9cm (e) \hskip 5.8cm (i)}
\vskip 4.9cm
\centerline{
\hskip 1mm
\includegraphics[width=0.4\columnwidth]{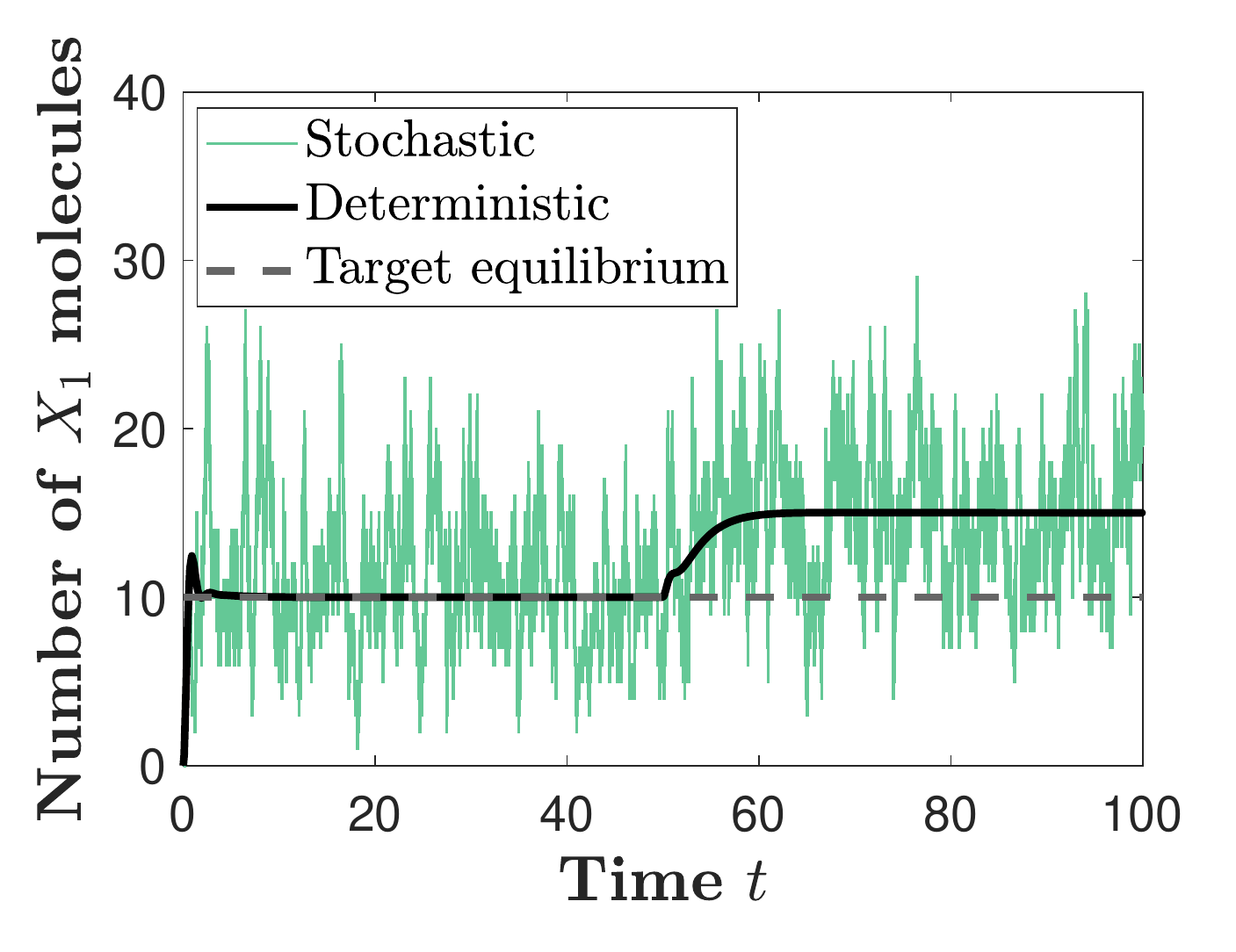}
\hskip -0.3cm
\includegraphics[width=0.4\columnwidth]{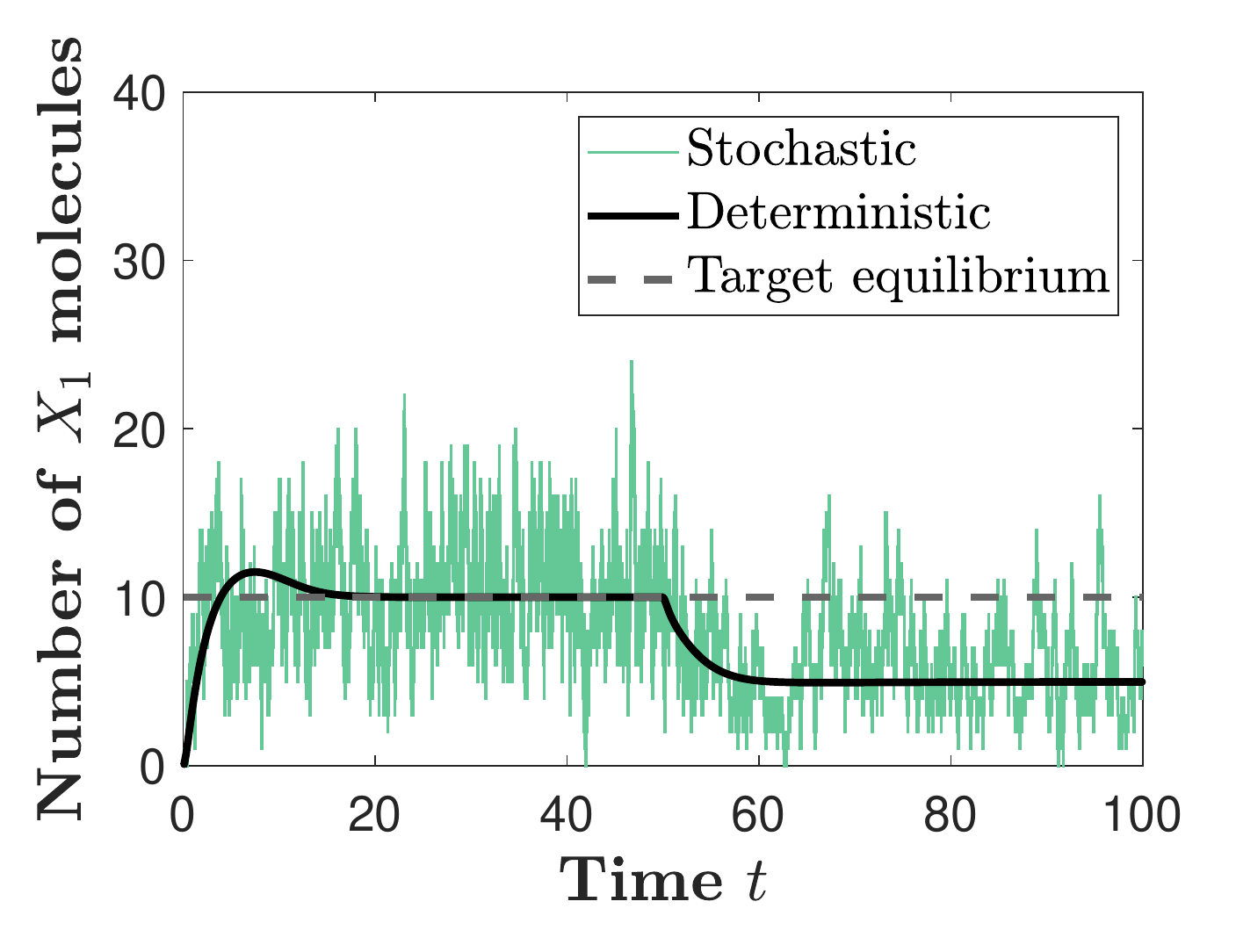}
\hskip -0.3cm
\includegraphics[width=0.4\columnwidth]{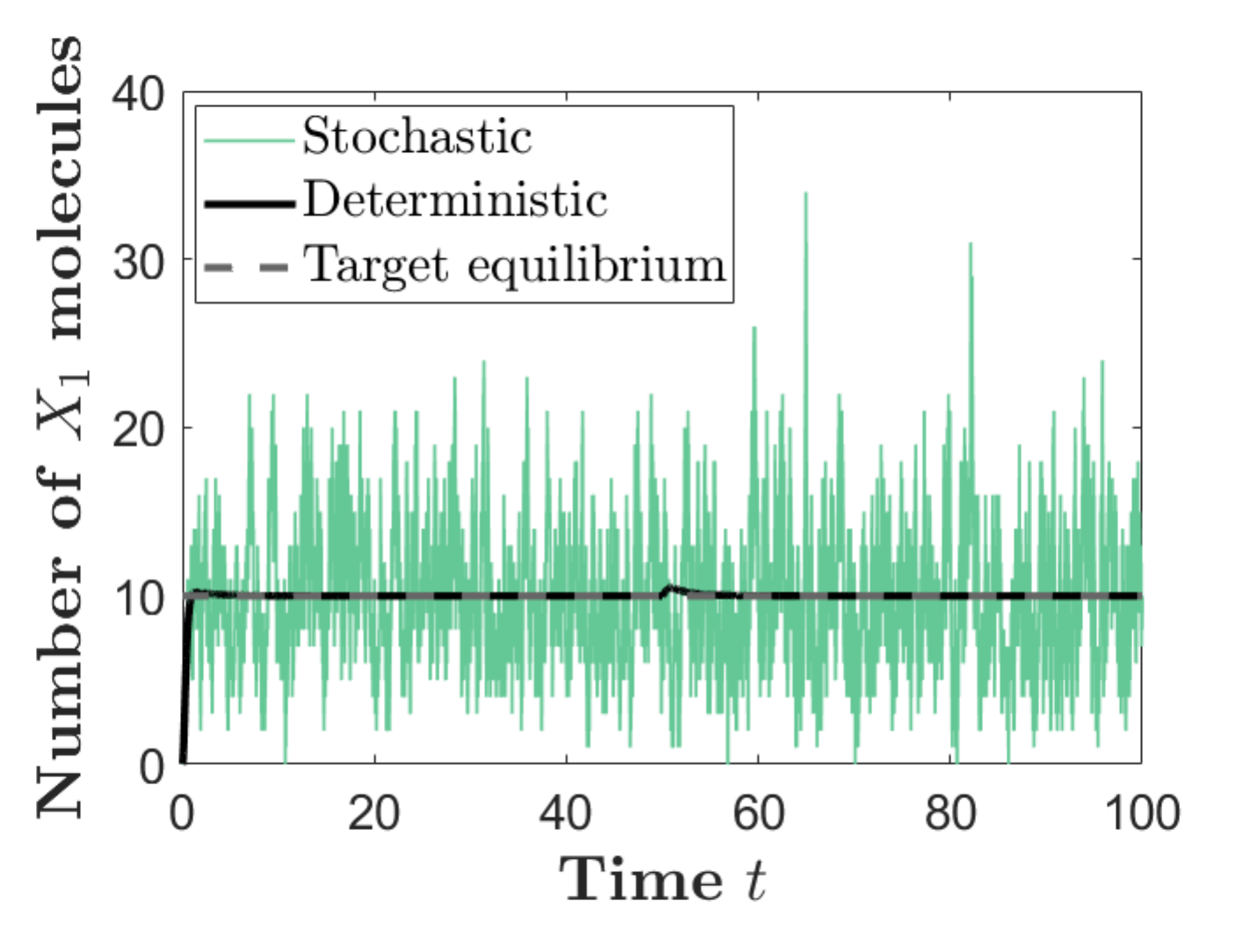}
}
\vskip -5.4cm
\leftline{\hskip -0.8cm (b) \hskip 5.9cm (f) \hskip 5.8cm (j)}
\vskip 4.9cm
\centerline{
\hskip 1mm
\includegraphics[width=0.4\columnwidth]{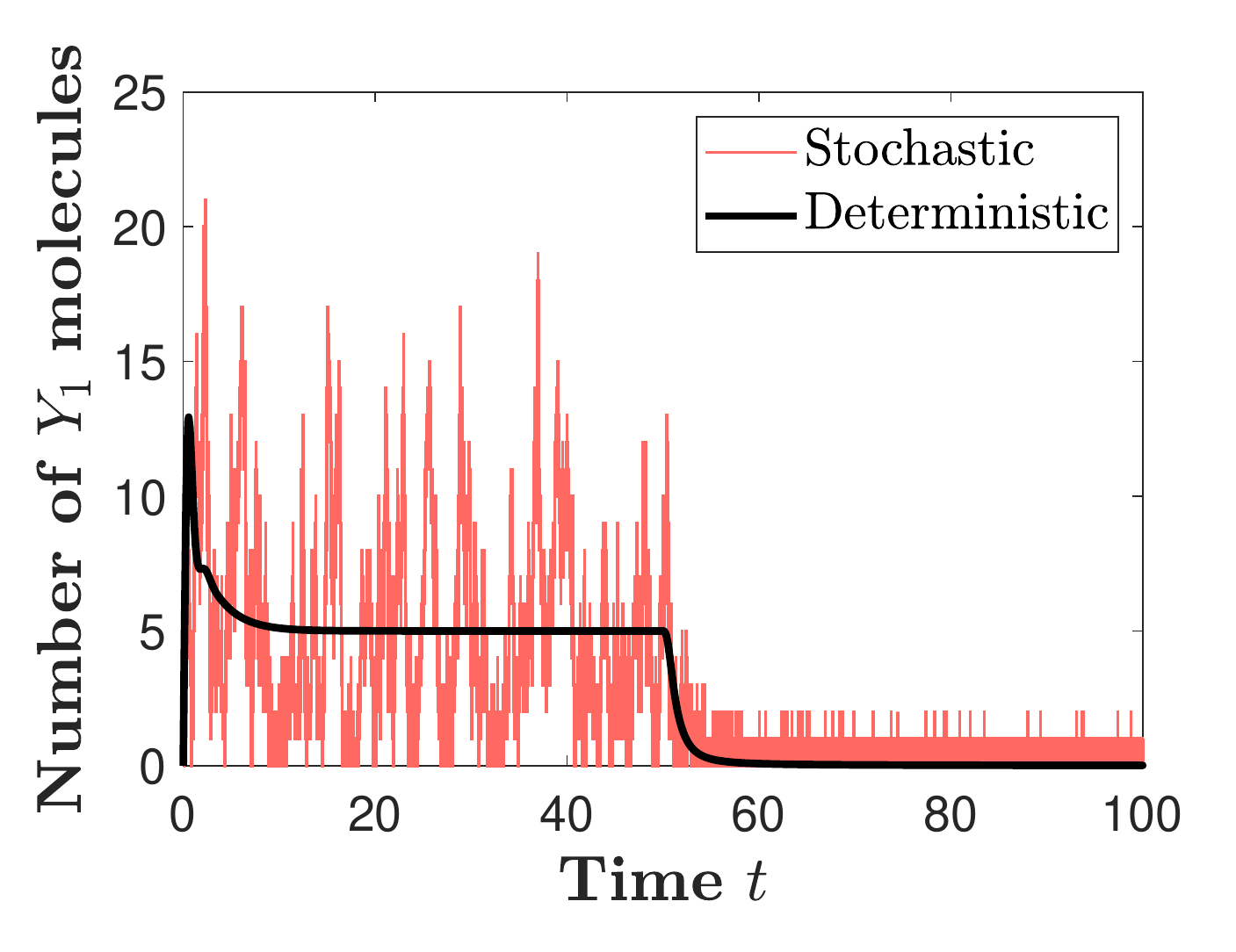}
\hskip -0.3cm
\includegraphics[width=0.4\columnwidth]{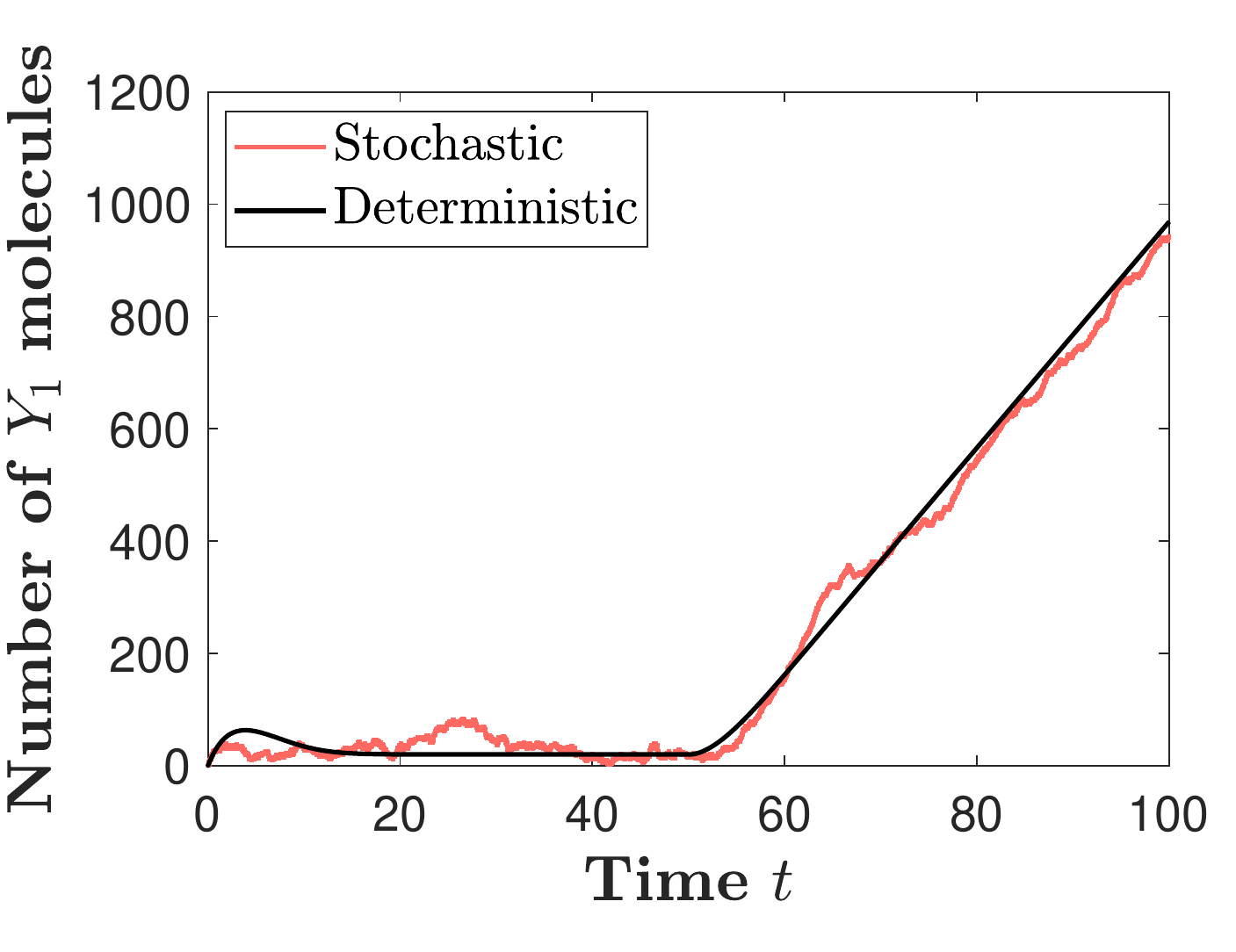}
\hskip -0.3cm
\includegraphics[width=0.4\columnwidth]{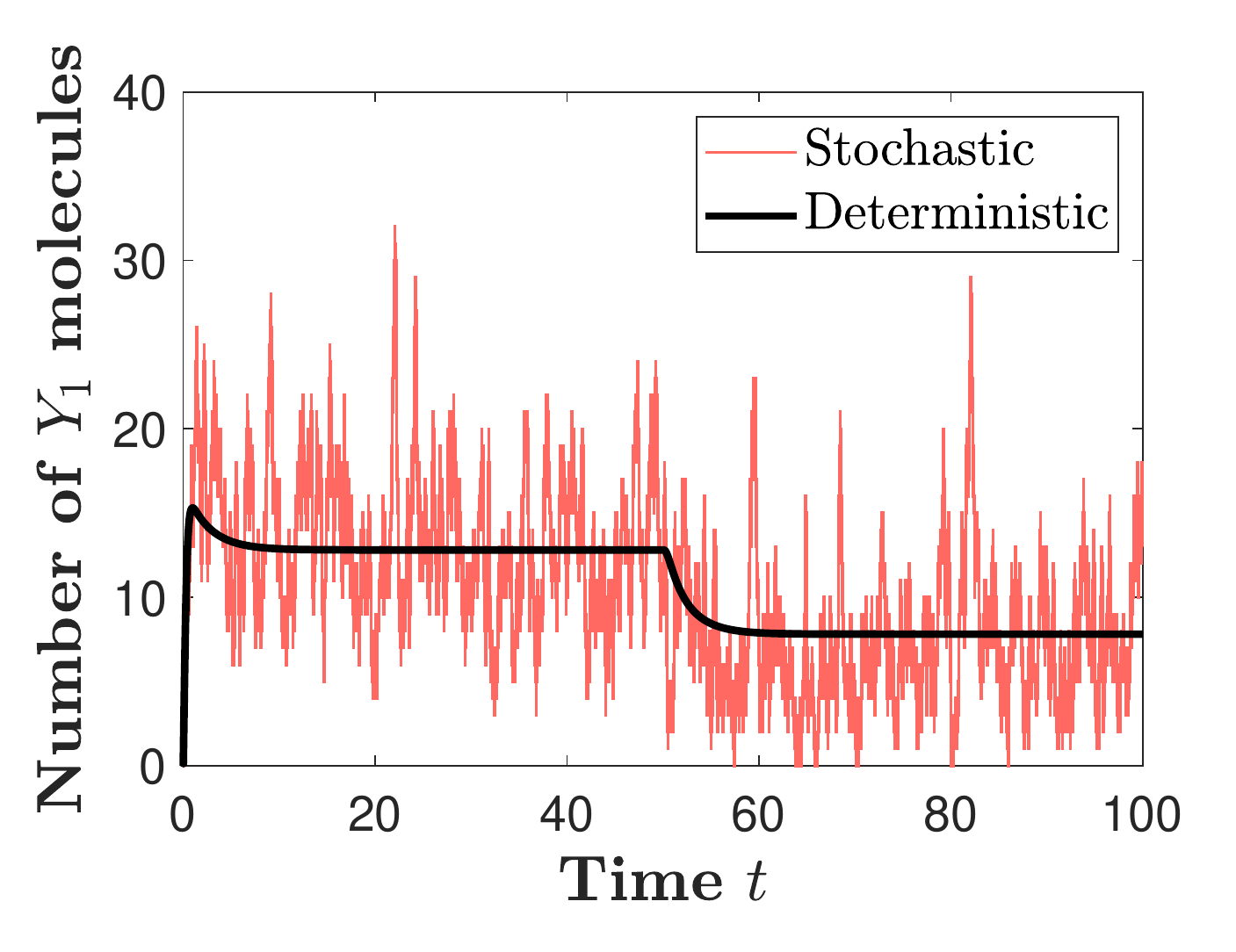}
}
\vskip -5.4cm
\leftline{\hskip -0.8cm (c) \hskip 5.9cm (g) \hskip 5.8cm (k)}
\vskip 4.9cm
\centerline{
\hskip 1mm
\includegraphics[width=0.4\columnwidth]{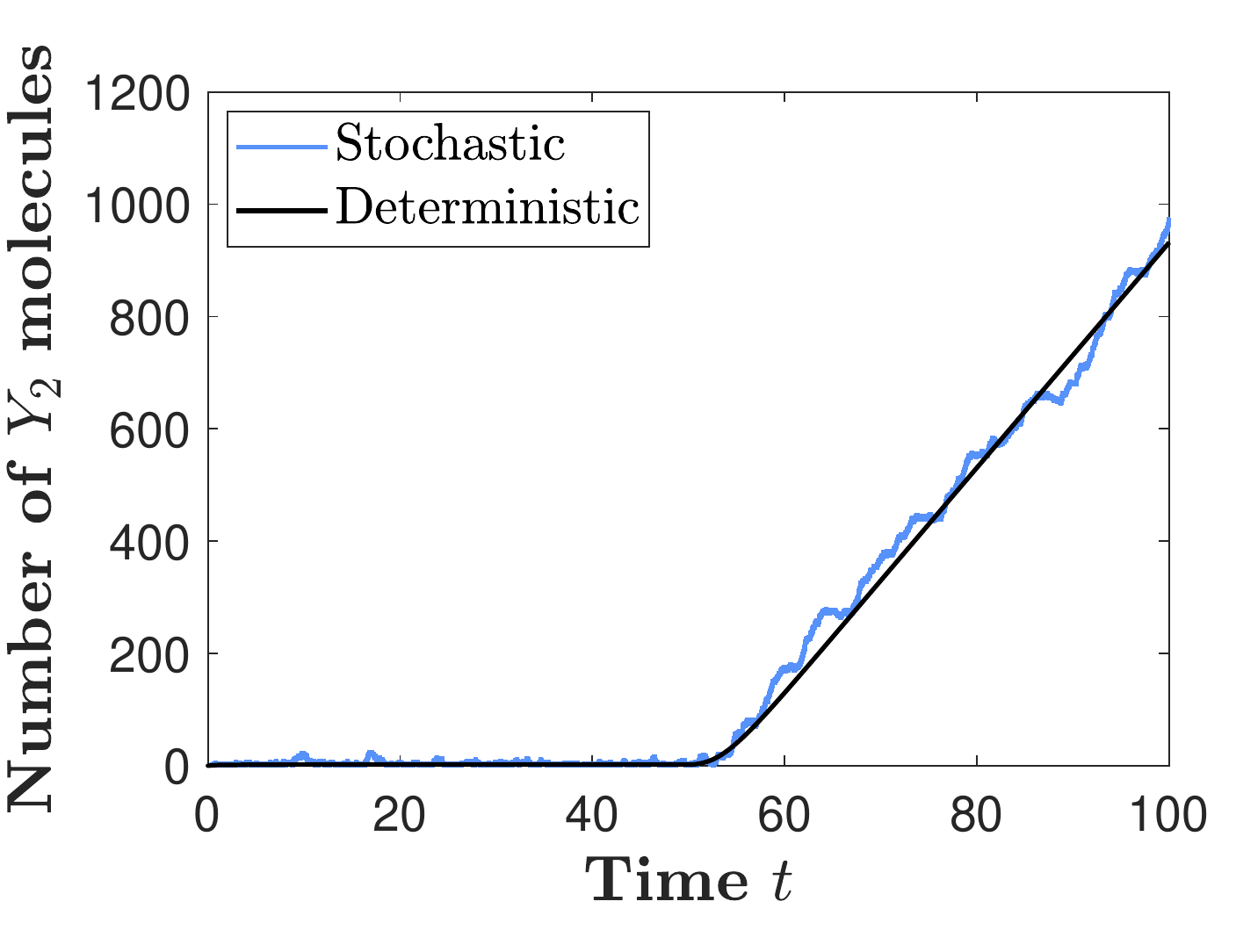}
\hskip -0.3cm
\includegraphics[width=0.4\columnwidth]{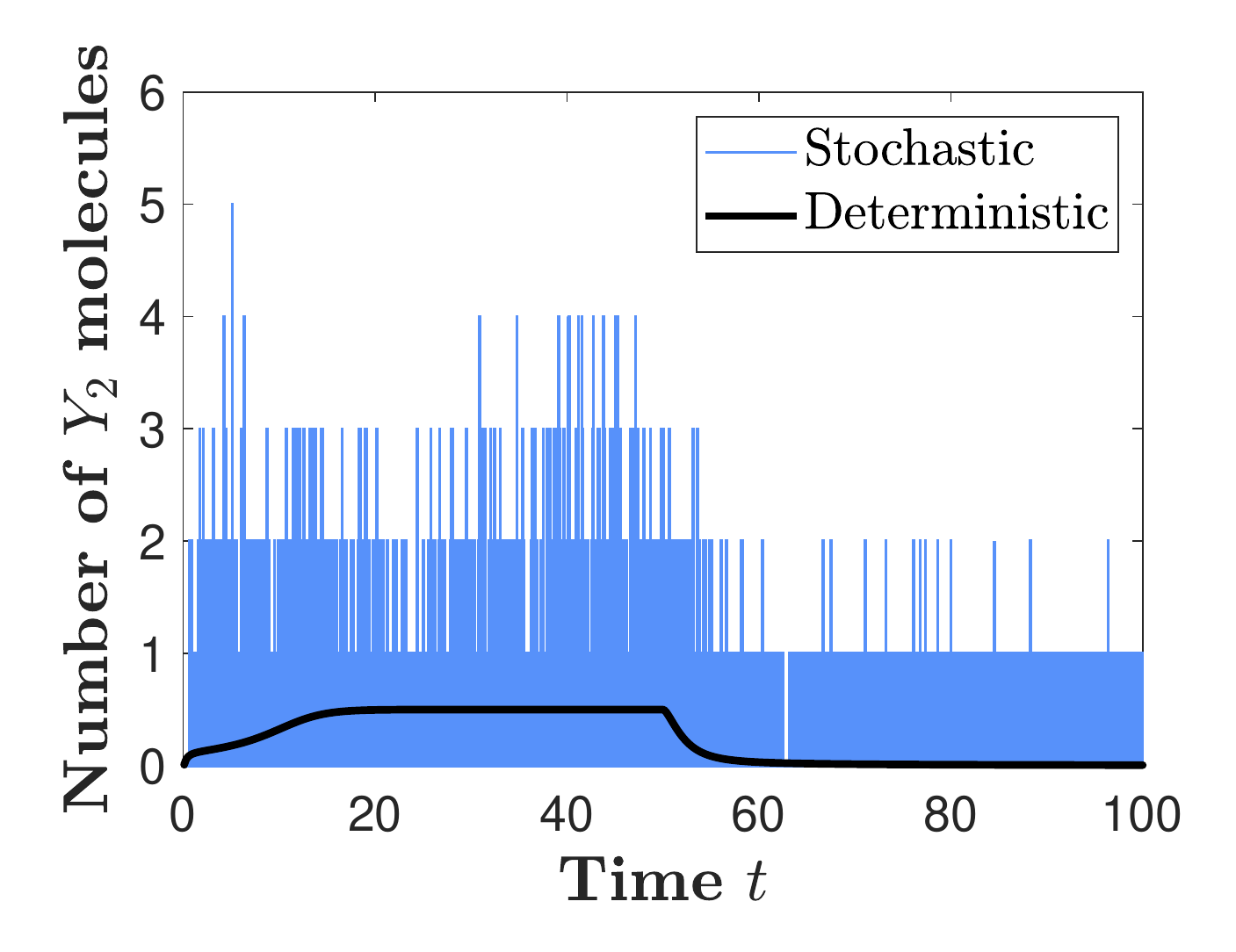}
\hskip -0.3cm
\includegraphics[width=0.4\columnwidth]{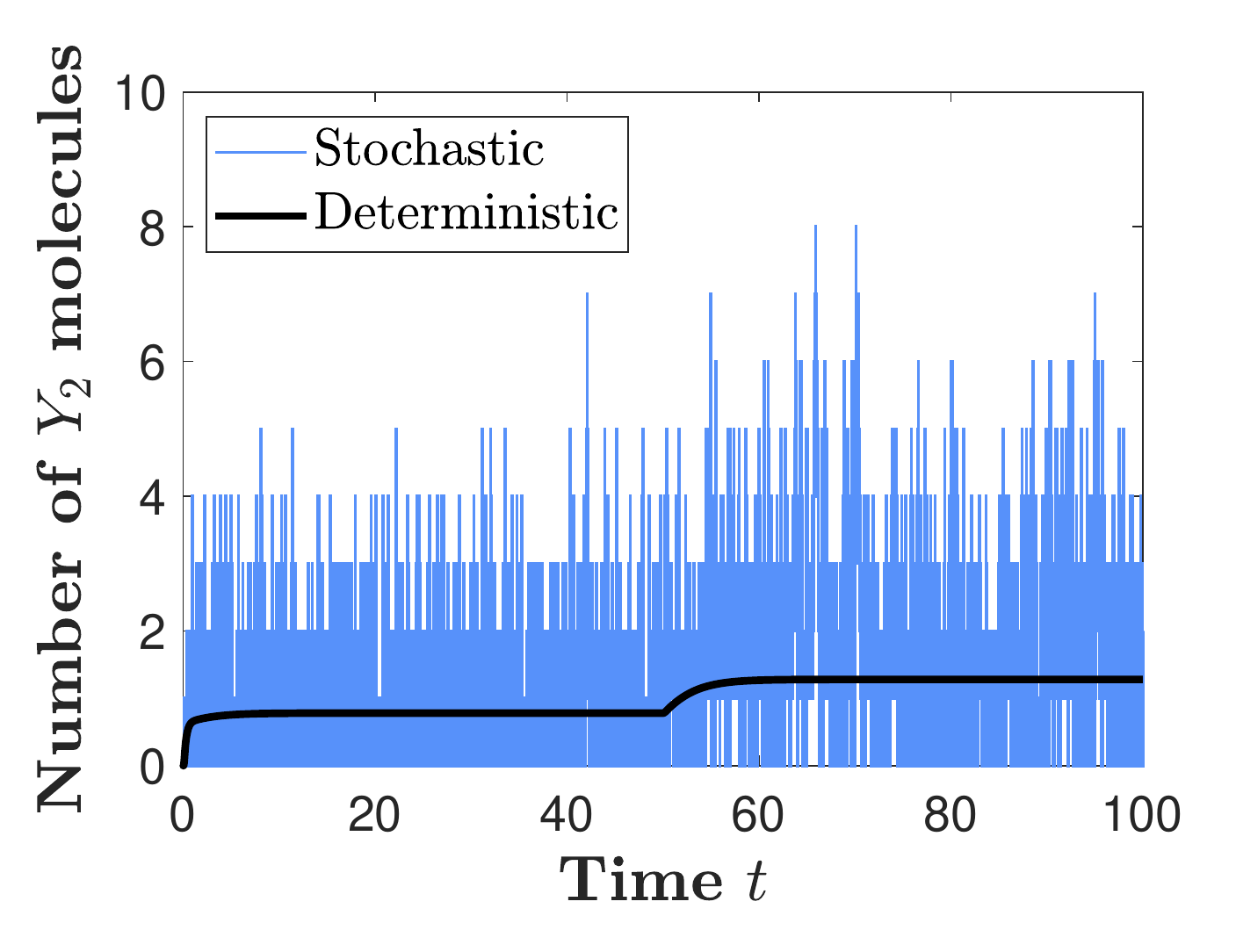}
}
\vskip -5.4cm
\leftline{\hskip -0.8cm (d) \hskip 5.9cm (h) \hskip 5.8cm (l)}
\vskip 4.4cm
\caption{\it{\emph{Application of the IFCs~(\ref{eq:IFCnetapp}) on
 the input network~(\ref{eq:input_1}) with rate coefficients
$(\alpha_1, \alpha_2, \alpha_3) = (2/5,2, 4)$,
and piecewise constant $\alpha_0 = \alpha_0(t)$ which changes at $t = 50$
and leads to a catastrophic bifurcation}.
Panel~{\rm (a)} displays a bifurcation diagram for the output 
network~{\rm(\ref{eq:input_1})}$\cup${\rm(\ref{eq:IFC_positive})},
while panels {\rm (b)--(d)} show the underlying deterministic 
and stochastic trajectories with $\alpha_0 = 4$ for $t < 50$
and $\alpha_0 = 12$ for $t \ge 50$, leading to a change in the 
parameter space indicated by the black arrow in panel~{\rm (a)}.
The control coefficients are fixed to $(\beta_0, \beta_1,\gamma_1, \gamma_2) = (40, 4, 4, 4)$. 
Analogous plots are shown in panels {\rm (e)--(h)} for the output 
network~{\rm(\ref{eq:input_1})}$\cup${\rm(\ref{eq:IFC_negative})}
with $\alpha_0 = 12$ for $t < 50$ and $\alpha_0 = 4$ for $t \ge 50$, 
and $(\beta_0, \beta_1,\gamma_1, \gamma_3) = (40, 4, 4, 4)$, 
and for the output network~{\rm(\ref{eq:input_1})}$\cup${\rm(\ref{eq:IFC_positive_negative})}
in panels {\rm (i)--(l)} using the same coefficient values as in panels {\rm (a)--(d)}, 
and with $\gamma_3 = 4$}.} \label{fig:linear}
\end{figure}  

Let us consider network~(\ref{eq:input_1}) with rate coefficients
$\boldsymbol{\alpha} = (\alpha_0, \alpha_1, \alpha_2, \alpha_3)$ fixed so that the 
input $x_1$-equilibrium from~(\ref{eq:input_eq}) is given by $x_1^{**} = 5$; 
then, the output $x_1$-equilibrium from~(\ref{eq:output_eq_1})
must satisfy the constraint $x_1^* > 5$. Let us stress that, since the input rate coefficients 
$\boldsymbol{\alpha}$ (and the structure
of the input network itself) are generally uncertain (see property (U) from Section~\ref{sec:intro}), 
condition $x_1^* > 5$ is not a-priori known. 
Assume the goal is to steer the output equilibrium to $10$,
i.e. we fix the control coefficients $\beta_0$ and $\gamma_1$ so that $x_1^* =  \beta_0/\gamma_1 = 10$;  
this setup is shown as a black dot at $(5, 10)$ in Figure~\ref{fig:linear}(a), 
which happens to lie in the region where the output network displays a nonnegative equilibrium. 
However, assume that at a future time, as a response to an environmental perturbation, 
an activating transcription factor binds to the underlying gene promoter, 
tripling the transcription rate of the input network, which we model by 
allowing the transcription rate coefficient $\alpha_0$ to be time-dependent, 
see Figure~\ref{fig:genetic}(b).
Such a perturbation would move the system
from coordinate $(5, 10)$ to $(15, 10)$, into the unstable region, 
which we show as a black arrow in Figure~\ref{fig:linear}(a). 
Put another way, the AIFC would fail at its main objective - 
maintaining accurate control robustly with respect to uncertainties
(environmental perturbations). 

In Figure~\ref{fig:linear}(b)--(d), we show the deterministic and stochastic trajectories 
for the species $X_1$, $Y_1$ and $Y_2$, obtained by solving the RREs~(\ref{eq:RREs_plus}) 
and applying the Gillespie algorithm~\cite{GillespieSSA} on~(\ref{eq:input_1})$\cup$(\ref{eq:IFC_positive}), 
respectively; also shown as a dashed grey line in Figure~\ref{fig:linear}(b) is the target equilibrium $x_1^* = 10$. 
For time $t < 50$, when the input network operates as in Figure~\ref{fig:genetic}(a),
and the output system is in the configuration $(5, 10)$
from Figure~\ref{fig:linear}(a), the AIFC achieves control. However, 
for time $t > 50$, when the transcription rate has increased as in Figure~\ref{fig:genetic}(b),
and the output system is at $(15, 10)$ from Figure~\ref{fig:linear}(a),
 control fails; even worse, the species $Y_2$ blows up for all admissible initial conditions.
Intuitively, this hazardous phenomenon (NEC) occurs because, when the target equilibrium is below
the input one, the best accuracy result that the AIFC can achieve is to minimally increase $X_1$. 
Such a task is accomplished with $y_1 \to 0$ which, as a consequence of a hyperbolic relationship between $y_1$ and $y_2$, 
enforces $y_2 \to \infty$, which is a worst stability result. 

\textbf{Pure negative interfacing}.
Consider controller~(\ref{eq:IFCnetapp}) 
with pure direct negative interfacing,  
denoted by $\mathcal{R}_{\beta, \gamma}^{-} \equiv 
\mathcal{R}_{\beta} \cup \mathcal{R}_{\gamma}^0 \cup \mathcal{R}_{\gamma}^{-}$
and given by
\begin{align}
\mathcal{R}_{\beta}(Y_1, Y_2): \;
& & \varnothing & \xrightarrow[]{\beta_0} Y_1, \nonumber \\
& & Y_1 + Y_2 & \xrightarrow[]{\beta_1} \varnothing, \nonumber \\
\mathcal{R}_{\gamma}^{0}(Y_2; \, X_1): \;
& & X_1 & \xrightarrow[]{\gamma_{1}} X_1 + Y_2, \nonumber \\
\mathcal{R}_{\gamma}^{-}(X_j; \, Y_2): \;
& & X_1 + Y_2 & \xrightarrow[]{\gamma_{3}} Y_2.
\label{eq:IFC_negative}
\end{align}
By analogous arguments as with controller~(\ref{eq:IFC_positive}), 
one can prove that deterministic and stochastic NECs occur when 
$x_1^* > x_1^{**}$ (equivalently, when
$\beta_0/\gamma_1 > \alpha_0 \alpha_2/(\alpha_1 \alpha_3)$),
i.e. controller~(\ref{eq:IFC_negative}) cannot steer the
output equilibrium above the input one; when control above the input 
equilibrium is attempted, species
$Y_1$ blows up. We display a bifurcation diagram and trajectories
 for the output network~(\ref{eq:input_1})$\cup$(\ref{eq:IFC_negative})
in Figure~\ref{fig:linear}(e)--(h).

\textbf{Combined positive and negative interfacing}.
Let us now analyze controller~(\ref{eq:IFCnetapp}) 
with both positive and negative interfacing directly applied
to the protein species $X_1$, as suggested by the
derivation in Section~\ref{sec:nonlinear}.
This controller is denoted by $\mathcal{R}_{\beta, \gamma}^{\pm} \equiv 
\mathcal{R}_{\beta} \cup \mathcal{R}_{\gamma}^0 \cup
 \mathcal{R}_{\gamma}^{+} \cup \mathcal{R}_{\gamma}^{+}$
and given by
\begin{align}
\mathcal{R}_{\beta}(Y_1, Y_2): \;
& & \varnothing & \xrightarrow[]{\beta_0} Y_1, \nonumber \\
& & Y_1 + Y_2 & \xrightarrow[]{\beta_1} \varnothing, \nonumber \\
\mathcal{R}_{\gamma}^{0}(Y_2; \, X_1): \;
& & X_1 & \xrightarrow[]{\gamma_{1}} X_1 + Y_2, \nonumber \\
\mathcal{R}_{\gamma}^{+}(X_1; \, Y_1): \;
& & Y_1 & \xrightarrow[]{\gamma_{2}} X_1 + Y_1, \nonumber \\
\mathcal{R}_{\gamma}^{-}(X_j; \, Y_2): \;
& & X_1 + Y_2 & \xrightarrow[]{\gamma_{3}} Y_2.
\label{eq:IFC_positive_negative}
\end{align}
The RREs of the output network~(\ref{eq:input_1})$\cup$(\ref{eq:IFC_positive_negative}) 
have two equilibria, given by
\begin{align}
x_1^* & = \frac{\beta_0}{\gamma_1}, 
\hspace{0.5cm}
x_2^* = \frac{\alpha_0}{\alpha_1}, 
\hspace{0.5cm}
y_2^* = \frac{\beta_0}{\beta_1} (y_1^*)^{-1},
\label{eq:output_eq_2}
\end{align}
where $y_1^*$ satisfies 
\begin{align}
(y_1^*)^2 + \frac{\alpha_3}{\gamma_2} \left(\frac{\alpha_0 \alpha_2}{\alpha_1 \alpha_3} - 
\frac{\beta_0}{\gamma_1} \right) y_1^* - 
\left(\frac{\gamma_3}{\gamma_1 \gamma_2} \frac{\beta_0^2}{\beta_1} \right) & = 0.  \label{eq:output_quadratic}
\end{align}
By design from Section~\ref{sec:nonlinear}, quadratic equation~(\ref{eq:output_quadratic}) 
always has exactly one positive equilibrium, so that no NEC can 
occur with controller~(\ref{eq:IFC_positive_negative});
we confirm this fact in Figure~\ref{fig:linear}(i)--(l).

\subsection{Arbitrary unimolecular input networks}
Test network~(\ref{eq:input_1}) demonstrates
that controller~(\ref{eq:IFCnetapp}) with only positive, or only negative,
 interfacing does not generically ensure stability
of the output network. In other words, the output network experiences NECs
over larger regions in the parameter space, as displayed in Figures~\ref{fig:linear}(a) and~(e).
On the other hand, controller~(\ref{eq:IFCnetapp}) with combined positive and negative
interfacing, applied directly to the species of interest, induces no
NEC, as shown in Figure~\ref{fig:linear}(i). 
A special feature of unimolecular networks is that distinct 
species cannot influence each other negatively. 
Consequently, to ensure existence of a nonnegative equilibrium, 
negative interfacing must generally be applied directly to the target 
species whose dynamics is controlled, while positive interfacing
can be applied directly or indirectly. We now more formally state this result; 
for more details and a proof, see Appendix~\ref{app:proof}. To aid the statement of the theorem,
consider two sets of unimolecular networks whose $x_1$-equilibrium is zero, $x_1^{**} = 0$: 
those that contain the target species $X_2$, and
those with the target species $X_2$ deleted (i.e. we fix $x_2^{**} \equiv 0$).
These two sets of networks form a negligibly small
subset of general unimolecular networks, and are called degenerate; the set of 
all other unimolecular networks is said to be \emph{nondegenerate}.

\begin{theorem}\label{theorem:IFCsummary}
Consider an arbitrary nondegenerate unimolecular input network 
$\mathcal{R}_{\alpha}$ whose {\rm RRE}s have an asymptotically stable equilibrium,
the family of controllers $\mathcal{R}_{\beta,\gamma}$ given by~{\rm (\ref{eq:IFCnetapp})}, 
and the output network $\mathcal{R}_{\alpha,\beta,\gamma} = 
\mathcal{R}_{\alpha} \cup \mathcal{R}_{\beta,\gamma}$.
Then, controller $\mathcal{R}_{\beta,\gamma}^{\pm}$ with 
both positive and negative interfacing, with negative interfacing
being direct, ensures that the output network 
$\mathcal{R}_{\alpha,\beta,\gamma}$ has a nonnegative equilibrium 
for all parameter values $(\boldsymbol{\alpha}, \boldsymbol{\beta},\boldsymbol{\gamma}) 
\in \mathbb{R}_{>}^{a + b + c}$. 
On the other hand, the other variants of the controller~{\rm (\ref{eq:IFCnetapp})}
do not generically ensure that the output network $\mathcal{R}_{\alpha,\beta,\gamma}$
has a nonnegative equilibrium; furthermore, when a nonnegative
equilibrium does not exist, these controllers induce deterministic and stochastic blow-ups
({\rm NEC}s) for all nonnegative initial conditions.
\end{theorem}
\begin{proof}
See Appendix~\ref{app:proof}.
\end{proof}
\noindent Note that the only variant of~(\ref{eq:IFCnetapp}) that generically ensures
a nonnegative equilibrium is also the one which may be experimentally most
challenging to implement. In particular, one must generally implement the second-order
reaction $\mathcal{R}_{\gamma}^{-}$ from~(\ref{eq:IFCnetapp}) 
applied \emph{directly} to the target species whose dynamics is controlled. 

When the equilibrium of the target species from the input network is zero, $x_1^{**} = 0$
(a degenerate case), both the positive-negative controller $\mathcal{R}_{\beta,\gamma}^{\pm}$, 
and the AIFC $\mathcal{R}_{\beta,\gamma}^{+}$, generically 
ensure existence of a nonnegative equilibrium.
However, these degenerate input networks can describe only a small class of biochemical processes.
For example, when there is no basal transcription, i.e. when $\alpha_0 = 0$, 
the equilibrium of network~(\ref{eq:input_1}) is zero and, consequently, 
the output equilibrium~(\ref{eq:output_eq_1}) is always nonnegative. 
In particular, the output equilibrium is then nonnegative 
independent of the uncertainties in the input coefficients, 
so that the key challenge (U) highlighted in Section~\ref{sec:intro} is mitigated.
This gene-expression input network without
basal transcription, $\alpha_0 = 0$, has been used in~\cite{Khammash}
to demonstrate a desirable stochastic behavior of the AIFC.
However, as we have shown in this section, when a more general gene-expression 
model is used, with  $\alpha_0 \ne 0$, the AIFC can fail and induce
both deterministic and stochastic catastrophes as a consequence
of the challenge (U). Similar degenerate input networks have also been used in~\cite{AIFC_1,AIFC_2}.

\section{Control of bimolecular input networks: Curse of dimensionality} \label{sec:second_order}
As expressed by challenge (N) in Section~\ref{sec:intro}, most intracellular networks are bimolecular,
rather than unimolecular, limiting the applicability of Theorem~\ref{theorem:IFCsummary}. 
For example, in Section~\ref{sec:first_order}, we have used the unimolecular input network~(\ref{eq:input_1}) 
with a time-dependent rate coefficient to model intracellular gene expression 
with regulated transcription. To obtain a more realistic model, instead of 
allowing for an effective time-dependent rate coefficient, the reduced network~(\ref{eq:input_1}) 
should be extended by including other coupled auxiliary species 
(e.g. transcription factors and genes) and processes (pre-transciptional and post-translational events);
the resulting extended input network is then bimolecular, and therefore
Theorem~\ref{theorem:IFCsummary} no longer applies.
In particular, a special property of unimolecular networks is that distinct species can influence each other only
positively; in contrast, distinct species can influence each other both positively and negatively
in bimolecular networks. For this reason, when stable unimolecular input networks are controlled,
NECs can be eliminated purely by ensuring that the 
controlling species equilibrium $\mathbf{y}^*$ is positive;
the input species equilibrium $\mathbf{x}^*$ is then necessarily nonnegative. 
In other words, the problem of controlling unimolecular networks 
is independent of the challenge (HD) from Section~\ref{sec:intro},
i.e. the problem does not become more challenging as the dimension of the input network increases.
On the other hand, we show in this section that, as a consequence of 
nonlinearities and positive-negative interactions among distinct species, 
the problem of controlling bimolecular networks
suffers from the \emph{curse of dimensionality}  - the problem becomes
more challenging as dimension of the input network increases. 
In particular, we show that, for bimolecular networks, ensuring 
that the controlling species equilibrium is positive is generally
not sufficient for nonnegativity of the input species equilibrium.

\subsection{Two-species reduced input network: Residual NEC} \label{sec:reduced}
Let us consider a two-dimensional reduced model of an intracellular process, given by 
the bimolecular input network $\mathcal{R}_{\alpha}^2(X_1, X_2)$ which reads
\begin{align}
\mathcal{R}_{\alpha}^2(X_1, X_2):
& & \varnothing \xrightarrow[]{\alpha_{0}} & X_1, \hspace{0.3cm}
X_1 \xrightarrow[]{\alpha_{1}} X_2, \hspace{0.3cm}
X_1 + X_2 \xrightarrow[]{\alpha_{2}} 2 X_2, \hspace{0.3cm}
X_2 \xrightarrow[]{\alpha_{3}} \varnothing, 
\label{eq:input_2}
\end{align}
where $X_1$ is produced from a source and converted 
into a degradable species $X_2$ via first- and second-order conversion reactions. 
We assume that $X_1$ is a target species, while $X_2$ is a residual
species, i.e. $X_2$ cannot be interfaced with a given controller.
The RREs of~(\ref{eq:input_2}) have a unique asymptotically stable equilibrium, given by
\begin{align}
x_1^{**} & = \frac{\alpha_0 \alpha_3}{\alpha_0 \alpha_2 + \alpha_1 \alpha_3}, 
\hspace{0.5cm}
x_2^{**} = I_2(x_1^{**}; \, \boldsymbol{\alpha}) = \frac{\alpha_0}{\alpha_3}, \label{eq:input2_eq}
\end{align}
where the function $I_2 = I_2(x_1; \, \boldsymbol{\alpha})$ is given by
\begin{align}
I_2(x_1; \, \boldsymbol{\alpha}) \equiv \frac{\alpha_1}{\alpha_2} x_1
\left(\frac{\alpha_3}{\alpha_2} - x_1 \right)^{-1}. \label{eq:input2_eqI}
\end{align}
We call~(\ref{eq:input2_eqI}) a \emph{residual invariant}, which is simply the $x_2$-equilibrium
expressed as a function of the $x_1$-equilibrium and the input coefficients $\boldsymbol{\alpha}$.

\begin{figure}[!htbp]
\vskip -1.6cm
\centerline{
\includegraphics[width=0.4\columnwidth]{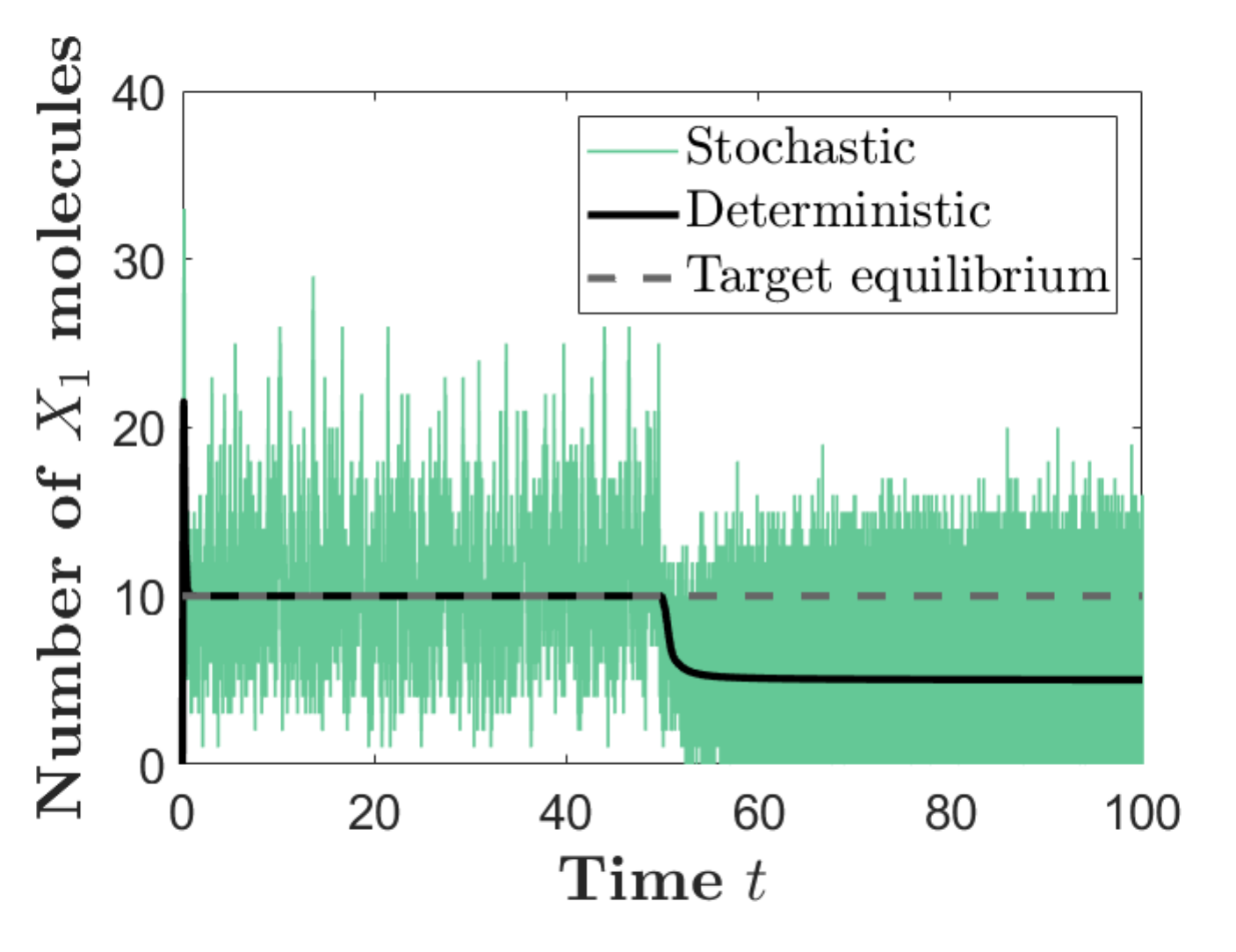}
\hskip 1mm
\includegraphics[width=0.4\columnwidth]{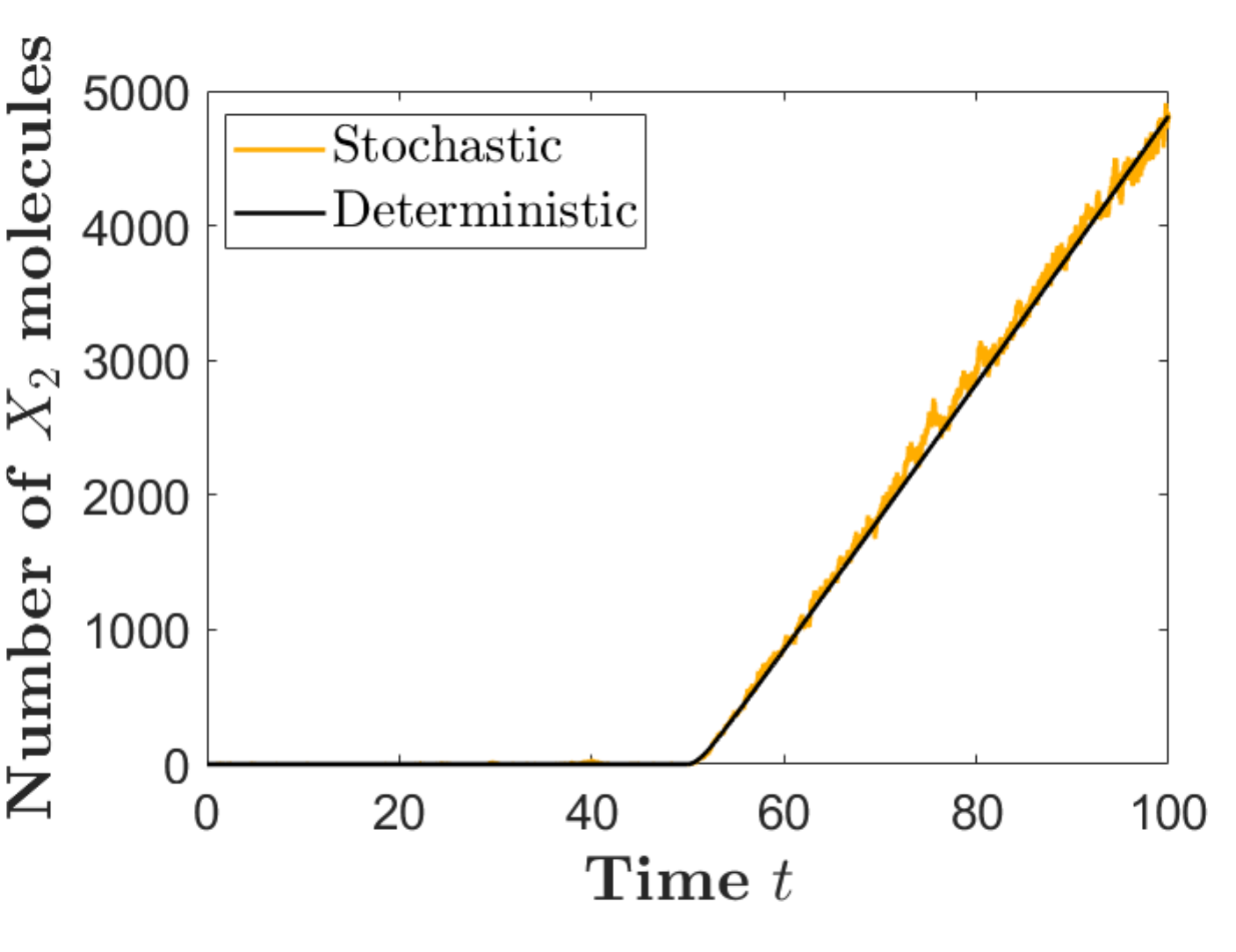}
}
\vskip -5.4cm
\leftline{\hskip 2.2cm (a) \hskip 6.3cm (b)}
\vskip 4.9cm
\centerline{
\includegraphics[width=0.4\columnwidth]{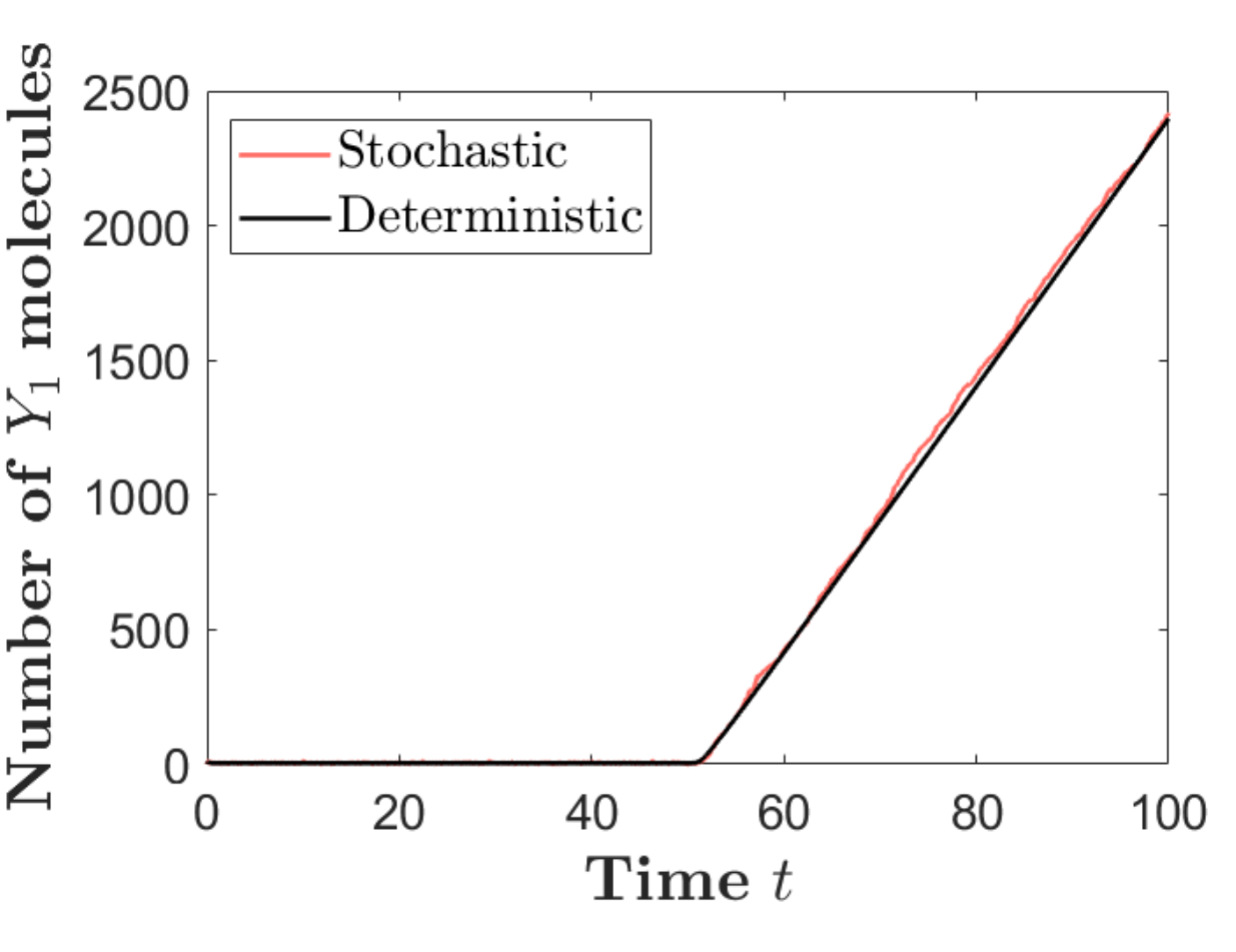}
\hskip 1mm
\includegraphics[width=0.4\columnwidth]{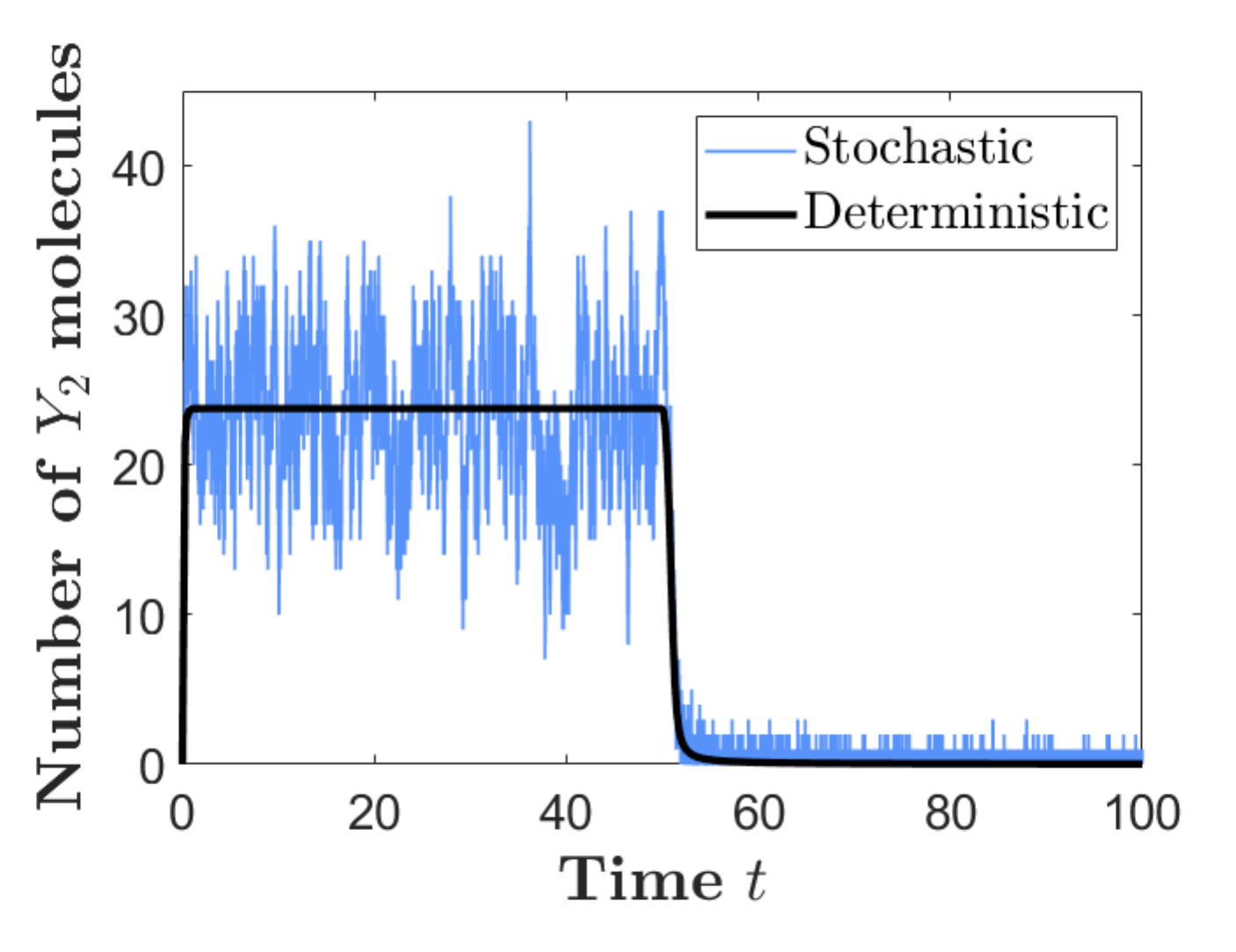}
}
\vskip -5.4cm
\leftline{\hskip 2.2cm (c) \hskip 6.3cm (d)}
\vskip 4.4cm
\caption{\it{\emph{Application of the IFC~(\ref{eq:IFC_positive_negative}) on
the input network~(\ref{eq:input_2}) with rate coefficients
$(\alpha_0, \alpha_1, \alpha_3) = (200, 1/7, 5)$,
and  $\alpha_2 = 1/3$ for $t < 50$, which changes to
$\alpha_2 = 1$ for $t \ge 50$ and leads to a catastrophic bifurcation}.
Panels~{\rm (a)}--{\rm (d)} display the deterministic 
and stochastic trajectories for the output 
network~{\rm(\ref{eq:IFC_positive_negative})}$\cup${\rm(\ref{eq:input_2})}, 
with control coefficients $(\beta_0, \beta_1,\gamma_1, \gamma_2, \gamma_3) 
= (100, 1,10,10,1)$.}}  \label{fig:nonlinear}
\end{figure}  

Let us embed the controller~(\ref{eq:IFC_positive_negative}) into~(\ref{eq:input_2}),
leading to the RREs of the output network given by
\begin{align}
\frac{\mathrm{d} x_1}{\mathrm{d} t} & = \left(\alpha_0 - \alpha_1 x_1 - \alpha_2 x_1 x_2 \right) 
+ h(x_1, y_1, y_2; \, \boldsymbol{\gamma}),
\hspace{0.3cm} \textrm{where } h(x_1, y_1, y_2; \, \boldsymbol{\gamma}) = \gamma_2 y_1 - \gamma_3 x_1 y_2,\nonumber \\
\frac{\mathrm{d} x_2}{\mathrm{d} t} & = \alpha_1 x_1 + \alpha_2 x_1 x_2 - \alpha_3 x_2, \nonumber \\
\frac{\mathrm{d} y_1}{\mathrm{d} t} & = g_1(x_1, y_1, y_2; \, \boldsymbol{\beta}, \boldsymbol{\gamma}) 
= \beta_0 - \beta_{1} y_1 y_2, \nonumber \\
\frac{\mathrm{d} y_2}{\mathrm{d} t} & = g_2(x_1, y_1, y_2; \, \boldsymbol{\beta}, \boldsymbol{\gamma})
= \gamma_{1} x_1 - \beta_{1} y_1 y_2,
 \label{eq:RREs_2}
\end{align}
which display two equilibria, one of which 
is never nonnegative, while the other equilibrium satisfies
\begin{align}
x_1^* & = \frac{\beta_0}{\gamma_1}, 
\hspace{0.5cm}
x_2^* = I_2 \left(x_1^{*}, \boldsymbol{\alpha} \right),
\hspace{0.5cm} y_1^* > 0, \; \; y_2^* > 0.
 \label{eq:output2_eq}
\end{align}
In particular, the functional form of the residual equilibrium $x_2^{*}$
from~(\ref{eq:output2_eq}) is the same as the form of $x_2^{**}$ from~(\ref{eq:input2_eq}); 
put another way, the form of the residual species equilibrium is invariant under control, 
justifying calling the function~(\ref{eq:input2_eqI}) a residual invariant.
To ensure that the output network~(\ref{eq:IFC_positive_negative})$\cup$(\ref{eq:input_2})
displays a nonnegative equilibrium, the residual invariant~(\ref{eq:input2_eqI}), 
now evaluated at the target equilibrium $x_1^* = \beta_0/\gamma_1$, 
must be nonnegative, giving rise to the condition
\begin{align}
I_2 \left(x_1^{*}, \boldsymbol{\alpha} \right) \ge 0 \iff
\frac{\beta_0}{\gamma_1} \le \frac{\alpha_3}{\alpha_2}.
 \label{eq:output2_condition}
\end{align}
Therefore, while~(\ref{eq:IFC_positive_negative}) unconditionally guarantees existence of a nonnegative equilibrium 
for stable unimolecular input networks (see Theorem~\ref{theorem:IFCsummary}), 
the same is generally false for bimolecular networks, as the equilibrium of the 
residual species, which are not interfaced with the controller, can become negative.
Let us note that the combination of parameters $\alpha_3/\alpha_2$ from~(\ref{eq:output2_condition})
cannot be interpreted as a component of the input equilibrium~(\ref{eq:input2_eq}).
Let us also note that the residual invariant evaluated at the input $x_1$-equilibrium is 
always nonnegative, $I_2 \left(x_1^{**}, \boldsymbol{\alpha} \right) = \alpha_0/\alpha_3 \ge 0$; 
equation~(\ref{eq:output2_condition}) shows that this unconditional
nonnegativity is violated when the control is applied.

In Figure~\ref{fig:nonlinear}, we display the deterministic and stochastic trajectories
for the output network~(\ref{eq:IFC_positive_negative})$\cup$ (\ref{eq:input_2})
over a time-interval such that condition~(\ref{eq:output2_condition}) is satisfied
for $t < 50$, and violated for $t \ge 50$ due to a change in $\alpha_2$. 
One can notice that the output network undergoes deterministic and 
stochastic NECs. Critically, not only does the controlling species $Y_1$ blow up, 
but also the residual species $X_2$. In other words, controller~(\ref{eq:IFC_positive_negative})
\emph{destabilizes} the originally asymptotically stable input network~(\ref{eq:input_2}).
We call this hazardous phenomenon a \emph{residual} NEC, as it arises 
because a residual species has no nonnegative equilibrium 
(equivalently, because a residual invariant is not nonnegative). 
Intuitively, when the concentration of the target species $X_1$ is increased
beyond the upper bound from~(\ref{eq:output2_condition}), residual species $X_2$, 
which influences $X_1$ negatively, counteracts the positive action of the
controlling species $Y_1$, resulting in a joint blow-up. We also display this
phenomenon in context of intracellular control in Figure~\ref{fig:cells}.

Let us stress that the condition $I_2 \left(x_1^{*}, \boldsymbol{\alpha} \right) \ge 0$
from~(\ref{eq:output2_condition}) must be obeyed by \emph{every} molecular controller 
(e.g. containing integral, proportional and/or derivative actions~\cite{Control_theory}) 
that cannot be interfaced with $X_2$. 
Put another way, no matter how one chooses the functions $g_1$, $g_2$ and $h$
in~(\ref{eq:RREs_2}), the inequality $I_2 \left(x_1^{*}, \boldsymbol{\alpha} \right) \ge 0$
must be satisfied, which imposes an upper bound on the achievable output equilibrium 
via $x_1^* < \alpha_3/\alpha_2$. The only way to eliminate this residual invariant condition
is to eliminate the residual species $X_2$, i.e. to design an appropriate controller that
can be interfaced with both $X_1$ and $X_2$. However, as stated in challenges (N), (HD)
and (U) in Section~\ref{sec:intro}, intracellular networks generally contain larger number of coupled biochemical species
with different biophysical properties, some of which may be unknown (hidden) or poorly characterized; 
therefore, it is generally unfeasible to demand that a controller is designed that can be interfaced 
with any desired species. In what follows, we further investigate this issue in context of model choice.

\subsection{Three-species extended input network: Phantom control} \label{sec:extended}
The two-dimensional network~(\ref{eq:input_2}) has been put forward as 
a reduced model of an intracellular network, obtained
by neglecting a number of molecular species that do not influence the
dynamics of $X_1$ and $X_2$, or by using perturbation theory to eliminate slower or
faster auxiliary species from considerations~\cite{Pavliotis}. However, the goal of such model reductions
is to capture the dynamics of the species $X_1$ and $X_2$ on a desired time-scale, 
and not necessarily to capture how the underlying higher-dimensional model responds to control. 
In this context, let us extend network~(\ref{eq:input_2}) by including a ``hidden" residual species 
$X_3$ into consideration, which interacts with $X_1$ and $X_2$ according to the three-dimensional input 
network $\mathcal{R}_{\alpha, \varepsilon}^3 = \mathcal{R}_{\alpha, \varepsilon}^3(X_1, X_2,X_3)$, given by
\begin{align}
\mathcal{R}_{\alpha, \varepsilon}^3:
& & \varnothing \xrightarrow[]{\alpha_{0}} & X_1, \hspace{0.3cm}
X_1 \xrightarrow[]{\alpha_{1}} X_2, \hspace{0.3cm}
X_1 + X_2 \xrightarrow[]{\alpha_{2}} 2 X_2, \hspace{0.3cm}
X_2 \xrightarrow[]{\alpha_{3}} \varnothing, \nonumber \\
& & X_3 \xrightarrow[]{\varepsilon} & X_1 + X_3, \hspace{0.3cm}
\varnothing \xrightarrow[]{\alpha_{4}}  X_3, \hspace{0.3cm}
X_3 \xrightarrow[]{\alpha_{5}} 2 X_3, \hspace{0.3cm}
X_2 + X_3 \xrightarrow[]{\alpha_{6}} X_2, 
\; \; \; \textrm{where } \frac{\alpha_5}{\alpha_6} < \frac{\alpha_0}{\alpha_3},
\; \;  0 < \varepsilon \ll 1.
\label{eq:input_3}
\end{align}
The residual species $X_3$ influences $X_1$
and $X_2$ only weakly via the slower reaction $X_3 \xrightarrow[]{\varepsilon}  
X_1 + X_3$ in~(\ref{eq:input_3}), where $0 < \varepsilon \ll 1$ is sufficiently small.
One can readily show that the dynamics of the species $X_1$ and $X_2$ 
from the input networks~(\ref{eq:input_2}) and~(\ref{eq:input_3}) are identical
as $\varepsilon \to 0$, which we denote by writing $\lim_{\varepsilon \to 0} \mathcal{R}_{\alpha, \varepsilon}^3
=  \mathcal{R}_{\alpha}^2$. Furthermore, the RREs of the network~(\ref{eq:input_3}) 
have a unique asymptotically stable positive equilibrium, given at the leading order by
\begin{align}
x_1^{**} & \approx \frac{\alpha_0 \alpha_3}{\alpha_0 \alpha_2 + \alpha_1 \alpha_3}, 
\hspace{0.5cm} 
x_2^{**} \approx I_2(x_1^{**}; \, \boldsymbol{\alpha}) = \frac{\alpha_0}{\alpha_3}, 
\hspace{0.5cm} 
x_3^{**} \approx I_3(x_1^{**}; \, \boldsymbol{\alpha})
= \frac{\alpha_4}{\alpha_6} \left(\frac{\alpha_0}{\alpha_3} - \frac{\alpha_5}{\alpha_6}\right)^{-1},
\label{eq:input3_eq}
\end{align}
where the residual invariants $I_2 = I_2(x_1; \, \boldsymbol{\alpha})$
and $I_3 = I_3(x_1; \, \boldsymbol{\alpha})$ are given by
\begin{align}
I_2(x_1; \, \boldsymbol{\alpha}) \equiv \frac{\alpha_1}{\alpha_2} x_1
\left(\frac{\alpha_3}{\alpha_2} - x_1 \right)^{-1}, 
\hspace{0.5cm} 
I_3(x_1; \, \boldsymbol{\alpha}) \equiv \frac{\alpha_4}{\alpha_6} 
\left(I_2(x_1; \, \boldsymbol{\alpha}) - \frac{\alpha_5}{\alpha_6} \right)^{-1}.
 \label{eq:input3_eqI}
\end{align}
 In what follows, we let $\boldsymbol{\alpha} = (\alpha_0, \alpha_1, \alpha_2, \alpha_3, \alpha_4, \alpha_5, \alpha_6)
= (200, 1/7, 1/3,5,1,4,1)$ and $\varepsilon = 10^{-2}$; in Figure~\ref{fig:nonlinear3_eps}, 
we demonstrate that the $(x_1, x_2)$-dynamics of networks~(\ref{eq:input_2}) and~(\ref{eq:input_3})
are then close.

\begin{figure}[!htbp]
\vskip  -1.5cm
\centerline{
\includegraphics[width=0.4\columnwidth]{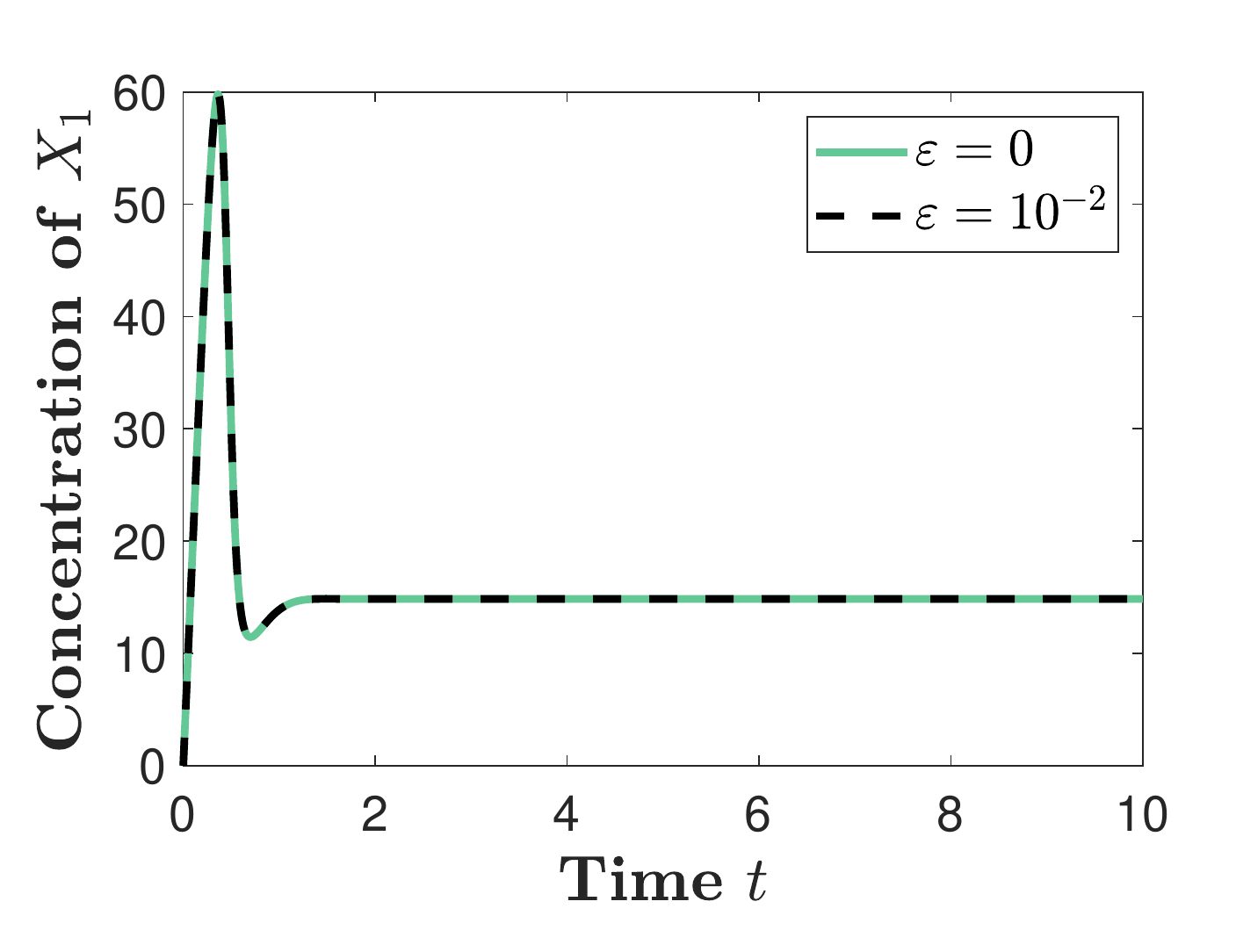}
\hskip 1mm
\includegraphics[width=0.4\columnwidth]{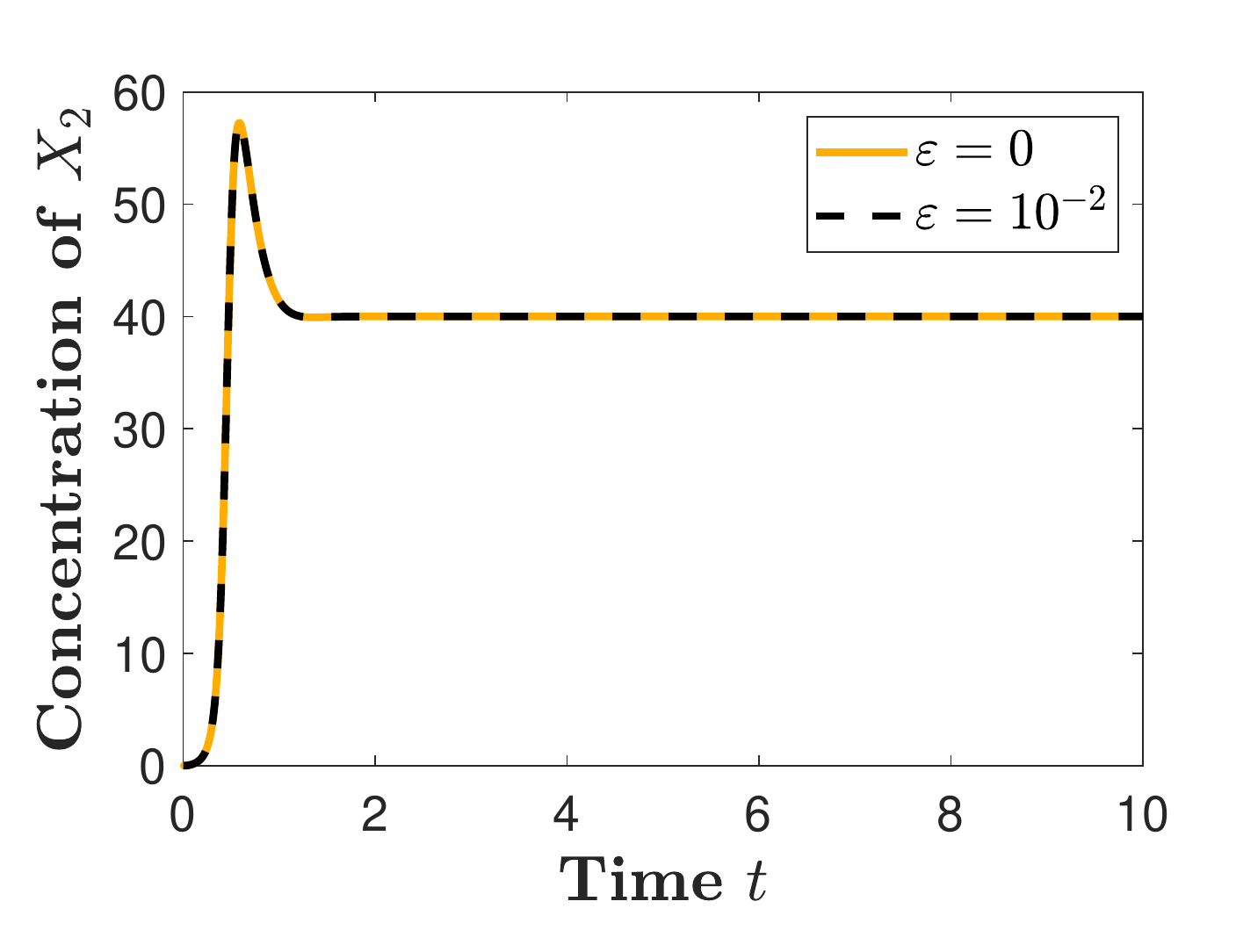}
}
\vskip -5.4cm
\leftline{\hskip 2.2cm (a) \hskip 6.3cm (b)}
\vskip 4.4cm
\caption{\it{\emph{Input network~(\ref{eq:input_3}) with rate coefficients $\boldsymbol{\alpha} = 
(\alpha_0, \alpha_1, \alpha_2, \alpha_3, \alpha_4, \alpha_5, \alpha_6)
= (200, 1/7, 1/3,5,1,4,1)$ and different values of $\varepsilon$}.
Panels~{\rm (a)}--{\rm (b)} display the deterministic trajectories for the species $X_1$
and $X_2$, respectively, from the input network~{\rm(\ref{eq:input_3})} with $\varepsilon = 0$
(equivalently, the input network~{\rm(\ref{eq:input_2})}) and 
with $\varepsilon = 10^{-2}$.}}  \label{fig:nonlinear3_eps}
\end{figure}  

Let us now embed the controller~(\ref{eq:IFC_positive_negative}) into~(\ref{eq:input_3}); 
the RREs of the output network~(\ref{eq:IFC_positive_negative})$\cup$(\ref{eq:input_3}) 
have two equilibria, both of which have identical $(x_1, x_2, x_3)$-components, given by
\begin{align}
x_1^* & = \frac{\beta_0}{\gamma_1}, 
\hspace{0.5cm}
x_2^* =  I_2(x_1^{*}; \, \boldsymbol{\alpha}),
\hspace{0.5cm} 
x_3^* = I_3(x_1^{*}; \, \boldsymbol{\alpha}).
 \label{eq:output3_eq}
\end{align}
In addition to requiring that $I_2(x_1^{*}; \, \boldsymbol{\alpha}) \ge 0$, 
one must now also demand that $I_3(x_1^{*}; \, \boldsymbol{\alpha}) \ge 0$,
to ensure that the (previously neglected) residual species $X_3$ displays a nonnegative equilibrium, 
leading to
\begin{align}
I_2 \left(x_1^{*}, \boldsymbol{\alpha} \right), 
I_3 \left(x_1^{*}, \boldsymbol{\alpha} \right) \ge 0 \iff
\frac{\alpha_3 \alpha_5}{\alpha_1 \alpha_6 + \alpha_2 \alpha_5} 
\le \frac{\beta_0}{\gamma_1} \le \frac{\alpha_3}{\alpha_2}.
 \label{eq:output3_condition}
\end{align}
By accounting for the residual species $X_3$, 
a lower bound is imposed on the achievable 
output equilibrium $x_1^* = \beta_0/\gamma_1$ in~(\ref{eq:output3_condition}),
while no such lower bound is imposed in~(\ref{eq:output2_condition}).
Therefore, while the reduced network~(\ref{eq:input_2}) is suitable 
to approximate the dynamics of $X_1$ and $X_2$ from the extended
network~(\ref{eq:input_3}), i.e. 
$\lim_{\varepsilon \to 0} \mathcal{R}_{\alpha, \varepsilon}^3 = \mathcal{R}_{\alpha}^2$,
network~(\ref{eq:input_2}) is not suitable to approximate how~(\ref{eq:input_3}) responds to control, 
i.e. $\lim_{\varepsilon \to 0} (\mathcal{R}_{\alpha, \varepsilon}^3 \cup \mathcal{R}_{\beta, \gamma}^{\pm}) 
\ne (\mathcal{R}_{\alpha}^2 \cup \mathcal{R}_{\beta, \gamma}^{\pm})$.
When a reduced network is successfully controlled under a parameter choice
for which a corresponding extended network fails to be controlled, we
say that a \emph{phantom control} occurs for the reduced network.
Hence, when the lower bound in~(\ref{eq:output3_condition})
is violated, network~(\ref{eq:IFC_positive_negative})$\cup$(\ref{eq:input_2}) 
displays phantom control.

For the chosen input coefficients $\boldsymbol{\alpha}$, 
it follows from~(\ref{eq:output3_condition}) that 
one can achieve the output equilibrium only within 
the smaller interval approximately given by $13.6 \le x_1^* \le 15$; 
therefore, even smaller uncertainties in the input coefficients
(challenge (U) from Section~\ref{sec:intro}) can then move
the system outside of this range, where the control fails.
In Figure~\ref{fig:nonlinear3}(a)--(c), we display the deterministic trajectories
for the species $X_1$, $X_3$ and $Y_2$ when the target equilibrium is given by
 $x_1^* = \beta_0/\gamma_1 = 5$, thus violating only the lower bound 
from~(\ref{eq:output3_condition}). One can notice
that a deterministic NEC occurs -  the target species $X_1$ fails to reach the desired equilibrium, while 
the residual species $X_3$ and the controlling species $Y_2$ blow-up; one can similarly
show that a stochastic NEC occurs. Analogous plots are shown in Figure~\ref{fig:nonlinear3}(d)--(f) 
when $x_1^* = \beta_0/\gamma_1 = 30$, violating the upper bound from~(\ref{eq:output3_condition});
one can notice that the species $X_2$ and $Y_1$ blow up, as in Figure~\ref{fig:nonlinear}.

\begin{figure}[!htbp]
\vskip  -2.0cm
\centerline{
\includegraphics[width=0.35\columnwidth]{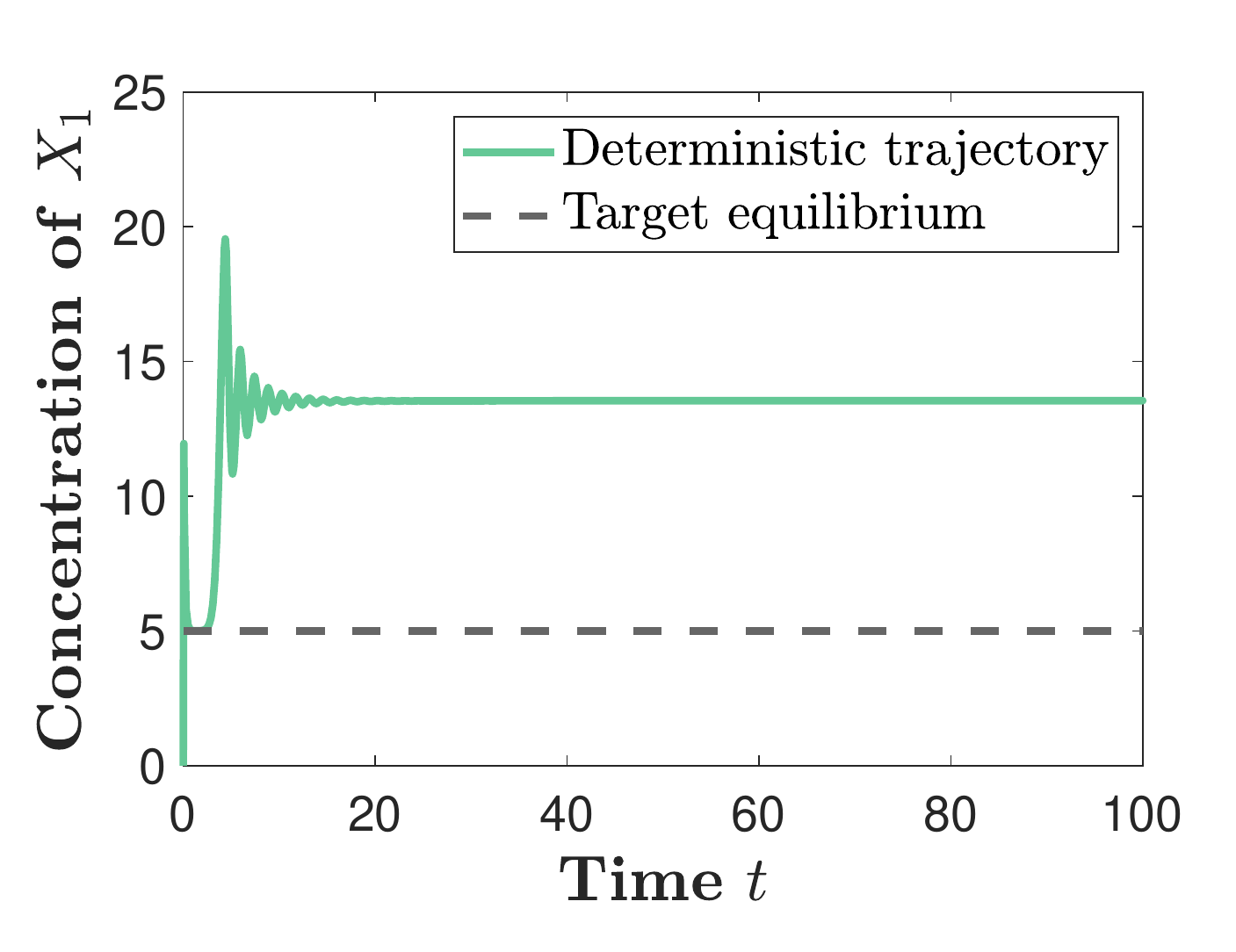}
\hskip 1mm
\includegraphics[width=0.35\columnwidth]{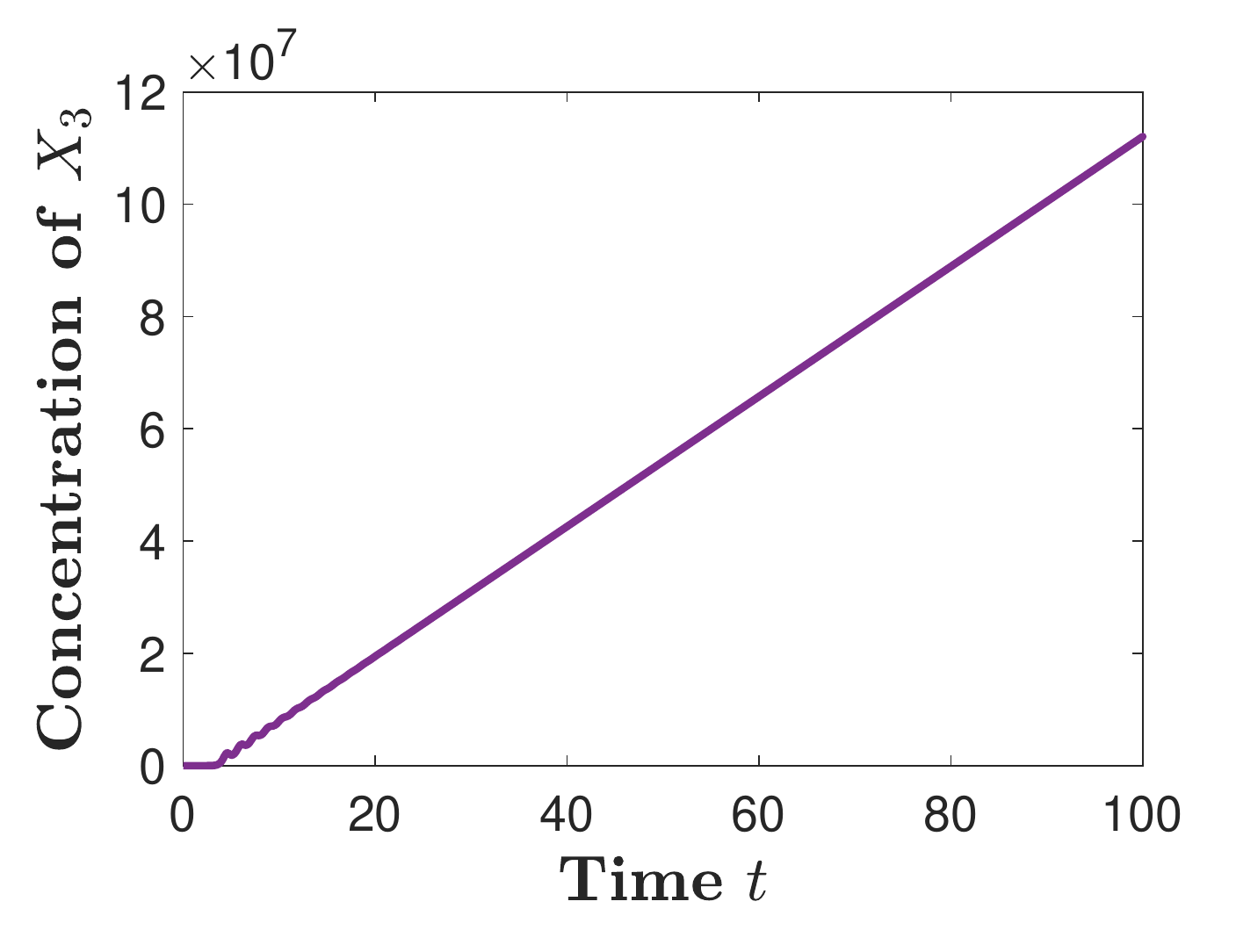}
\hskip 1mm
\includegraphics[width=0.35\columnwidth]{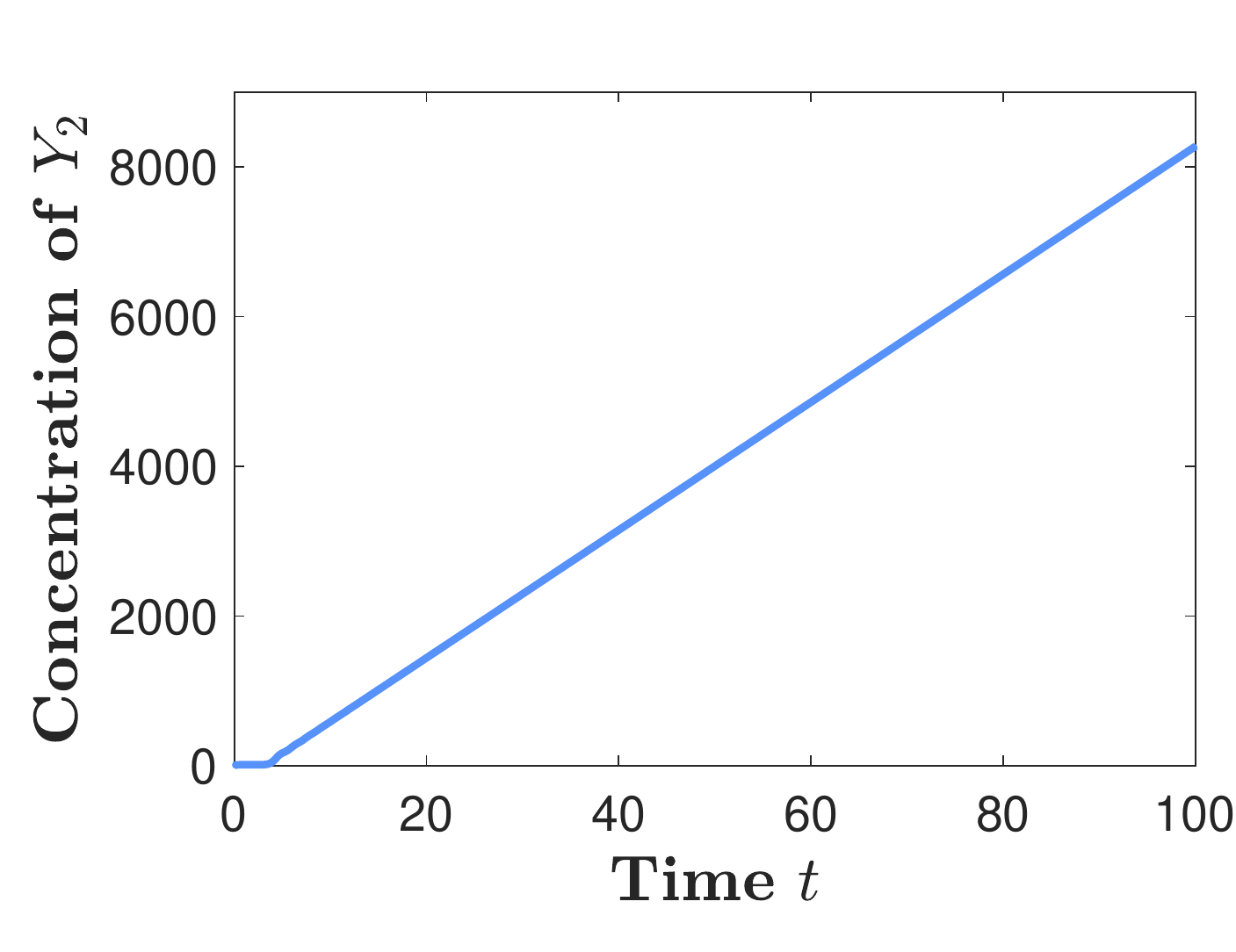}
}
\vskip -4.8cm
\leftline{\hskip -0.4cm (a) \hskip 5.5cm (b) \hskip 5.6cm (c)}
\vskip 4.2cm
\centerline{
\includegraphics[width=0.35\columnwidth]{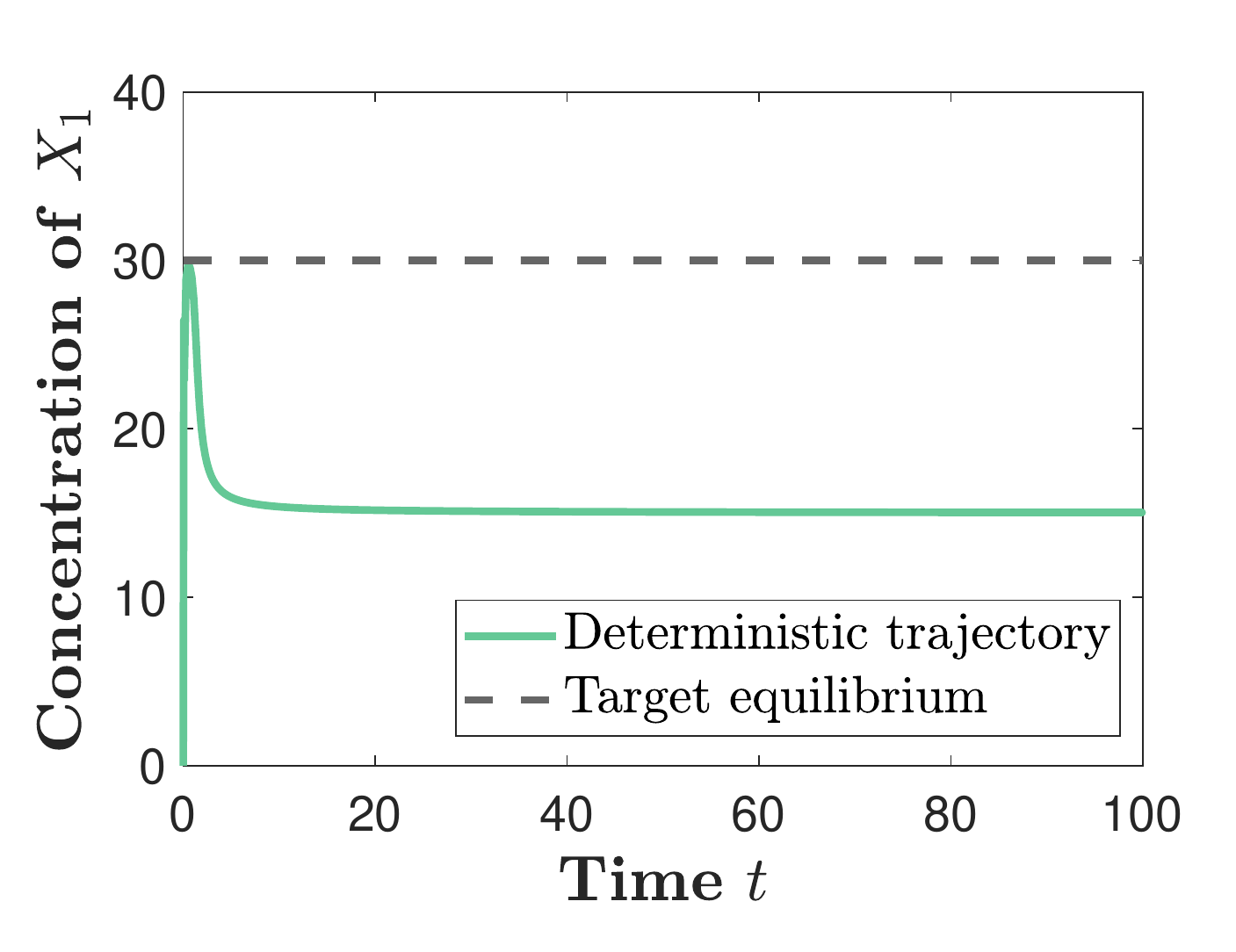}
\hskip 1mm
\includegraphics[width=0.35\columnwidth]{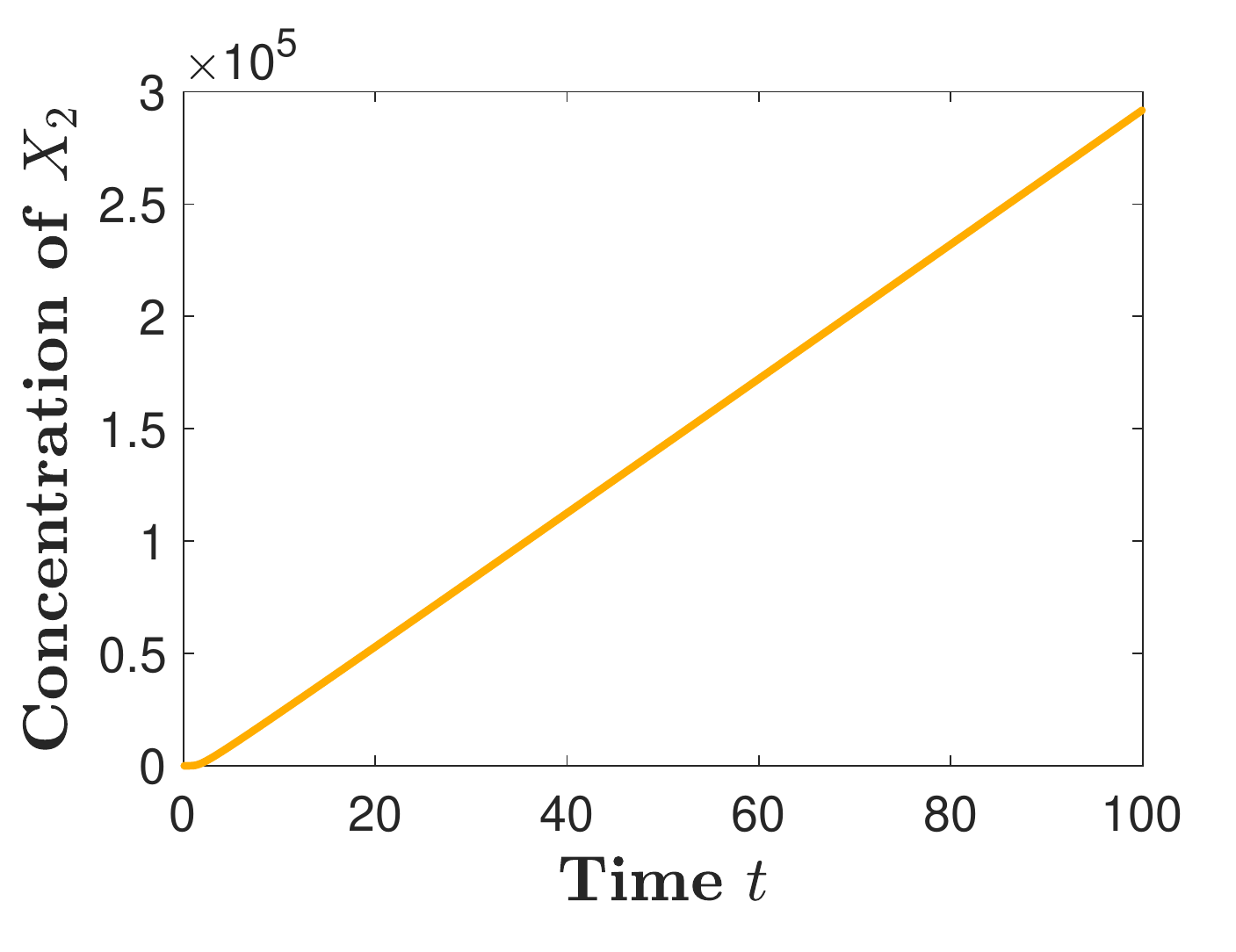}
\hskip 1mm
\includegraphics[width=0.35\columnwidth]{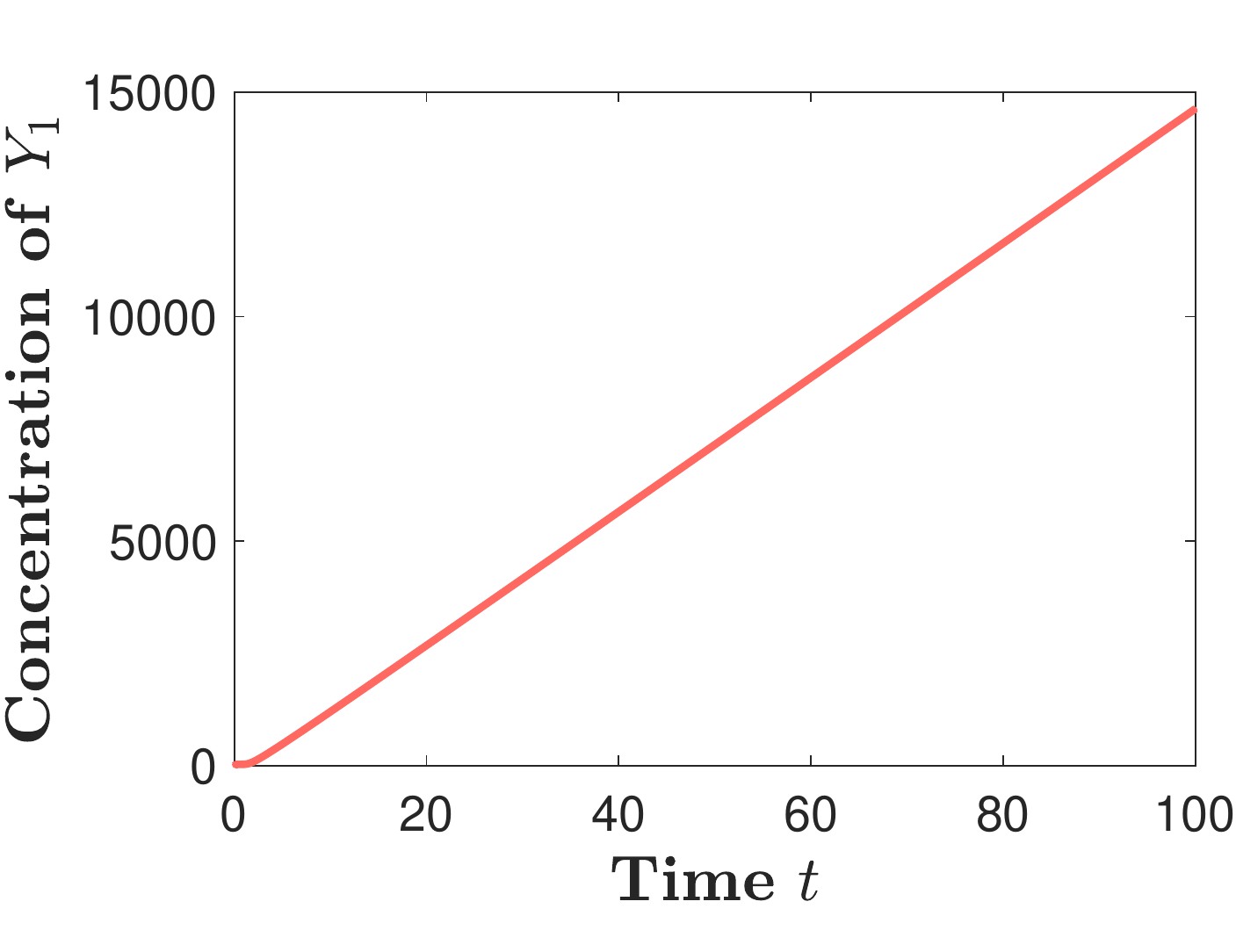}
}
\vskip -4.8cm
\leftline{\hskip -0.4cm (d) \hskip 5.5cm (e) \hskip 5.6cm (f)}
\vskip 3.9cm
\caption{\it{\emph{Application of the IFC~(\ref{eq:IFC_positive_negative}) on
the input network~(\ref{eq:input_3}) with rate coefficients
$(\alpha_0, \alpha_1, \alpha_2, \alpha_3, \alpha_4, \alpha_5, \alpha_6)
= (200, 1/7, 1/3,5,1,4,1)$ and $\varepsilon = 10^{-2}$}.
Panels~{\rm (a)}--{\rm (c)} display some of the deterministic trajectories for the output 
network~{\rm(\ref{eq:IFC_positive_negative})}$\cup${\rm(\ref{eq:input_3})}, 
when control coefficients are fixed to $(\beta_0, \beta_1,\gamma_1, \gamma_2, \gamma_3) 
= (50, 1,10,100,10)$. Analogous plots are shown in Panels~{\rm (d)}--{\rm (f)} 
when $(\beta_0, \beta_1,\gamma_1, \gamma_2, \gamma_3) = (300, 1,10,100,10)$.}}  \label{fig:nonlinear3}
\end{figure}  

\subsection{Arbitrary bimolecular input networks} \label{sec:nonnegativity}
One can continue the model-refinement process which led from
network~(\ref{eq:input_2}) to~(\ref{eq:input_3}), by including 
more auxiliary species $X_4, X_5, X_6, \ldots$, 
each of which generally introduces an additional constraint, $I_4, I_5, I_6, \ldots \ge 0$, 
which must be obeyed for an equilibrium to be nonnegative.
More generally, let $\mathcal{R}_{\alpha}$ be an arbitrary $N$-dimensional input network
satisfying properties (N), (HD) and (U) from Section~\ref{sec:intro},
$\mathcal{R}_{\beta, \gamma}$ an arbitrary $M$-dimensional IFC, 
and   $\mathcal{R}_{\alpha, \beta, \gamma} = \mathcal{R}_{\alpha} \cup \mathcal{R}_{\beta, \gamma}$
 the corresponding $(N + M)$-dimensional output network.
In order to ensure that an equilibrium $(\mathbf{x}^*, \mathbf{y}^*) \in \mathbb{R}^{N + M}$
of $\mathcal{R}_{\alpha, \beta, \gamma} $ is nonnegative, there are exactly two options. 

The first option is to choose appropriate values for the control coefficients
$\boldsymbol{\beta}$ and $\boldsymbol{\gamma}$.
However, the proportion of the state-space $\mathbb{R}^{N + M}$
occupied by the nonnegative orthant $\mathbb{R}_{\ge}^{N + M}$ 
is given by $2^{-(N + M)}$, which decreases exponentially as the dimension of the input
network $N$ increases - a fact known as the curse of dimensionality. 
Therefore, it is unlikely that an unguided choice of values for 
$\boldsymbol{\beta}$ and $\boldsymbol{\gamma}$ would achieve a
nonnegative equilibrium. In particular, the values of control coefficients
must be chosen so that the equilibria of 
\emph{all} of the $(N + M)$ species are nonnegative, which involves
solving a system of $(N + M)$ nonlinear inequalities 
with uncertain coefficients $\boldsymbol{\alpha}$
- an intractable theoretical problem. Furthermore, owing 
to a large number of inequalities, the allowed values for 
$\boldsymbol{\beta}$ and $\boldsymbol{\gamma}$ may be confined to smaller sets, which can
lead to larger parameter regime where IFCs can catastrophically fail. 
Such considerations have been demonstrated 
already for the two-dimensional network~(\ref{eq:input_2}) containing only one bimolecular reaction,
and the three-dimensional network~(\ref{eq:input_3}) containing two bimolecular reactions.

A necessary condition to bypass this intractable problem is to eliminate all of the residual invariant 
constraints $I_n, I_{n+1}, I_{n+2}, \ldots \ge 0$, which can be achieved only by eliminating all of the residual species. 
Therefore, the second option is to design a suitable controller that can be interfaced with \emph{all} of the
$N$ input species - an unfeasible experimental problem. However, for theoretical purposes, 
assume all of the input species are targetable; does there then exist an IFC
that ensures existence of a nonnegative equilibrium for any choice of the parameters 
$\boldsymbol{\alpha}$, $\boldsymbol{\beta}$ and $\boldsymbol{\gamma}$, 
thus mitigating challenge (U)?

\begin{theorem}\label{theorem:nonlinear}
Assume $\mathcal{R}_{\alpha}$ is an arbitrary mass-action input network
with $N$ input species, all of which are targetable. Then, there exists a bimolecular integral-feedback controller 
$\mathcal{R}_{\beta, \gamma}$, containing $2 N$ controlling species, 
such that the output network $\mathcal{R}_{\alpha, \beta, \gamma}
= \mathcal{R}_{\alpha} \cup \mathcal{R}_{\beta, \gamma}$ has a positive
equilibrium for all parameter values $(\boldsymbol{\alpha}, \boldsymbol{\beta},\boldsymbol{\gamma}) 
\in \mathbb{R}_{>}^{a + b + c}$.
\end{theorem}

\begin{proof}
See Appendix~\ref{app:biproof} for a constructive proof.
\end{proof}
\noindent In Appendix~\ref{app:biproof}, we design a controller that achieves this task
by generalizing the approach from Section~\ref{sec:nonlinear}, 
and apply the controller of the form~(\ref{eq:IFC_positive_negative}) to every input species; 
therefore, the dimension of the controller scales with the dimension of the input network.

\section{Discussion} \label{sec:discussion}
In this paper, we have demonstrated that molecular IFCs can display 
severe stability issues when applied to biochemical networks
subject to uncertainties. In particular, all nonnegative equilibria of the controlled 
network can vanish under IFCs, and some of the species abundances can then blow up. 
We call this hazardous phenomenon a \emph{negative-equilibrium catastrophe} (NEC).
In context of electro-mehanical systems, analogous phenomenon
is known as integrator windup~\cite{Control_theory} - equilibria
of the controlled system reach beyond the boundary of physically allowed values.
For some electro-mechanical systems, one only requires the equilibria to
be real (as opposed to complex); for biochemical systems,
one additionally requires that the equilibria are also nonnegative. 
Let us stress that requiring an equilibrium to be nonnegative
is significantly more restrictive than only requiring it is real. 
For example, while real \emph{linear} systems of equations generically have a unique real solution,
 finding parameter regimes where a positive solution exists is non-trivial~\cite{LinearPositive}; 
for \emph{nonlinear} systems, determining such parameter regimes is generally even more challenging, 
see Section~\ref{sec:nonnegativity}. 
The consequences of these issues, which are unavoidable for biochemical systems,
have been under-explored in the molecular control literature to date.

We have shown in Section~\ref{sec:linear} that, due to the nonnegativity constraint,
affine (unimolecular) biochemical systems cannot achieve integral control and, even worse, 
lead to catastrophes (NECs); in contrast, many affine electro-mechanical systems can achieve
integral control~\cite{Control_theory}.
In Section~\ref{sec:nonlinear}, using the theoretical framework from~\cite{Me_Homoclinic}, 
we have then constructed a family of bimolecular (nonlinear) IFCs~(\ref{eq:IFCnetapp}).
In Section~\ref{sec:first_order}, we have proved in Theorem~\ref{theorem:IFCsummary}
 that a particular two-dimensional (two-species) molecular IFC of the form~(\ref{eq:IFCnetapp})
ensures existence of a nonnegative equilibrium when applied to stable input unimolecular 
networks of arbitrary dimensions; in particular, NECs can be eliminated in a dimension-independent manner
 when unimolecular networks are controlled.
In contrast, in Section~\ref{sec:second_order}, we have demonstrated that
control of bimolecular networks suffers from the \emph{curse of dimensionality} - 
every species in the input network generally introduces a constraint which must be 
obeyed for a nonnegative equilibrium to exist, leading to an intractable problem.
For theoretical purposes, we have proved in Theorem~\ref{theorem:nonlinear} that, 
 assuming all of the input species are known and targetable - a generally 
experimentally unfeasible assumption, then there exists  a higher-dimensional 
IFC that always eliminates NECs. 
Let us note that, in all of the biochemical networks studied in this paper,
NECs simultaneously occur at both deterministic and stochastic levels. 
In particular, as opposed to the instability arising from bounded deterministic 
oscillations~\cite{Khammash,AIFC_1,AIFC_2},
which average out at the stochastic level, NECs generally persist in the stochastic setting.

Intracellular networks are in general bimolecular, higher-dimensional
and subject to uncertainties, as respectively described by the properties 
(N), (HD) and (U) in Section~\ref{sec:intro}. 
Due to these challenges, generally only reduced (approximate) models 
of intracellular networks are available, which are obtained by eliminating a number of the underlying auxiliary
coupled molecular species and reactions. The objective of these lower-dimensional reduced models is to 
capture the dynamics of desired intracellular species on a time-scale of interest~\cite{Pavliotis}. 
However, using reduced models for purpose of control is generally unjustified due to NECs, 
i.e. reduced models do not necessarily capture how the underlying extended models respond to control. 
In particular, it takes only one of the many input species to 
display a negative equilibrium for control to fail and a catastrophic event to unfold; 
hence, including a previously neglected molecular species into a successfully controlled reduced 
model can result in an extended model for which control fails - a phenomenon 
we call \emph{phantom control}, see Section~\ref{sec:second_order}. 
Let us note that, when a reduced model displaying a NEC
is extended to include e.g. finite resources (buffer) or dilution,
then some of the underlying species concentrations, instead of growing to infinity, 
would reach finite, but larger, values; nevertheless, the effects of unwanted large concentrations, 
such as sequestration of ribosomes and depletion of metabolites, are potentially very harmful.
NECs therefore place a fundamental limit to applicability of molecular IFCs in synthetic biology.
In particular, to avoid NECs, instead of a systematic approach, 
an ad-hoc approach is generally necessary, 
consisting of gathering detailed experimental information about a desired intracellular network 
and designing suitable higher-dimensional controllers that
 can be interfaced with larger number of appropriate input species.

\section*{Acknowledgements}
This work was supported by the EPSRC grant EP/P02596X/1.
Thomas E. Ouldridge would like to thank the Royal Society for
a University Research Fellowship.

\appendix 

\section{Appendix: Background} \label{app:background}
\emph{Notation}. Given sets $\mathcal{A}_1$ and $\mathcal{A}_2$, 
their union, intersection, and difference are denoted by $\mathcal{A}_1 \cup \mathcal{A}_2$, 
$\mathcal{A}_1 \cap \mathcal{A}_2$, and $\mathcal{A}_1 \setminus \mathcal{A}_2$, respectively.
The largest element of a set of numbers $\mathcal{A}$ is denoted by $\max \mathcal{A}$. 
The empty set is denoted by $\emptyset$. 
Set $\mathbb{Z}$ is the space of integer numbers, 
$\mathbb{Z}_{\ge}$ the space of nonnegative integer numbers, 
and $\mathbb{Z}_{>}$ the space of positive integer numbers. 
Similarly, $\mathbb{R}$ is the space of real numbers, 
$\mathbb{R}_{\ge}$ the space of nonnegative real numbers,
 and $\mathbb{R}_{>}$ the space of positive real numbers. 
Euclidean column vectors are denoted in boldface, 
$\mathbf{x} = (x_1, x_2, \ldots, x_N)^{\top} \in \mathbb{R}^{N} = \mathbb{R}^{N \times 1}$, 
where $\cdot^{\top}$ denotes the transpose operator. The $i$-standard basis
vector is denoted by $\mathbf{e}_i \equiv (\delta_{i,1}, \delta_{i,2}, \ldots, \delta_{i,N})^{\top} \in \mathbb{R}^N$, 
where $\delta_{i,j} = 0$ if $i \ne j$, and $\delta_{i,j} = 1$ if $i = j$.
The zero-vector is denoted by $\mathbf{0} \equiv (0, 0, \ldots, 0)^{\top} \in \mathbb{R}^N$, and we let
$(\mathbf{1} \equiv \sum_{i = 1}^N \mathbf{e}_i) \in \mathbb{R}^N$.
Given two Euclidean vectors $\mathbf{x}, \mathbf{y} \in \mathbb{R}^{N}$,
their inner-product is denoted by $\langle \mathbf{x}, \mathbf{y} \rangle 
\equiv \sum_{i = 1}^N x_i y_i$. Abusing the notation slightly, given two sequences
$u, v : \mathbb{Z}_{\ge} \to \mathbb{R}$, their inner-product is also denoted by
 $\langle \mathbf{u}(\mathbf{x}), \mathbf{v}(\mathbf{x}) \rangle 
\equiv \sum_{\mathbf{x} \in \mathbb{Z}_{\ge}} u(\mathbf{x}) v(\mathbf{x})$; 
we let $\|u(\mathbf{x}) \|_{l_1} \equiv \sum_{\mathbf{x} \in \mathbb{Z}_{\ge}} |u(\mathbf{x})|$ 
denote the $l^1$-norm of $u(\mathbf{x})$. 
Given a matrix $A \in \mathbb{R}^{N \times N}$, with $(i,j)$-element $\alpha_{i,j} \in \mathbb{R}$,
we denote the $i$-row and $j$-column of $A$ by
$\boldsymbol{\alpha}_{i, \cdot} \equiv (\alpha_{i,1}, \alpha_{i,2}, \ldots, \alpha_{i,N}) \in \mathbb{R}_{\ge}^{1 \times N}$, 
and $\boldsymbol{\alpha}_{\cdot, j} \equiv (\alpha_{1,j}, \alpha_{2,j}, \ldots, \alpha_{N,j})^{\top} \in \mathbb{R}_{\ge}^{N \times 1}$, respectively. 
The identity matrix is denoted by $I \equiv (\mathbf{e}_1, \mathbf{e}_2, \ldots, \mathbf{e}_N) \in \mathbb{R}^{N \times N}$, 
the zero matrix by $0 \equiv (\mathbf{0}, \mathbf{0}, \ldots, \mathbf{0}) \in \mathbb{R}^{N \times N}$, 
and diagonal matrices are denoted by $\textrm{diag}(\alpha_{1,1}, \alpha_{2,2}, \ldots, \alpha_{N,N}) 
\equiv (\alpha_{1,1} \mathbf{e}_1, \alpha_{2,2} \mathbf{e}_2, \ldots, \alpha_{N, N} \mathbf{e}_N)\in \mathbb{R}^{N \times N}$.
Determinant of a matrix $A$ is denoted by $|A|$.

\subsection{Biochemical reaction networks} \label{app:CRNs}
Let $\mathcal{R}_{\alpha} = \mathcal{R}_{\alpha}(\mathcal{X})$ be a reaction network
describing interactions, under mass-action kinetics, between $N$ biochemical species $\mathcal{X} = \{X_1, X_2, \ldots, X_N\}$
in a well-mixed unit-volume reactor~\cite{Feinberg}, as specified by the following $a$ reactions:
\begin{align}
\mathcal{R}_{\alpha}(\mathcal{X}): \; \; 
& & \sum_{l = 1}^N \nu_{j, l} X_l & \xrightarrow[]{\alpha_{j}} \sum_{l = 1}^N \bar{\nu}_{j, l} X_l,
\; \; \; \;
j \in \mathcal{A} = \{1, 2, \ldots, a\}.
\label{eq:networks}
\end{align}
Here, $\boldsymbol{\alpha} = (\alpha_1, \alpha_2, \ldots, \alpha_A) \in \mathbb{R}_{>}^a$
are the positive \emph{rate coefficients} of the reactions from $\mathcal{R}_{\alpha}$.
Nonnegative vectors $\boldsymbol{\nu}_{j, \cdot} = (\nu_{j, 1}, \nu_{j,2}, \ldots, \nu_{j,N})^{\top} \in \mathbb{Z}_{\ge}^N$
and  $\boldsymbol{\bar{\nu}}_{j, \cdot} = (\bar{\nu}_{j, 1}, \bar{\nu}_{j,2}, \ldots, \bar{\nu}_{j,N})^{\top} \in \mathbb{Z}_{\ge}^N$
are the \emph{reactant} and \emph{product stoichiometric coefficients} of the $j$-reaction, respectively;
if $\boldsymbol{\nu}_{j, \cdot} = \mathbf{0}$ (respectivelty, $\boldsymbol{\bar{\nu}}_{j, \cdot} = \mathbf{0}$), 
then the reactant (respectively, product) of the $j$-reaction is the \emph{null-species}, denoted by $\varnothing$, 
 representing species that are not explicitly modelled. 
When convenient, we denote two irreversible reactions 
$(\sum_{l = 1}^N \nu_{i, l} X_l   \xrightarrow[]{\alpha_{i}}  
\sum_{l = 1}^N \bar{\nu}_{i, l} X_l) \in \mathcal{R}_{\alpha}$
and
$(\sum_{l = 1}^N \bar{\nu}_{i, l} X_l  \xrightarrow[]{\alpha_{j}} 
\sum_{l = 1}^N \nu_{i, l} X_l) \in \mathcal{R}_{\alpha}$
 jointly as the single reversible reaction
$(\sum_{l = 1}^N \nu_{i, l} X_l \xrightleftharpoons[\alpha_j]{\alpha_i} 
\sum_{l = 1}^N \bar{\nu}_{i, l} X_l) \in\mathcal{R}_{\alpha}$.
Species $X_i$ is a \emph{catalyst} in the $j$-reaction from $\mathcal{R}_{\alpha}$ if $\nu_{j, i} = \bar{\nu}_{j, i} \ne 0$; 
if $X_i$ is a catalyst in all of the reaction from $\mathcal{R}_{\alpha}$, then
we write $\mathcal{R}_{\alpha} = \mathcal{R}_{\alpha}(\mathcal{X} \setminus X_i; \, X_i)$.
The \emph{order of the j-reaction} from network $\mathcal{R}_{\alpha}$
is given by $\langle \mathbf{1}, \boldsymbol{\nu}_{j, \cdot} \rangle \in \mathbb{Z}_{\ge}$.
The \emph{order of reaction network} $\mathcal{R}_{\alpha}$ is given by  
$\textrm{max} \{\langle \mathbf{1}, \boldsymbol{\nu}_{j, \cdot} \rangle | j \in \mathcal{A}\}$;
first-order (respectively, second-order)
 reaction networks are also said to be \emph{unimolecular} (respectively, \emph{bimolecular}). 

Given a class of biochemical reaction networks, parametrized by the underlying rate 
coefficients, it may be of interest if a given property is likely to be true 
when all admissible values of the rate 
coefficients are considered, which motivates the following definition.
In what follows, we implicitly use Lebesgue measure for sets. 

\begin{definition} [\textbf{Genericity}]  \label{def:generic} 
Consider a mass-action reaction network $\mathcal{R}_{\alpha}$
parametrized by the rate coefficients $\boldsymbol{\alpha} \in \mathbb{S}_{\alpha}$, where
$\mathbb{S}_{\alpha} \subset \mathbb{R}_{>}^a$ is a nonempty open set.
Assume $\mathbb{S}_{\alpha}$ is partitioned according to 
$\mathbb{S}_{\alpha} = \Omega_{\alpha} \cup \omega_{\alpha}$, 
with $\Omega_{\alpha} \cap \omega_{\alpha} = \emptyset$, 
where $\omega_{\alpha}$ is a set of measure zero.
A property is said to be \emph{generic} 
for the set $\mathbb{S}_{\alpha}$ and network $\mathcal{R}_{\alpha}$ 
if it holds for all $\boldsymbol{\alpha} \in \Omega_{\alpha}$ 
and fails to hold for all $\boldsymbol{\alpha} \in \omega_{\alpha}$.
\end{definition}
\noindent \emph{Example}. Empy set $\omega_{\alpha} = \emptyset$, and
set $\omega_{\alpha} = \{\boldsymbol{\alpha}_1, \boldsymbol{\alpha}_2, \ldots\}$,
containing finitely or countably infinitely many points, have measure zero. 
All the points $\boldsymbol{\alpha}$ where a non-trivial polynomial 
$\mathcal{P}(\boldsymbol{\alpha})$ vanishes is also a set of measure zero~\cite{AG}; 
for example, given a nonzero matrix $A = A(\boldsymbol{\alpha})$, with $(i,j)$-element $\alpha_{i,j}$, 
the set of all points $\boldsymbol{\alpha}$ such that $|A(\boldsymbol{\alpha})| = 0$ has zero measure.

\subsection{Dynamical models of reaction networks} \label{app:dynamics}
In what follows, we present deterministic and stochastic models
of mass-action reaction networks, and provide definitions in context of blow-ups.  

\subsubsection{Deterministic model} \label{app:deterministic}
Let $\mathbf{x}(t; \, \boldsymbol{\alpha}) = 
(x_1(t; \, \boldsymbol{\alpha}), x_2(t; \, \boldsymbol{\alpha}), \ldots, x_N(t; \, \boldsymbol{\alpha}))^{\top} \in \mathbb{R}_{\ge}^N$
be a concentration vector at time $t \in \mathbb{R}_{\ge}$ for the species $\mathcal{X} = \{X_1, X_2, \ldots, X_N\}$
from the network $\mathcal{R}_{\alpha}$, given by~(\ref{eq:networks}).
A deterministic model of the reaction network $\mathcal{R}_{\alpha}$ describes 
the time-evolution of $\mathbf{x} = \mathbf{x}(t; \, \boldsymbol{\alpha})$ as
a system of first-order ordinary differential equations (ODEs), called the 
\emph{reaction-rate equations} (RREs)~\cite{Feinberg,RadekBook}, given by
\begin{align}
\frac{\mathrm{d} \mathbf{x}}{\mathrm{d}  t} & = 
\boldsymbol{\mathcal{K}}(\mathbf{x}; \, \boldsymbol{\alpha}) =  \sum_{j \in \mathcal{A}}
\alpha_j \Delta \mathbf{x}_{j, \cdot} \mathbf{x}^{\boldsymbol{\nu}_{j, \cdot}}, \label{eq:RREs}
\end{align}
where $\Delta \mathbf{x}_{j, \cdot} \equiv 
(\bar{\boldsymbol{\nu}}_{j, \cdot} - \boldsymbol{\nu}_{j, \cdot}) \in \mathbb{Z}^N$
is the \emph{reaction vector} of the $j$-reaction, and
$\mathbf{x}^{\boldsymbol{\nu_{j,\cdot}}} \equiv \prod_{i = 1}^N x_i^{\nu_{j,i}}$ with $0^0 \equiv 1$. 
Function $\boldsymbol{\mathcal{K}}(\cdot; \, \boldsymbol{\alpha}) : \mathbb{R}^N \to \mathbb{R}^N$, 
called a \emph{kinetic function} (see also Appendix~\ref{app:kinetictrans}),
is a polynomial of degree $m = \textrm{max} \{\langle \mathbf{1}, \boldsymbol{\nu}_{j, \cdot} \rangle | 
j \in \mathcal{A}\}$ in $\mathbf{x}$, which we denote by
$\boldsymbol{\mathcal{K}}(\mathbf{x}; \, \boldsymbol{\alpha}) \in 
\mathbb{P}_{m}(\mathbb{R}^N; \, \mathbb{R}^N)$.
Vector $\mathbf{x}^* = \mathbf{x}^*(\boldsymbol{\alpha}) < \boldsymbol{\infty}$ 
is called an \emph{equilibrium} of the RREs~(\ref{eq:RREs})
if $\boldsymbol{\mathcal{K}}(\mathbf{x}^*; \, \boldsymbol{\alpha}) = \mathbf{0}$.

We now present an important property of the RREs~\cite{Me_Homoclinic}.
\begin{theorem} \label{theorem:trapping}
The nonnegative orthant $\mathbb{R}_{\ge}^N$ is an invariant 
set for the {\rm ODE}s~{\rm (\ref{eq:RREs})}.
\end{theorem}
\begin{proof}
See~\cite{Me_Homoclinic}.
\end{proof}

In this paper, unimolecular networks are of interest and, to this end, we introduce
the following definition. 

\begin{definition} [\textbf{Cross-nonnegative matrix}]  \label{def:crossmatrix} 
 A matrix $A \in \mathbb{R}^{N \times N}$ with nonnegative off-diagonal elements, 
$\alpha_{i,j} \ge 0$ for all $i,j \in \{1,2, \ldots, N\}$ such that $i \ne j$, 
 is said to be \emph{cross-nonnegative}.
\end{definition}
\noindent \emph{Remark}. Cross-nonnegative matrices are known as 
negative $Z$, quasi-positive, essentially nonnegative, and Metzler matrices
 in the literature~\cite{Mmatrix,PositiveSystems}. 

In context of unimolecular reaction networks, Theorem~\ref{theorem:trapping} implies the following corollary.

\begin{corollary} \label{theorem:trapping_firstorder}
For unimolecular reaction networks, the kinetic function from the {\rm RRE}s~{\rm (\ref{eq:RREs})} 
is given by $\boldsymbol{\mathcal{K}}(\mathbf{x}; \, \boldsymbol{\alpha}) 
= (\boldsymbol{\alpha}_{\cdot,0} + A \mathbf{x})$, where $\boldsymbol{\alpha}_{\cdot,0} 
= (\alpha_{1,0}, \alpha_{2,0}, \ldots, \alpha_{N,0}) \in \mathbb{R}_{\ge}^N$ is a nonnegative vector,
while $A \in \mathbb{R}^{N \times N}$ is a cross-nonnegative matrix. 
\end{corollary}
A linear ODE system $\mathrm{d}\mathbf{x}/\mathrm{d} t
= (\boldsymbol{\alpha}_{\cdot,0} + A \mathbf{x})$ is said to be \emph{asymptotically stable}
(also simply referred to as \emph{stable} in this paper) if all the eigenvalues of $A$ have negative real parts. 

\subsubsection{Stochastic model} \label{app:stochastic}
Let $\mathbf{X}(t; \boldsymbol{\alpha}) =
 (X_1(t; \boldsymbol{\alpha}), X_2(t; \boldsymbol{\alpha}), \ldots, X_N(t; \boldsymbol{\alpha}))^{\top} \in \mathbb{Z}_{\ge}^N$
be a copy-number vector at time $t \in \mathbb{R}_{\ge}$ for the species $\mathcal{X} = \{X_1, X_2, \ldots, X_N\}$
from the network $\mathcal{R}_{\alpha}$, given by~(\ref{eq:networks}); 
abusing the notation slightly, we denote the points in the state-space for $\mathbf{X}(t; \boldsymbol{\alpha})$ 
using the same symbol as the concentration vector from~(\ref{eq:RREs}), i.e.
by $\mathbf{x} = (x_1, x_2, \ldots, x_N)^{\top} \in \mathbb{Z}_{\ge}^N$.
A stochastic model of the reaction network $\mathcal{R}_{\alpha}$ describes 
the time-evolution of $\mathbf{X} = \mathbf{X}(t; \boldsymbol{\alpha})$ as a continuous-time discrete-space 
Markov chain~\cite{GillespieDerivation} characterized
via a difference-operator $\mathcal{L}_{\alpha}$, called the \emph{generator}, given by~\cite{VanKampen}
\begin{align}
\mathcal{L}_{\alpha} u(\mathbf{x}) & =  
\sum_{j = 1}^A \alpha_{j} \mathbf{x}^{\underline{\boldsymbol{\nu}_{j,\cdot}}}
 (E_{\mathbf{x}}^{+\Delta \mathbf{\mathbf{x}}_{j,\cdot}} - 1) 
u(\mathbf{x}), \label{eq:generator}
\end{align}
where $u : \mathbb{Z}_{\ge}^N \to \mathbb{R}$ belongs to a suitable function space.
Here, $\Delta \mathbf{x}_{j, \cdot} = (\bar{\boldsymbol{\nu}}_{j, \cdot} - \boldsymbol{\nu}_{j, \cdot}) \in \mathbb{Z}^N$
is the $j$-reaction vector, while
$\mathbf{x}^{\underline{\boldsymbol{\nu}_{j,\cdot}}} = \prod_{i = 1}^N x_i^{\underline{\nu_{j,i}}}$, 
with $x_i^{\underline{\nu_{j,i}}} = x_i (x_i - 1) \ldots (x_i - \nu_{j,i} - 1)$
and $x^{\underline{0}} \equiv 1$ for all $x \in \mathbb{Z}_{\ge}$. 
Furthermore, $E_{\mathbf{x}}^{+\Delta \mathbf{x}_{j,\cdot}} = \prod_{i = 1}^N E_{x_i}^{+\Delta x_{j,i}}$ is a step-operator such that 
$E_{\mathbf{x}}^{+\Delta \mathbf{x}_{j,\cdot}} u(\mathbf{x}) = u(\mathbf{x} + \Delta \mathbf{x}_{j,\cdot})$. 

Let $p(\cdot,t; \, \boldsymbol{\alpha}) : \mathbb{Z}_{\ge}^N \to [0,1]$ be the probability-mass function (PMF) 
at time $t$ of the Markov chain with generator~(\ref{eq:generator}), and let
 $f_{\mathbf{x}} : \mathbb{Z}_{\ge}^N \to \mathbb{R}$ be a suitable function of the species copy-numbers.
We let $\mathbb{E} f_{\mathbf{x}} = \mathbb{E} f_{\mathbf{x}}(t; \, \boldsymbol{\alpha})
\equiv \langle f_{\mathbf{x}}(\mathbf{x}), p(\mathbf{x}, t; \, \boldsymbol{\alpha}) \rangle$
denote the expectation of $f_{\mathbf{x}}$ at time $t$ with respect to $p(\mathbf{x},t; \, \boldsymbol{\alpha})$.
In this paper, we focus on the average species copy-numbers, i.e. on the
 first-moment vector $\mathbb{E} \mathbf{X} = \mathbb{E} \mathbf{X}(t; \, \boldsymbol{\alpha})
= (\mathbb{E} X_1(t; \, \boldsymbol{\alpha}), \mathbb{E} X_2(t; \, \boldsymbol{\alpha}), \ldots,
 \mathbb{E} X_N(t; \, \boldsymbol{\alpha}))^{\top} \in \mathbb{R}_{\ge}^N$,
which evolves in time according to the ODEs~\cite{RadekBook,VanKampen}
\begin{align}
\frac{\mathrm{d} \mathbb{E} \mathbf{X}}{\mathrm{d}  t} & = 
\mathbb{E} [\mathcal{L}_{\alpha} \mathbf{X}] =
 \sum_{j \in \mathcal{A}} \alpha_j  \Delta \mathbf{x}_{j, \cdot} 
\mathbb{E} \mathbf{X}^{\underline{\boldsymbol{\nu}_{j,\cdot}}}. \label{eq:average}
\end{align}
In the special case of unimolecular reaction networks, the first-moment 
equations~(\ref{eq:average}) and the RREs~(\ref{eq:RREs})
are formally equivalent; more generally, 
the less-detailed deterministic and the more-detailed stochastic models
match in the thermodynamic limit~\cite{Kurtz}.

\subsubsection{Blow-up and negative-equilibrium catastrophe} \label{app:negblowup}
In this paper, we focus on the circumstances when some of the species abundance experiences
an unbounded growth, and introduce the following definition for this purpose. 
\begin{definition} [\textbf{Blow-up}]  \label{def:blowup} 
Reaction network $\mathcal{R}_{\alpha}(\mathcal{X})$ is said to \emph{blow up deterministically} 
for a given initial condition if $\lim_{t \to \infty} x_i(t; \, \boldsymbol{\alpha}) = \infty$ for some $i \in \{1, 2, \ldots, N\}$, 
where the species concentration $\mathbf{x}(t; \, \boldsymbol{\alpha}) \in \mathbb{R}_{\ge}^N$
satisfies~{\rm (\ref{eq:RREs})}; $\mathcal{R}_{\alpha}(\mathcal{X})$ is said to \emph{blow up stochastically} 
for a given initial condition if $\lim_{t \to \infty} \mathbb{E} X_i(t; \, \boldsymbol{\alpha}) 
= \infty$ for some $i \in \{1, 2, \ldots, N\}$, where the first-moment of the species copy-numbers
 $\mathbb{E} \mathbf{X}(t; \, \boldsymbol{\alpha})  \in \mathbb{R}_{\ge}^N$ satisfies~{\rm (\ref{eq:average})}.
\end{definition}
Nonnegative ODEs, such as equations~(\ref{eq:RREs}) and~(\ref{eq:average}),
need not have a nonnegative time-independent solution, i.e. the dynamics can be confined to 
an unbounded invariant set devoid of any equilibria. In this context, we introduce the
following definition. 

\begin{definition} [\textbf{Negative-equilibrium catastrophes (NECs)}]  \label{def:catastrophe} 
Reaction network $\mathcal{R}_{\alpha}(\mathcal{X})$ is said to display
a deterministic (respectively, a stochastic) \emph{negative-equilibrium catastrophe} ({\rm NEC})
if, for all $\boldsymbol{\alpha} \in \mathbb{R}_{>}^a$ such that 
the {\rm RRE}s have no nonnegative equilibria, $\mathcal{R}_{\alpha}(\mathcal{X})$
blows up deterministically (respectively, stochastically) for some nonnegative initial conditions.
\end{definition}
\noindent \emph{Remark}. A nonnegative equilibrium can cease to exist 
by attaining a negative component, becoming complex, or vanishing all together.

\section{Appendix: Stochastic biochemical control} \label{app:biochemical_control}
In this section, we formulate the problem of achieving biochemical control over a given reaction network, 
starting with the following definition. 
\begin{definition} [\textbf{Black, grey and white box}]  \label{def:blackbox} 
Network $\mathcal{R}_{\alpha}(\mathcal{X}) \ne \emptyset$
with unknown (respectively, only partially known) 
structure and dynamics is called a \emph{black-box} (respectively, \emph{grey-box}) network; 
$\mathcal{R}_{\alpha}(\mathcal{X}) \ne \emptyset$ is called a \emph{white-box} network
if its structure and dynamics are completely known.
\end{definition}
Given a black- or grey-box \emph{input} (uncontrolled) reaction network $\mathcal{R}_{\alpha} = \mathcal{R}_{\alpha}(\mathcal{X})$,
the objective of biochemical control is to design a \emph{controller} network $\mathcal{R}_{\beta,\gamma}$
in order to ensure that the dynamics of desired input species $\mathcal{X}$
is suitably controlled in the resulting \emph{output} (controlled) network
$\mathcal{R}_{\alpha,\beta,\gamma} \equiv \mathcal{R}_{\alpha} \cup \mathcal{R}_{\beta,\gamma}$.
To this end, we partition the input species into $\mathcal{X} = \mathcal{X}_{\tau} \cup \mathcal{X}_{\rho}$, 
where $\mathcal{X}_{\tau}= \{X_1, X_2, \ldots, X_{N_{\tau}}\}$ are the $1 \le N_{\tau} \le N$ \emph{target} species 
that can be interfaced with a given controller,
while $\mathcal{X}_{\rho} = \mathcal{X} \setminus \mathcal{X}_{\tau} = \{X_{N_{\tau}+1}, X_{N_{\tau}+2}, \ldots, X_N\} $
are the $N_{\rho} = (N - N_{\tau})$ \emph{residual} species that cannot be interfaced with the controller. 
The controller can be decomposed into two sub-networks, $\mathcal{R}_{\beta,\gamma} = 
\mathcal{R}_{\beta,\gamma}(\mathcal{X}_{\tau},\mathcal{Y})
= \mathcal{R}_{\beta}(\mathcal{Y}) \cup \mathcal{R}_{\gamma}(\mathcal{X}_{\tau},\mathcal{Y})$,
where $\mathcal{R}_{\beta} = \mathcal{R}_{\beta}(\mathcal{Y})$, called the \emph{core}, 
contains all the reactions that involve only the \emph{controlling} species $\mathcal{Y} = \{Y_1, Y_2, \ldots, Y_M\}$, 
while $\mathcal{R}_{\gamma} = \mathcal{R}_{\gamma}(\mathcal{X}_{\tau},\mathcal{Y})$, 
called the \emph{interface}, contains all of the remaining reactions, involving both $\mathcal{X}_{\tau}$ and $\mathcal{Y}$.
Put more simply, the core $\mathcal{R}_{\beta}$ describes internal dynamics of the controlling species, 
while the interface $\mathcal{R}_{\gamma}$ describes how the target and controlling species interact.  

In what follows, we focus on controlling average copy-number of a single target species;  see~\cite{Me_Morphing}
for a more general multi-species control of the full PMF of both target and residual species.
To this end, we denote the rate coefficients from the sub-networks 
$\mathcal{R}_{\alpha}$, $\mathcal{R}_{\beta}$ and $\mathcal{R}_{\gamma}$
 by $\boldsymbol{\alpha} \in \mathbb{R}_{>}^a$,  
$\boldsymbol{\beta} \in \mathbb{R}_{>}^b$ and $\boldsymbol{\gamma} \in \mathbb{R}_{>}^c$, respectively.
We also let $\mathbb{E} X_1(t; \, \cdot, \cdot \, ; \, \boldsymbol{\alpha}) 
: \mathbb{R}_{>}^b \times \mathbb{R}_{>}^c \to \mathbb{R}$
be the first-moment of the target species $X_1 \in \mathcal{X}_{\tau}$, 
which is a function of the control parameters $(\boldsymbol{\beta},\boldsymbol{\gamma})$ 
for every time $t$, input coefficient $\boldsymbol{\alpha}$
and every initial condition. More precisely, 
$\mathbb{E} X_1 = \mathbb{E} X_1(t; \, \boldsymbol{\beta}, \boldsymbol{\gamma}; \, \boldsymbol{\alpha})
= \langle x_1, p(\mathbf{x}, \mathbf{y}, t; \, \boldsymbol{\beta},\boldsymbol{\gamma}; \, \boldsymbol{\alpha}) \rangle$, 
where $p(\mathbf{x}, \mathbf{y}, t; \, \boldsymbol{\beta},\boldsymbol{\gamma}; \, \boldsymbol{\alpha})$ is the time-dependent
PMF of the output network. In what follows, we denote the gradient operator with 
respect to $\mathbf{x}= (x_1, x_2, \ldots, x_N)$ by
$\boldsymbol{\nabla}_{\mathbf{x}} \equiv (\partial/\partial_{x_1}, 
\partial/\partial_{x_2}, \ldots, \partial/\partial_{x_N})$. 

\begin{definition} [\textbf{Control}]  \label{def:control} 
Consider a black-box input network $\mathcal{R}_{\alpha}(\mathcal{X})$,
and the corresponding output network 
$\mathcal{R}_{\alpha,\beta,\gamma}(\mathcal{X},\mathcal{Y})
= \mathcal{R}_{\alpha}(\mathcal{X}) \cup 
\mathcal{R}_{\beta,\gamma}(\mathcal{X}_{\tau},\mathcal{Y})$. 
Assume we are given a nonempty open set $\mathbb{S}_{\alpha} \subset \mathbb{R}_{>}^a$,
stability index $p \in \mathbb{Z}_{>}$, 
and a target value $\bar{x}_1 \in \mathbb{R}_{>}$
together with a tolerance $\varepsilon \in \mathbb{R}_{\ge}$. 
Then, the first-moment $\mathbb{E} X_1 = 
\mathbb{E} X_1(t; \, \boldsymbol{\beta}, \boldsymbol{\gamma}; \, \boldsymbol{\alpha})$
 is said to be \emph{controlled} in the long-run if there exists a nonempty open set
$\mathbb{S}_{\beta,\gamma} \subset \mathbb{R}_{>}^{b+c}$
such that the following two conditions are satisfied
for all initial conditions $(\mathbf{X}(0), \mathbf{Y}(0)) \in \mathbb{Z}_{\ge}^{N+M}$
and rate coefficients $(\boldsymbol{\alpha}, \boldsymbol{\beta}, \boldsymbol{\gamma}) 
\in \mathbb{S}_{\alpha} \times \mathbb{S}_{\beta,\gamma}$:
\begin{enumerate}
\item[{\rm (C.I)}] \textbf{Stability}: 
$\{\lim_{t \to \infty} \mathbb{E} X_i^p(t; \, \boldsymbol{\beta}, \boldsymbol{\gamma}; \, \boldsymbol{\alpha}) < \infty\}_{i = 1}^N$ 
and $\{\lim_{t \to \infty} \mathbb{E} Y_i^p(t; \, \boldsymbol{\beta},\boldsymbol{\gamma}; \, \boldsymbol{\alpha}) < \infty\}_{i = 1}^M$.
\item[{\rm (C.II)}] \textbf{Accuracy}: 
$\lim_{t \to \infty} \Big |\mathbb{E} X_1(t; \, \boldsymbol{\beta}, \boldsymbol{\gamma}; \, \boldsymbol{\alpha})
 - \bar{x}_1 \Big| \le \varepsilon$, where 
$\lim_{t \to \infty} \boldsymbol{\nabla}_{\boldsymbol{\beta}, \boldsymbol{\gamma}} 
\mathbb{E} X_1(\boldsymbol{\beta},\boldsymbol{\gamma}; \, \boldsymbol{\alpha}) \neq \mathbf{0}$.
\end{enumerate}
\end{definition}
Condition~(C.I) requires boundedness of the long-time moments, up to order $p > 1$,
 of \emph{all} of the species from the output network $\mathcal{R}_{\alpha,\beta,\gamma}(\mathcal{X},\mathcal{Y})$. 
Minimally, one requires that the long-time first-moments are bounded, i.e.
that the controller does not trigger a stochastic blow-up (see Definition~\ref{def:blowup}).
The larger $p \ge 1$ one chooses, the thinner the tail of the long-time output PMF, 
which ensures a more stable behavior of the output network. 
Condition~(C.II) demands that the long-time average of $X_1$, 
which is required to depend on at least one control parameter, is sufficiently close to the target value.
Let us remark that~(C.I)--(C.II) must hold within neighborhoods of the underlying rate coefficient values, 
reflecting the fact that measurement and fine-tuning of rate coefficients is not error-free. 
For the same reason, we have put forward a more relaxed accuracy criterion by allowing nonzero tolerance in condition~(C.II).

\subsection{Robust control} \label{app:linearIC}
Condition (C.II) from Definition~\ref{def:control} may be challenging to 
achieve due to the fact that the long-time first-moment $\mathbb{E} X_1$ 
generally depends on the initial conditions, and on the input coefficients $\boldsymbol{\alpha}$, 
which motivates the following definition.

\begin{definition}[\textbf{Robustness}]\label{def:robust}
Consider a black-box input network $\mathcal{R}_{\alpha}(\mathcal{X})$,
and the corresponding output network $\mathcal{R}_{\alpha, \beta, \gamma}(\mathcal{X}, \mathcal{Y})
 = \mathcal{R}_{\alpha}(\mathcal{X}) \cup 
\mathcal{R}_{\beta,\gamma}(\mathcal{X}_{\tau},\mathcal{Y})$. 
Assume conditions {\rm (C.I)--(C.II)} from {\rm Definition~\ref{def:control}} are satisfied. 
Then, network $\mathcal{R}_{\beta,\gamma}(\mathcal{X}_{\tau},\mathcal{Y})$
is a \emph{robust} controller of the first-moment $\mathbb{E} X_1$ in the long-run if the
 following two conditions are also satisfied:
\begin{enumerate}
\item[\emph{(R.I)}] \textbf{Robustness to initial conditions}.
There exists a unique stationary {\rm PMF} $p(\mathbf{x}, \mathbf{y}; \, 
\boldsymbol{\beta},\boldsymbol{\gamma}; \, \boldsymbol{\alpha})$ such that
 $\lim_{t \to \infty} \| p(\mathbf{x}, \mathbf{y}, t; \, \boldsymbol{\beta},\boldsymbol{\gamma}; \, \boldsymbol{\alpha})
- p(\mathbf{x}, \mathbf{y}; \, \boldsymbol{\beta},\boldsymbol{\gamma}; \, \boldsymbol{\alpha}) \|_{l_1} = 0$
for all $(\mathbf{X}(0), \mathbf{Y}(0)) \in \mathbb{Z}_{\ge}^{N+M}$
and $(\boldsymbol{\alpha}, \boldsymbol{\beta}, \boldsymbol{\gamma}) 
\in \mathbb{S}_{\alpha} \times \mathbb{S}_{\beta,\gamma}$.
\item[\emph{(R.II)}] \textbf{Robustness to input coefficients}. The stationary first-moment,
given by $\mathbb{E} X_1^*(\boldsymbol{\beta},\boldsymbol{\gamma}; \, \boldsymbol{\alpha})$
$\equiv \langle x_1,  p(\mathbf{x}, \mathbf{y}; \, \boldsymbol{\beta},\boldsymbol{\gamma}; \, \boldsymbol{\alpha}) > 0$,
satisfies $\boldsymbol{\nabla}_{\boldsymbol{\alpha}} 
\mathbb{E} X_1^*(\boldsymbol{\beta},\boldsymbol{\gamma}; \, \boldsymbol{\alpha}) = \mathbf{0}$
for all $(\boldsymbol{\alpha}, \boldsymbol{\beta}, \boldsymbol{\gamma}) 
\in \mathbb{S}_{\alpha} \times \mathbb{S}_{\beta,\gamma}$.
\end{enumerate}
\end{definition}
\noindent \emph{Remark}. Robust controllers are also called
\emph{integral-feedback} controllers~\cite{Control_theory}.

\noindent \emph{Remark}. Conditions (C.II) and  (R.II) demand that the first-moment of $X_1$ 
is nondegenerate with respect to the control parameters, 
$\boldsymbol{\nabla}_{\boldsymbol{\beta}, \boldsymbol{\gamma}} 
\mathbb{E} X_1^*(\boldsymbol{\beta},\boldsymbol{\gamma}; \, \boldsymbol{\alpha})\neq \mathbf{0}$,
 and that it is degenerate with respect to the input parameters,
$\boldsymbol{\nabla}_{\boldsymbol{\alpha}} 
\mathbb{E} X_1^*(\boldsymbol{\beta},\boldsymbol{\gamma}; \, \boldsymbol{\alpha}) = \mathbf{0}$, 
respectively. As we prove in Lemma~\ref{lemma:linear_intfed} below, 
the degeneracy condition enforces singularity of an appropriate kinetic matrix.

\section{Appendix: Nonexistence of unimolecular integral-feedback controllers} \label{app:nonexistence} 
Let $\mathcal{R}_{\alpha}(\mathcal{X})$ be a black-box input network
with a desired target species $X_1 \in \mathcal{X}_{\tau}$,
and let $\mathcal{R}_{\beta,\gamma}(\mathcal{X}_{\tau},\mathcal{Y})
= \mathcal{R}_{\beta}(\mathcal{Y}) \cup \mathcal{R}_{\gamma}(\mathcal{X}_{\tau},\mathcal{Y})$
be a unimolecular network. The first-moment equations for the 
output network $\mathcal{R}_{\alpha, \beta, \gamma}(\mathcal{X}, \mathcal{Y}) 
= \mathcal{R}_{\alpha}(\mathcal{X})
 \cup \mathcal{R}_{\beta,\gamma}(\mathcal{X}_{\tau},\mathcal{Y})$  can be written in the following form:
\begin{align}
\frac{\mathrm{d} \mathbb{E} \mathbf{X}_{\rho}}{\mathrm{d}  t} & = 
\mathbb{E} [\mathcal{L}_{\alpha} \mathbf{X}_{\rho}], \nonumber \\
\frac{\mathrm{d} \mathbb{E} \mathbf{X}_{\tau}}{\mathrm{d}  t} & = 
\mathbb{E} [\mathcal{L}_{\alpha} \mathbf{X}_{\tau}] + C^{1,1} \mathbb{E} \mathbf{X}_{\tau} + C^{1,2} \mathbb{E} \mathbf{Y}, \nonumber \\
\frac{\mathrm{d} \mathbb{E} \mathbf{Y}}{\mathrm{d}  t} & = 
\boldsymbol{\beta}_{\cdot,0} + C^{2,1} \mathbb{E} \mathbf{X}_{\tau} + \bar{C}^{2,2} \mathbb{E} \mathbf{Y},
 \label{eq:average_firstorder}
\end{align}
where $\mathcal{L}_{\alpha}$ is the (unknown) generator of the input network, 
 $\boldsymbol{\beta}_{\cdot,0} = (\beta_{1,0}, \beta_{2,0}, \ldots,\beta_{M,0})\in \mathbb{R}_{\ge}^M$
is induced by the core network $\mathcal{R}_{\beta}(\mathcal{Y})$,
while cross-nonnegative matrix $C^{1,1} \in \mathbb{R}^{N_{\tau} \times N_{\tau}}$ and nonnegative matrices 
$C^{1,2} \in \mathbb{R}_{\ge}^{N_{\tau} \times M}$ and $C^{2,1} \in \mathbb{R}_{\ge}^{M \times N_{\tau}}$
are induced by the interfacing network $\mathcal{R}_{\gamma}(\mathcal{X}_{\tau},\mathcal{Y})$. 
Matrix $\bar{C}^{2,2} \in \mathbb{R}_{\ge}^{M \times M}$ can be written 
as $\bar{C}^{2,2} = (B + C^{2,2})$, with cross-nonnegative matrices $B$ and $C^{2,2}$ being 
induced by $\mathcal{R}_{\beta}(\mathcal{Y})$ and $\mathcal{R}_{\gamma}(\mathcal{X}_{\tau},\mathcal{Y})$, 
respectively. Note that $\bar{C}^{2,2}$ encodes all of the reactions that change the copy-numbers of 
the controlling species $\mathcal{Y}$, either in $\mathcal{X}_{\tau}$-independent (via matrix $B$) 
or $\mathcal{X}_{\tau}$-dependent manner (via matrix $C^{2,2}$). 
Let us also note that, by definition, $\boldsymbol{\beta} \in \mathbb{R}_{>}^b$
contains nonzero elements from vector $\boldsymbol{\beta}_{\cdot,0}$, and 
 absolute value of nonzero elements from matrix $B$; 
similarly, $\boldsymbol{\gamma} \in \mathbb{R}_{>}^c$ contains 
absolute values of nonzero elements from $C^{1,1}$, $C^{1,2}$, $C^{2,1}$ and $C^{2,2}$. 
We now present conditions on the matrices $C^{2,1}  = (\boldsymbol{\gamma}_{\cdot,1}^{2,1}, 
\boldsymbol{\gamma}_{\cdot,2}^{2,1}, \ldots, \boldsymbol{\gamma}_{\cdot,n}^{2,1})$
and $\bar{C}^{2,2}$ that are necessary for unimolecular networks to exert robust (integral-feedback) control. 

\begin{lemma}\label{lemma:linear_intfed}
Assume a unimolecular reaction network $\mathcal{R}_{\beta,\gamma}(\mathcal{X}_{\tau},\mathcal{Y})$ 
is a robust controller for a black-box network $\mathcal{R}_{\alpha}(\mathcal{X})$. 
Then, $\bar{C}^{2,2}$ from~{\rm (\ref{eq:average_firstorder})} is a singular matrix 
all eigenvalues of which have nonpositive real parts for all
$(\boldsymbol{\beta}, \boldsymbol{\gamma}) \in \mathbb{S}_{\beta,\gamma}$.
Furthermore, the first column of matrix $C^{2,1}$ is nonzero.
\end{lemma}
\begin{proof}
Assume $\bar{C}^{2,2}$ has an eigenvalue with a positive real part for some
$(\boldsymbol{\beta}, \boldsymbol{\gamma}) \in \mathbb{S}_{\beta,\gamma}$;
it then follows from~(\ref{eq:average_firstorder}) that 
$\lim_{t \to \infty} \mathbb{E} Y_i = \infty$ for some $i \in \{1, 2, \ldots, M\}$. Therefore,
 condition~(C.I) from Definition~\ref{def:robust} does not hold for all 
$(\boldsymbol{\beta}, \boldsymbol{\gamma}) \in \mathbb{S}_{\beta,\gamma}$.

Assume $\bar{C}^{2,2}$ is nonsingular for some
$(\boldsymbol{\beta}, \boldsymbol{\gamma}) \in \mathbb{S}_{\beta,\gamma}$.
Let $(\mathbb{E} \mathbf{X}_{\tau}^*, \mathbb{E} \mathbf{X}_{\rho}^*, \mathbb{E} \mathbf{Y}^*)$
be the unique stationary first-moments, obtained by 
setting to zero the left-hand side in~(\ref{eq:average_firstorder}), and in the ODEs for higher-order moments. 
Then, one can eliminate $\mathbb{E} \mathbf{Y}^*$ via $\mathbb{E} \mathbf{Y}^* = - (\bar{C}^{2,2})^{-1}
 (\boldsymbol{\beta}_{\cdot,0} + C^{2,1} \mathbb{E} \mathbf{X}_{\tau}^*)$. The remaining 
system of equations for $(\mathbb{E} \mathbf{X}_{\tau}^*, \mathbb{E} \mathbf{X}_{\rho}^*)$, 
and hence $\mathbb{E} X_1^*$, depends on $\boldsymbol{\alpha}$ via the unknown generator $\mathcal{L}_{\alpha}$. 
Therefore, condition~(R.II) from Definition~\ref{def:robust} does not hold for all 
$(\boldsymbol{\beta}, \boldsymbol{\gamma}) \in \mathbb{S}_{\beta,\gamma}$.

Assume the first column of $C^{2,1}$ is zero, $\boldsymbol{\gamma}_{\cdot,1}^{2,1} = \mathbf{0}$. 
By the Fredholm alternative theorem, a necessary condition 
for the stationary first-moment $\mathbb{E} \mathbf{Y}^*$ to exist is then given by
$(\langle \mathbf{w}, \boldsymbol{\beta}_{\cdot,0} \rangle + 
\sum_{i = 2}^n \langle \mathbf{w}, \boldsymbol{\gamma}_{\cdot,i}^{2,1}\rangle \mathbb{E} X_i^*) = 0$
for every $\mathbf{w} \in \mathbb{R}^M$ such that $(\bar{C}^{2,2})^{\top} \mathbf{w} = \mathbf{0}$.
The stationary average $\mathbb{E} X_1^*$ is not uniquely determined by these constraints and, hence, 
depends on $\boldsymbol{\alpha}$. 
Therefore, condition~(R.II) from Definition~\ref{def:robust} does not hold.
\end{proof}

\begin{theorem}\label{theorem:linear_intfed}
There does not exist a unimolecular integral-feedback controller.
\end{theorem}
\begin{proof}
Suppose, for contradiction, that $\mathcal{R}_{\beta,\gamma}(\mathcal{X}_{\tau},\mathcal{Y})$ is a
robust (integral-feedback) controller, 
i.e. that the corresponding output network $\mathcal{R}_{\alpha, \beta, \gamma}(\mathcal{X}, \mathcal{Y}) 
= \mathcal{R}_{\alpha}(\mathcal{X}) \cup \mathcal{R}_{\beta,\gamma}(\mathcal{X}_{\tau},\mathcal{Y})$ 
satisfies Definition~\ref{def:robust}. Lemma~(\ref{lemma:linear_intfed}) then implies
that $(\bar{C}^{2,2})^{\top}$ is a singular cross-nonnegative matrix 
all eigenvalues of which have nonpositive real parts; in what follows, we assume
that $(\bar{C}^{2,2})^{\top}$ is irreducible. It then follows from~\cite[Theorem 5.6]{Mmatrix}
that there exists a positive vector $\mathbf{w} \in \mathbb{R}_{>}^N$ 
such that $(\bar{C}^{2,2})^{\top} \mathbf{w} = \mathbf{0}$, so that
$\mathrm{d}/\mathrm{d} t \langle \mathbf{w}, \mathbb{E} \mathbf{Y} \rangle
= (\langle \mathbf{w}, \boldsymbol{\beta}_{\cdot,0} \rangle + 
\sum_{i = 1}^n \langle \mathbf{w}, \boldsymbol{\gamma}_{\cdot,i}^{2,1}\rangle \mathbb{E} X_i)$. 
Using the fact that, by  Lemma~(\ref{lemma:linear_intfed}), $\boldsymbol{\gamma}_{\cdot,1}^{2,1} \ne \mathbf{0}$, 
and that, by Definition~\ref{def:robust}, $\lim_{t \to \infty} \mathbb{E} X_1(t; \, \boldsymbol{\beta},
 \boldsymbol{\gamma}; \, \boldsymbol{\alpha}) 
= \mathbb{E} X_1^*(\boldsymbol{\beta},\boldsymbol{\gamma}) > 0$, it follows that 
$\mathrm{d}/\mathrm{d} t \langle \mathbf{w}, \mathbb{E} \mathbf{Y} \rangle > 0$ in the long-run. 
Therefore, $\lim_{t \to \infty} \mathbb{E} Y_i = \infty$ for some $i \in \{1, 2, \ldots, M\}$,
implying that condition~(C.I) from Definition~\ref{def:control} does not hold. Hence, 
 $\mathcal{R}_{\beta,\gamma}(\mathcal{X}_{\tau},\mathcal{Y})$ is not a robust controller. 
If $\bar{C}^{2,2}$ is reducible, then analogous
argument can be applied to each of the underlying irreducible components, 
implying the statement of the theorem. Due to linearity of the controller, the same argument implies
that a deterministic integral-feedback controller does not exist. 
\end{proof}

The proof of Theorem~\ref{theorem:linear_intfed}
shows that a unimolecular integral-feedback controller does not exist because
the underlying reaction network $\mathcal{R}_{\beta,\gamma}(\mathcal{X}_{\tau},\mathcal{Y})$
triggers a deterministic and stochastic NEC (see also Definition~\ref{def:catastrophe} 
in Section~\ref{app:negblowup}) for all nonnegative 
initial conditions, violating the stability condition~(C.I) from Definition~\ref{def:control}.
For example, taking $C^{2,1} = (\boldsymbol{\gamma}_{\cdot,1}^{2,1}, \mathbf{0}, \ldots, \mathbf{0})$, it follows that 
the deterministic equilibrium of the target species is
$x_1^* = -\langle \mathbf{w}, \boldsymbol{\beta}_{\cdot,0}
 \rangle/\langle \mathbf{w},\boldsymbol{\gamma}_{\cdot,1}^{2,1} \rangle < 0$.

\section{Appendix: Hyperbolic kinetic transformation} \label{app:kinetictrans} 
In this section, we briefly discuss how to map arbitrary polynomial 
ODEs into dynamically similar mass-action RREs~\cite{Me_Homoclinic}.
To this end, let us note that every component of every polynomial function 
$\boldsymbol{\mathcal{P}}(\mathbf{x}; \, \boldsymbol{\alpha})
= (\mathcal{P}_1(\mathbf{x}; \, \boldsymbol{\alpha}), \mathcal{P}_2(\mathbf{x}; \, \boldsymbol{\alpha}),$
$ \ldots, \mathcal{P}_N(\mathbf{x}; \, \boldsymbol{\alpha}))^{\top} \in 
\mathbb{P}_{m}(\mathbb{R}^N; \, \mathbb{R}^N)$, where $\boldsymbol{\alpha} \in \mathbb{R}_{>}^a$, 
 can be written as
\begin{align}
\mathcal{P}_i(\mathbf{x}; \, \boldsymbol{\alpha}) = 
\sum_{j \in \mathcal{A}_{i}^{+}} \alpha_j \mathbf{x}^{\boldsymbol{\nu}_{j, \cdot}} 
- \sum_{j \in \mathcal{A}_{i, \mathcal{K}}^{-}} \alpha_j \mathbf{x}^{\boldsymbol{\nu}_{j, \cdot}} 
- \sum_{j \in \mathcal{A}_{i, \mathcal{N}}^{-}} \alpha_j \mathbf{x}^{\boldsymbol{\nu}_{j, \cdot}},
\; \; \; \; \; 
\forall i \in \{1, 2, \ldots, N\}. \label{eq:general_poly}
\end{align}
Here, $\mathcal{A}_{i}^{+}$ and $\mathcal{A}_{i}^{-} \equiv (\mathcal{A}_{i, \mathcal{K}}^{-} \cup \mathcal{A}_{i, \mathcal{N}}^{-})$
are the indices of all of the distinct positive and negative terms in $\mathcal{P}_i(\mathbf{x}; \,\boldsymbol{\alpha})$,
respectively, where
$\mathcal{A}_{i, \mathcal{K}}^{-} = \{j \in \mathcal{A}_i^{-} | \nu_{i,j} \ne 0\}$
and 
$\mathcal{A}_{i, \mathcal{N}}^{-} = \{j \in \mathcal{A}_i^{-} | \nu_{i,j} = 0\}$.

\begin{definition}[\textbf{Cross-negative terms; kinetic functions}]\label{def:cross_negative} 
Consider a polynomial function $\boldsymbol{\mathcal{P}}(\mathbf{x}; \, \boldsymbol{\alpha}) \in 
\mathbb{P}_{m}(\mathbb{R}^N; \, \mathbb{R}^N)$ with $\boldsymbol{\alpha} \in \mathbb{R}_{>}^a$, 
whose $i$-component is given by~{\rm (\ref{eq:general_poly})}.
Monomials $\{- \alpha_j \mathbf{x}^{\boldsymbol{\nu}_{j, \cdot}} \}_{j \in \mathcal{A}_{i, \mathcal{N}}^{-}}$ 
are called \emph{cross-negative} terms.
Polynomial functions without cross-negative terms, $\mathcal{A}_{i, \mathcal{N}}^{-} = \emptyset$,
 are called \emph{kinetic functions} and denoted by $\boldsymbol{\mathcal{K}}(\mathbf{x}; \, \boldsymbol{\alpha})$; 
the set of all kinetic functions of degree at most $m$ is denoted by 
$\mathbb{P}_{m}^{\mathcal{K}}(\mathbb{R}^N; \, \mathbb{R}^N)$. 
Polynomial functions with a cross-negative term, $\mathcal{A}_{i, \mathcal{N}}^{-} \ne \emptyset$,
 are called \emph{non-kinetic functions}
and denoted by $\boldsymbol{\mathcal{N}}(\mathbf{x}; \, \boldsymbol{\alpha})$; 
the set of all non-kinetic functions of degree at most $m$ is denoted by 
$\mathbb{P}_{m}^{\mathcal{N}}(\mathbb{R}^N; \, \mathbb{R}^N)$. 
\end{definition}
\noindent \emph{Remark}. Definition~\ref{def:cross_negative} captures the fact that 
biochemical reactions cannot consume a species when the concentration
of that species is zero, e.g. see network~(1); 
Theorem~\ref{theorem:trapping} is a direct consequence
of the absence of cross-negative terms in kinetic functions~\cite{Me_Homoclinic}.

In order to map an arbitrary input non-kinetic function 
$\boldsymbol{\mathcal{N}}(\mathbf{x}; \, \boldsymbol{\alpha}) 
\in \mathbb{P}_{m}^{\mathcal{N}}(\mathbb{R}^N; \, \mathbb{R}^N)$
into an output kinetic one 
$\boldsymbol{\mathcal{K}}(\bar{\mathbf{x}}; \, \bar{\boldsymbol{\alpha}}) 
\in \mathbb{P}_{\bar{m}}^{\mathcal{K}}(\mathbb{R}^{\bar{N}}; \, \mathbb{R}^{\bar{N}})$,
while preserving desired dynamical features over suitable time-intervals, 
a number of so-called \emph{kinetic transformations} have been developed in~\cite{Me_Homoclinic}. 
Such mappings involve a dimension expansion (i.e. an introduction of additional biochemical species, $\bar{N} \ge N$)
or an increase in non-linearity ($\bar{m} \ge m$). We now present one such kinetic transformation.

\begin{definition}[\textbf{Hyberbolic transformation}]\label{def:hyberbolic} 
Consider a system of polynomial {\rm ODE}s given by
\begin{align}
\frac{\mathrm{d} x_i}{\mathrm{d}  t} & = 
\sum_{j \in \mathcal{A}_{i}^{+}} \alpha_j \mathbf{x}^{\boldsymbol{\nu}_{j, \cdot}} 
- \sum_{j \in \mathcal{A}_{i, \mathcal{K}}^{-}} \alpha_j \mathbf{x}^{\boldsymbol{\nu}_{j, \cdot}},
\hspace{3.2cm} x_i(0) = x_i^0 \ge 0,
\; \; \; \; \; \textrm{for } i \in \{1, 2, \ldots, n\}, \nonumber \\
\frac{\mathrm{d} x_i}{\mathrm{d}  t} & = 
\sum_{j \in \mathcal{A}_{i}^{+}} \alpha_j \mathbf{x}^{\boldsymbol{\nu}_{j, \cdot}} 
- \sum_{j \in \mathcal{A}_{i, \mathcal{K}}^{-}} \alpha_j \mathbf{x}^{\boldsymbol{\nu}_{j, \cdot}} 
- \sum_{j \in \mathcal{A}_{i, \mathcal{N}}^{-}} \alpha_j \mathbf{x}^{\boldsymbol{\nu}_{j, \cdot}},
\; \; \; \; \; x_i(0) = x_i^0 > 0,
\; \; \; \; \; \textrm{for } i \in \{n+1, n+2, \ldots, N\}, \label{eq:pretrans}
\end{align}
with the index sets as defined in~{\rm (\ref{eq:general_poly})}.
Consider also the {\rm RRE}s given by
\begin{align}
\frac{\mathrm{d} x_i}{\mathrm{d}  t} & = 
\sum_{j \in \mathcal{A}_{i}^{+}} \alpha_j \mathbf{x}^{\boldsymbol{\nu}_{j, \cdot}} 
- \sum_{j \in \mathcal{A}_{i, \mathcal{K}}^{-}} \alpha_j \mathbf{x}^{\boldsymbol{\nu}_{j, \cdot}},
\hspace{2.0cm} x_i(0) = x_i^0,
\; \; \; \; \; \textrm{for } i \in \{1, 2, \ldots, n\}, \nonumber \\
\frac{\mathrm{d} x_i}{\mathrm{d}  t} & = 
\sum_{j \in \mathcal{A}_{i}^{+}} \alpha_j \mathbf{x}^{\boldsymbol{\nu}_{j, \cdot}} 
- \sum_{j \in \mathcal{A}_{i, \mathcal{K}}^{-}} \alpha_j \mathbf{x}^{\boldsymbol{\nu}_{j, \cdot}} 
- \frac{1}{\varepsilon} x_i \bar{x}_i,
\; \; \; \; \; x_i(0) = x_i^0,
\; \; \; \; \; \textrm{for } i \in \{n+1, n+2, \ldots, N\}, \nonumber \\
\frac{\mathrm{d} \bar{x}_i}{\mathrm{d}  t} & = 
\sum_{j \in \mathcal{A}_{i, \mathcal{N}}^{-}} \alpha_j \mathbf{x}^{\boldsymbol{\nu}_{j, \cdot}}
- \frac{1}{\varepsilon} x_i \bar{x}_i,
\hspace{2.9cm} \bar{x}_i(0)  > 0,
\; \; \; \; \; \textrm{for } i \in \{n+1, n+2, \ldots, N\}, \label{eq:trans}
\end{align}
with a parameter $\varepsilon \in \mathbb{R}_{>}$. 
Kinetic transformation $\Psi_{H}^{\varepsilon} : \mathbb{P}_{m}(\mathbb{R}^N; \, \mathbb{R}^N)
\to \mathbb{P}_{\bar{m}}^{\mathcal{K}}(\mathbb{R}^{(2 N - n)}; \, \mathbb{R}^{(2 N -n)})$, 
mapping the right-hand side of {\rm ODE}s~{\rm (\ref{eq:pretrans})}
to the right-hand side of~{\rm (\ref{eq:trans})}, with $\bar{m} \le m + 2$, 
is called a \emph{hyberbolic transformation}. 
\end{definition}

By comparing equilibria of~(\ref{eq:pretrans}) and~(\ref{eq:trans}), 
the following proposition is established. 
\begin{proposition}\label{proposition:hyperbolic0}
Let $\mathbf{x}^* = (x_1^*, x_2^*, \ldots, x_N^*) \in \mathbb{R}^N$ be an 
equilibrium of the {\rm ODE} system~{\rm (\ref{eq:pretrans})}
with $x_{i}^* \ne 0$ for all $i \in \{n+1, n+2, \ldots, N\}$. 
Then, $(\mathbf{x}^*, \mathbf{\bar{x}}^*) \in \mathbb{R}^{2 N - n}$ 
is an equilibrium of the {\rm ODE} system~{\rm (\ref{eq:trans})}
for some $\mathbf{\bar{x}}^* \in \mathbb{R}^{N - n}$. 
\end{proposition}

Let us note that, while the $\mathbf{x}$-equilibria are preserved by the hyperbolic kinetic
transformation, their properties, such as stability, are not necessarily
preserved; a stronger dynamical preservation is ensured by taking $\varepsilon$
sufficiently small. 

\begin{theorem}\label{theorem:hyperbolic}
Solutions of~{\rm (\ref{eq:pretrans})}, with $(x_{n+1}, x_{n+2}, \ldots, x_{N}) \in \mathbb{R}_{>}^{N-n}$, 
are asymptotically equivalent to the solutions of~{\rm (\ref{eq:trans})} in the limit $\varepsilon \to 0$.
\end{theorem}
\begin{proof}
Let us introduce a change of coordinates $y_i = (x_i - \bar{x}_i)$ for $i \in \{n+1, n+2, \ldots, N\}$, 
leading to a regular singularly perturbed system in the variables
$(x_1, x_2, \ldots, x_n, y_{n+1}, y_{n+2}, \ldots, y_{N},\bar{x}_{n+1}, \bar{x}_{n+2},$ $\ldots,\bar{x}_N)$. 
The adjoined sub-system has an isolated equilibrium $(\bar{x}_{n+1}^*, \bar{x}_{n+2}^*, \ldots, \bar{x}_N^*) = \mathbf{0}$
which, when substituted into the degenerate sub-system in the variables 
$(x_1, x_2, \ldots, x_n,$$ y_{n+1}, y_{n+2},$ $\ldots,y_{N})$, leads to~(\ref{eq:pretrans}).
If $(x_{n+1}, x_{n+2}, \ldots, x_{N})$ $\in \mathbb{R}_{>}^{N-n}$, then the trivial adjoined equilibrium is stable,
and Tikhonov's theorem~\cite{Tikhonov} implies the statement of Theorem~\ref{theorem:hyperbolic}.
\end{proof}

\section{Appendix: Robust control of unimolecular input networks} \label{app:proof} 
 In this section, we analyze performance of the class of second-order integral-feedback controllers
 $\mathcal{R}_{\beta,\gamma}(\{X_1, X_2\}, \{Y_1, Y_2\}) = \mathcal{R}_{\beta}(Y_1, Y_2) 
\cup \mathcal{R}_{\gamma}^{0}(Y_2; \, X_1) \cup \mathcal{R}_{\gamma}^{+}(X_i; \, Y_1) 
\cup \mathcal{R}_{\gamma}^{-}(X_j; \, Y_2)$, given by~(8),
by embedding them into the class of unimolecular input networks
whose RREs, given by
\begin{align}
\frac{\mathrm{d} \mathbf{x}}{\mathrm{d} t}& = \boldsymbol{\alpha}_{\cdot, 0} + A \mathbf{x},
\label{eq:input_firstorder}
\end{align}
have an asymptotically stable equilibrium. We also assume the input network 
has two target species $X_{\tau} = \{X_1, X_2\}$, 
and the goal is to control the concentration/first-moment of the target species $X_1$.

In what follows, we let $A_{(i_1, i_2, \ldots, i_n),(j_1, j_2, \ldots, j_m)} \in 
\mathbb{R}^{(N-n) \times (N-m)}$ denote
the sub-matrix obtained by removing from $A \in \mathbb{R}^{N \times N}$
the rows $\{i_1, i_2, \ldots, i_n\} \subset \{1, 2\, \ldots, N\}$ 
and columns $\{j_1, j_2, \ldots, j_m\}$ $\subset \{1, 2\, \ldots, N\}$. 
Furthermore, we let $A_{\boldsymbol{\alpha}_{\cdot, j} \to \mathbf{x}} \in \mathbb{R}^{N \times N}$ 
denote the matrix obtained by replacing the $j$-column of $A$ by a vector $\mathbf{x} \in \mathbb{R}^N$.  

\begin{theorem}{\rm (Deterministic equilibria)}\label{theorem:equilibrium}
Let $\mathcal{R}_{\alpha} = \mathcal{R}_{\alpha}(\mathcal{X})$, with species $\mathcal{X} = \{X_1, X_2, \ldots, X_N\}$, 
be the class of unimolecular input networks whose {\rm RRE}s~{\rm(\ref{eq:input_firstorder})} 
have an asymptotically stable equilibrium.
Let $\mathcal{R}_{\beta,\gamma} = \mathcal{R}_{\beta,\gamma}(\{X_1, X_2\}, \{Y_1, Y_2\})$ 
be the controller given by~{\rm(8)}.
Then, the output network $\mathcal{R}_{\alpha,\beta,\gamma}
= \mathcal{R}_{\alpha} \cup \mathcal{R}_{\beta,\gamma}$ satisfies the following properties:
\begin{enumerate}
\item[{\rm (i)}] \textbf{Pure positive interfacing}. Let $\mathcal{R}_{\gamma}^{-} = \emptyset$.
If positive interfacing is direct, $\mathcal{R}_{\gamma}^{+} = \mathcal{R}_{\gamma}^{+}(X_1, Y_1) \ne \emptyset$,
then there exists a nonnegative equilibrium if and only if
$\frac{\beta_0}{\gamma_1} >  \frac{|A_{\boldsymbol{\alpha}_{\cdot,1} \to -\boldsymbol{\alpha}_{\cdot,0}}|}{|A|}$; 
the same is true if positive interfacing is indirect, 
$\mathcal{R}_{\gamma}^{+} = \mathcal{R}_{\gamma}^{+}(X_2, Y_1) \ne \emptyset$ , 
and if $|A_{2,1}| \ne 0$.
\item[{\rm (ii)}] \textbf{Pure negative interfacing}. Let $\mathcal{R}_{\gamma}^{+} = \emptyset$.
If negative interfacing is direct,
$\mathcal{R}_{\gamma}^{-} = \mathcal{R}_{\gamma}^{-}(X_1, Y_2) \ne \emptyset$,
then there exists a nonnegative equilibrium if and only if
$\frac{\beta_0}{\gamma_1} <  \frac{|A_{1 \to -\boldsymbol{\alpha}_{\cdot,0}}|}{|A|}$.
If negative interfacing is indirect,
$\mathcal{R}_{\gamma}^{-} = \mathcal{R}_{\gamma}^{-}(X_2, Y_2) \ne \emptyset$, 
and if $|A_{2,1}| \ne 0$, then there exists a nonnegative equilibrium if and only if
$\frac{|(A_{\boldsymbol{\alpha}_{\cdot,1} \to -\boldsymbol{\alpha}_{\cdot,0}})_{2,2}|}{|A_{2,2}|} < 
\frac{\beta_0}{\gamma_1} < \frac{|A_{\boldsymbol{\alpha}_{\cdot,1} \to -\boldsymbol{\alpha}_{\cdot,0}}|}{|A|}$.
\item[{\rm (iii)}] \textbf{Combined indirect negative interfacing}. Let 
$\mathcal{R}_{\gamma}^{-} =  \mathcal{R}_{\gamma}^{-}(X_2, Y_2) \ne \emptyset$.
If positive interfacing is direct,
$\mathcal{R}_{\gamma}^{+} = \mathcal{R}_{\gamma}^{+}(X_1, Y_1) \ne \emptyset$,
then there exists a nonnegative equilibrium if and only if
$\frac{\beta_0}{\gamma_1} >  \frac{|(A_{\boldsymbol{\alpha}_{\cdot,1} \to -\boldsymbol{\alpha}_{\cdot,0}})_{2,2}|}{|A_{2,2}|}$; 
the same is true if positive interfacing is indirect, 
$\mathcal{R}_{\gamma}^{+} = \mathcal{R}_{\gamma}^{+}(X_2, Y_1) \ne \emptyset$, 
and if $|A_{2,1}| \ne 0$.
\item[{\rm (iv)}] \textbf{Combined direct negative interfacing}. Let 
$\mathcal{R}_{\gamma}^{-} =  \mathcal{R}_{\gamma}^{-}(X_1, Y_2) \ne \emptyset$.
If positive interfacing is direct, 
$\mathcal{R}_{\gamma}^{+} = \mathcal{R}_{\gamma}^{+}(X_1, Y_1) \ne \emptyset$,
then there exists a nonnegative equilibrium for any choice of $\boldsymbol{\beta} \in \mathbb{R}_{>}^{2}$
 and $\boldsymbol{\gamma} \in \mathbb{R}_{>}^3$;  
the same is true if positive interfacing is indirect, 
$\mathcal{R}_{\gamma}^{+} = \mathcal{R}_{\gamma}^{+}(X_2, Y_1) \ne \emptyset$,
and if $|A_{2,1}| \ne 0$.
\end{enumerate}
\end{theorem}

\begin{proof}
Equilibria of the output network, denoted by 
$(\mathbf{x}^*, y_1^*, y_2^*) \in \mathbb{R}^{N+2}$, satisfy
\begin{align}
\mathbf{0} & = \boldsymbol{\alpha}_{\cdot, 0} + A \mathbf{x}^* 
+ \gamma_{2} y_1^* \mathbf{e}_i  - \gamma_{3} x_j^* y_2^* \mathbf{e}_j, \; \; \; \; \; \textrm{for } i,j \in \{1, 2\}, \nonumber \\
x_1^* & = \frac{\beta_0}{\gamma_1} \in \mathbb{R}_{>},
\hspace{0.5cm} y_1^* y_2^* = \frac{\beta_0}{\beta_1}, 
\hspace{0.5cm} \textrm{where } y_1^*, y_2^* \notin \{0\},
\label{eq:output_firstorder}
\end{align}
where $(\boldsymbol{\alpha}_{\cdot, 0} + A \mathbf{x})$ is the unknown kinetic function 
from~(\ref{eq:input_firstorder}), see also Lemma~\ref{theorem:trapping_firstorder}.
Statements (iii) and (iv) of the theorem are now proved; 
statements (i) and (ii) follow analogously. 

\emph{Case}: $\mathcal{R}_{\gamma}^{+} = \mathcal{R}_{\gamma}^{+}(X_1, Y_1) \ne \emptyset$ 
and $\mathcal{R}_{\gamma}^{-} =  \mathcal{R}_{\gamma}^{-}(X_2, Y_2) \ne \emptyset$.
Defining vectors
$\mathbf{\bar{x}}^* \equiv (x_2^*, x_3^*, \ldots, x_N^*)^{\top}$, 
$\boldsymbol{\bar{\alpha}}_{\cdot, 0} \equiv (\alpha_{2,0}, \alpha_{3,0}, \ldots, \alpha_{N,0})^{\top}$, 
$\boldsymbol{\bar{\alpha}}_{1, \cdot} \equiv (\alpha_{1,2}, \alpha_{1,3}, \ldots, \alpha_{1,N})^{\top}$
and
$\boldsymbol{\bar{\alpha}}_{\cdot, 1} \equiv (\alpha_{2,1}, \alpha_{3,1}, \ldots, \alpha_{N,1})^{\top}$, 
equation~(\ref{eq:output_firstorder}) can be written as
\begin{align}
0 & =  \left( \alpha_{1,0} - \alpha_{1,1} \frac{\beta_0}{\gamma_1} + 
\left \langle \boldsymbol{\bar{\alpha}}_{1, \cdot}, \mathbf{\bar{x}^*} \right \rangle  \right) y_2^*
+ \gamma_{2} \frac{\beta_0}{\beta_1}, \nonumber \\
\mathbf{\bar{x}}^* & = - \left(A_{1,1} - \gamma_{3} y_2^* \, \mathrm{diag}(\mathbf{e}_1) \right)^{-1} 
\left(\boldsymbol{\bar{\alpha}}_{\cdot, 0}
+ \boldsymbol{\bar{\alpha}}_{\cdot, 1} \frac{\beta_0}{\gamma_1} \right). 
\label{eq:output_case3.1}
\end{align}
Assume there exists $y_2^* > 0$. Since all eigenvalues of $A$ have
negative real parts, the same is true for $A_{1,1}$, implying 
nonpositivity of matrix $\left(A_{1,1} - \gamma_{3} y_2^* \, \mathrm{diag}(\mathbf{e}_1) \right)^{-1}
 \in \mathbb{R}_{\le}^{(N-1) \times (N-1)}$~\cite[Theorem 4.3]{Mmatrix}, 
so that $\mathbf{\bar{x}}^* \in \mathbb{R}_{\ge}^{N-1}$. 
Matrix $\left(A_{1,1} - \gamma_{3} y_2^* \, \mathrm{diag}(\mathbf{e}_1) \right)^{-1}$ 
can be written in a cofactor form
\begin{align}
\frac{\left(A_{1,1} - \gamma_{3} y_2^* \, \mathrm{diag}(\mathbf{e}_1) \right)^{-1}}{|A_{1,1}| - \gamma_{3} y_2^* |A_{(1,2), (1,2)}|}
 &  =
\begin{bmatrix}
 C_{1,1} & C_{2,1} & \ldots & C_{N-1,1} \\
 C_{1,2} & C_{2,2} & \ldots & C_{N-1,2} \\
 \vdots & \vdots & \ddots & \vdots \\
 C_{1,N-1} & C_{2,N-1} & \ldots & C_{N-1,N-1}
   \end{bmatrix}
- \gamma_{3} y_2
\begin{bmatrix}
 0 & 0 & \ldots & 0 \\
0 & D_{1,1} & \ldots & D_{N-2,1} \\
 \vdots & \vdots & \ddots & \vdots \\
0 & D_{1,N-2} & \ldots & D_{N-2,N-2}
   \end{bmatrix},
\label{eq:inverse_case3}
\end{align}
where  $C_{i.j}$ and $D_{i,j}$ are the $(i,j)$-cofactors of $A_{1,1}$ and $A_{(1,2), (1,2)}$, respectively. 
Substituting~(\ref{eq:inverse_case3}) into~(\ref{eq:output_case3.1}), one obtains the quadratic equation
\begin{align}
0 & = \gamma_{3} \frac{|A_{2,2}|}{|A_{1,1}|} \left(\frac{\beta_0}{\gamma_1} - 
\frac{|(A_{\boldsymbol{\alpha}_{\cdot,1} \to -\boldsymbol{\alpha}_{\cdot,0}})_{2,2}|}{|A_{2,2}|} \right) y_2^2 \nonumber \\
& - \left[  \frac{|A|}{|A_{1,1}|}  \left(\frac{\beta_0}{\gamma_1} - \frac{|A_{\boldsymbol{\alpha}_{\cdot,1} \to - \boldsymbol{\alpha}_{\cdot,0}}|}{|A|} \right) - 
\gamma_{2} \gamma_{3} \frac{\beta_0}{\beta_1} \frac{|A_{(1,2),(1,2)}|}{|A_{1,1}|} \right] y_2
- \gamma_{2} \frac{\beta_0}{\beta_1}. \label{eq:quadratic3.1}
\end{align}
Since $A$ is cross-nonnegative with a negative spectral abscissa, 
it follows that $|A_{2,2}|/|A_{1,1}| > 0$,  $|A|/|A_{1,1}| < 0$, $|A_{(1,2),(1,2)}|/|A_{1,1}| < 0$, 
$|(A_{\boldsymbol{\alpha}_{\cdot,1} \to -\boldsymbol{\alpha}_{\cdot,0}})_{2,2}|/|A_{2,2}| \ge 0$,
and
$|A_{\boldsymbol{\alpha}_{\cdot,1} \to -\boldsymbol{\alpha}_{\cdot,0}}|/|A| \ge 0$.
If $(\beta_0/\gamma_1 - |(A_{\boldsymbol{\alpha}_{\cdot,1} \to -\boldsymbol{\alpha}_{\cdot,0}})_{2,2}|/|A_{2,2}|) > 0$, 
it follows from~(\ref{eq:quadratic3.1}) that there 
exists an equilibrium $(y_1^*, y_2^*) \in \mathbb{R}_{>}^{2}$, 
consistent with the assumption.
On the other hand, if  $(\beta_0/\gamma_1 - |(A_{\boldsymbol{\alpha}_{\cdot,1} \to -\boldsymbol{\alpha}_{\cdot,0}})_{2,2}|/|A_{2,2}|) \le 0$,
using the fact that $|A_{\boldsymbol{\alpha}_{\cdot,1} \to -\boldsymbol{\alpha}_{\cdot,0}}|/|A| \ge 
|(A_{\boldsymbol{\alpha}_{\cdot,1} \to -\boldsymbol{\alpha}_{\cdot,0}})_{2,2}|/|A_{2,2}|$, 
it follows that $(y_1^*, y_2^*) \notin \mathbb{R}_{>}^{2}$.

\emph{Case}: $\mathcal{R}_{\gamma}^{+} = \mathcal{R}_{\gamma}^{+}(X_2, Y_1) \ne \emptyset$ 
and $\mathcal{R}_{\gamma}^{-} =  \mathcal{R}_{\gamma}^{-}(X_2, Y_2) \ne \emptyset$. 
Equation~(\ref{eq:output_firstorder}) reads
\begin{align}
0 & =   \alpha_{1,0} - \alpha_{1,1} \frac{\beta_0}{\gamma_1} + 
\left \langle \boldsymbol{\bar{\alpha}}_{1, \cdot}, \mathbf{\bar{x}^*} \right \rangle, \nonumber \\
\mathbf{\bar{x}}^* & = - \left(A_{1,1} - \gamma_{3} y_2^* \, \mathrm{diag}(\mathbf{e}_1) \right)^{-1} 
\left(\boldsymbol{\bar{\alpha}}_{\cdot, 0}
+ \boldsymbol{\bar{\alpha}}_{\cdot, 1} \frac{\beta_0}{\gamma_1}
  + \gamma_{2} \frac{\beta_0}{\beta_1} (y_2^*)^{-1} \mathbf{e}_1 \right),
\label{eq:output_case3.2}
\end{align}
from which, using~(\ref{eq:inverse_case3}), one obtains
\begin{align}
0 & =  \gamma_{3} \frac{|A_{2,2}|}{|A_{1,1}|} \left(\frac{\beta_0}{\gamma_1} 
-\frac{|(A_{\boldsymbol{\alpha}_{\cdot,1} \to - \boldsymbol{\alpha}_{\cdot,0}})_{2,2}|}{|A_{2,2}|} \right) y_2^2
-  \frac{|A|}{|A_{1,1}|} \left(\frac{\beta_0}{\gamma_1}  - 
\frac{|A_{\boldsymbol{\alpha}_{\cdot,1} \to -\boldsymbol{\alpha}_{\cdot,0}}|}{|A|} \right) y_2 
 + \gamma_{2} \frac{\beta_0}{\beta_1} \frac{|A_{2,1}|}{|A_{1,1}|}. \label{eq:quadratic3.2}
\end{align}
If $|A_{2,1}| \ne 0$ and $(\beta_0/\gamma_1 - |(A_{\boldsymbol{\alpha}_{\cdot,1} \to - 
\boldsymbol{\alpha}_{\cdot,0}})_{2,2}|/|A_{2,2}|) > 0$, using the fact that $|A_{2,1}|/|A_{1,1}| < 0$,
it follows that there exists an equilibrium $(y_1^*, y_2^*) \in \mathbb{R}_{>}^{2}$. 
If $|A_{2,1}| \ne 0$ and $(\beta_0/\gamma_1 - |(A_{\boldsymbol{\alpha}_{\cdot,1} \to - 
\boldsymbol{\alpha}_{\cdot,0}})_{2,2}|/|A_{2,2}|) \le 0$, then $(y_1^*, y_2^*) \notin \mathbb{R}_{>}^{2}$.

\emph{Case}: $\mathcal{R}_{\gamma}^{+} = \mathcal{R}_{\gamma}^{+}(X_1, Y_1) \ne \emptyset$ 
and $\mathcal{R}_{\gamma}^{-} =  \mathcal{R}_{\gamma}^{-}(X_1, Y_2) \ne \emptyset$. 
Equation~(\ref{eq:output_firstorder}) reads
\begin{align}
0 & = \gamma_{2} (y_1^*)^2 + \left(\alpha_{1,0} - \alpha_{1,1} \frac{\beta_0}{\gamma_1} + 
 \langle \boldsymbol{\bar{\alpha}}_{1, \cdot}, \mathbf{\bar{x}}^* \rangle \right) y_1
 - \frac{\beta_0^2 \gamma_{3}}{\beta_1 \gamma_1}, \nonumber \\
\mathbf{\bar{x}}^*  & = - A_{1,1}^{-1} \left( \boldsymbol{\bar{\alpha}}_{\cdot, 0} +  
\boldsymbol{\bar{\alpha}}_{\cdot, 1} \frac{\beta_0}{\gamma_1} \right) \in \mathbb{R}_{\ge}^{N-1},
\label{eq:output_case4.1}
\end{align}
from which it follows that $(y_1^*, y_2^*) \in \mathbb{R}_{>}^{2}$ for any choice of 
$\boldsymbol{\beta} \in \mathbb{R}_{>}^{2}$ and $\boldsymbol{\gamma} \in \mathbb{R}_{>}^3$. 

\emph{Case}: $\mathcal{R}_{\gamma}^{+} = \mathcal{R}_{\gamma}^{+}(X_2, Y_1) \ne \emptyset$ 
and $\mathcal{R}_{\gamma}^{-} =  \mathcal{R}_{\gamma}^{-}(X_1, Y_2) \ne \emptyset$. 
Equation~(\ref{eq:output_firstorder}) reads
\begin{align}
0 & =  \left( \alpha_{1,0} - \alpha_{1,1} \frac{\beta_0}{\gamma_1} + 
\left \langle \boldsymbol{\bar{\alpha}}_{1, \cdot}, \mathbf{\bar{x}^*} \right \rangle  \right) y_1^*
- \frac{\beta_0^2 \gamma_{3}}{\beta_1 \gamma_1}, \nonumber \\
\mathbf{\bar{x}}^* & = - A_{1,1}^{-1}\left(\boldsymbol{\bar{\alpha}}_{\cdot, 0}
+ \boldsymbol{\bar{\alpha}}_{\cdot, 1} \frac{\beta_0}{\gamma_1} +
 \gamma_{2} y_1^* \mathbf{e}_1 \right). 
\label{eq:output_case4.2}
\end{align}
Assume there exists $y_1^* > 0$. Then, it follows that $\mathbf{x}^* \in \mathbb{R}_{\ge}^{N-1}$.
Substituting~(\ref{eq:inverse_case3}) with $\gamma_3 = 0$ 
into~(\ref{eq:output_case4.2}), one obtains the quadratic equation
\begin{align}
0 & = \gamma_{2}  \frac{|A_{2,1}|}{|A_{1,1}|} (y_1^*)^2 -
\frac{|A|}{|A_{1,1}|} \left(\frac{\beta_0}{\gamma_1} - 
\frac{|A_{\boldsymbol{\alpha}_{\cdot,1} \to -\boldsymbol{\alpha}_{\cdot,0}}|}{|A|}   \right) y_1^*
 + \frac{\beta_0^2 \gamma_{3}}{\beta_1 \gamma_1}.
\label{eq:quadratic4.2}
\end{align}
If $|A_{2,1}| \ne 0$, there
exists an equilibrium $(y_1^*, y_2^*) \in \mathbb{R}_{>}^{2}$ for any choice of 
$\boldsymbol{\beta} \in \mathbb{R}_{>}^{2}$
 and $\boldsymbol{\gamma} \in \mathbb{R}_{>}^3$, consistent with the assumption.
\end{proof}
\noindent \emph{Remark}. If the $x_1$-component of the asymptotically stable equilibrium 
of the input network $\mathcal{R}_{\alpha}(\mathcal{X})$
is zero, $x_1^* = 0$, then, like case (iv), cases (i) and (iii) from Theorem~\ref{theorem:equilibrium} 
also generically ensure an existence of a nonnegative
equilibrium; case (ii) then necessarily leads to negative equilibria. 
However, stable unimolecular networks with $x_1^* = 0$ are a negligible subset
 of the more general stable unimolecular networks and are, hence, of limited practical relevance.

Let us note that Theorem~\ref{theorem:equilibrium} has an intuitive
interpretation. In particular, $x_1^{**} = |A_{\boldsymbol{\alpha}_{\cdot,1} \to -\boldsymbol{\alpha}_{\cdot,0}}|/|A|$
is the $x_1$-component of the equilibrium of the input network, 
$(x_1^{**})|_2 = |(A_{\boldsymbol{\alpha}_{\cdot,1} \to -\boldsymbol{\alpha}_{\cdot,0}})_{2,2}|/|A_{2,2}|$
is the $x_1$-component of the equilibrium of the restricted
input network $\mathcal{R}_{\alpha}(\mathcal{X} \setminus X_2)$ with species concentration $x_2 \equiv 0$.
Furthermore, condition $|A_{2,1}| \ne 0$ requires that
the species $X_2$ influences $X_1$ in the input network; more precisely, $|A_{2,1}| \ne 0$ requires that there
exists at least one directed path from $X_2$ to $X_1$ in the digraph
induced by the Jacobian matrix of $\mathcal{R}_{\alpha}(\mathcal{X})$.

We now prove that if, for a particular choice of the rate coefficients, 
the output network from Theorem~\ref{theorem:equilibrium} has
no nonnegative equilibria, then the network displays deterministic
and stochastic blow-ups, i.e. we
prove that the output network undergoes deterministic and
stochastic NECs (see also Definition~\ref{def:catastrophe} 
in Section~\ref{app:negblowup}). In what follows, the critical values of 
$\beta_0/\gamma_1$ at which nonnegative equilibria
cease to exist in Theorem~\ref{theorem:equilibrium} are called \emph{bifurcation points}. 

\begin{theorem}{\rm (Negative-equilibrium catastrophe)}\label{theorem:IFCblowup}
Consider a unimolecular input network $\mathcal{R}_{\alpha}$
whose {\rm RRE}s have an asymptotically stable equilibrium,
and a controller $\mathcal{R}_{\beta,\gamma}$ of the form~{\rm(8)}.
Then, excluding the bifurcation points, 
the output network $\mathcal{R}_{\alpha,\beta,\gamma} = 
\mathcal{R}_{\alpha} \cup \mathcal{R}_{\beta,\gamma}$ displays
deterministic and stochastic negative-equilibrium catastrophe 
for all nonnegative initial conditions. 
\end{theorem}

\begin{proof}
Consider the case with pure direct positive interfacing from 
Theorem~\ref{theorem:equilibrium}(i); the RREs for the concentration 
$\mathbf{x} = (x_1, x_2, \ldots, x_N) \in \mathbb{R}_{\ge}^N$ 
and $(y_2 - y_1) \in \mathbb{R}^2$ are given by
\begin{align}
\frac{\mathrm{d} \mathbf{x}}{\mathrm{d}  t} & = 
\boldsymbol{\alpha}_{\cdot, 0} + A \mathbf{x} + \gamma_2 y_1 \mathbf{e}_1, \nonumber \\
\frac{\mathrm{d}}{\mathrm{d}  t} \left(y_2 - y_1\right) & = 
\gamma_1 x_1 - \beta_0. \label{eq:RRE_1}
\end{align}
Let $\mathbf{w} \equiv |A|^{-1} (C_{1,1}, C_{2,1}, \ldots, C_{N,1}) \in \mathbb{R}^{N}$, 
where $C_{i,j}$ is the $(i,j)$-cofactor of matrix $A \in \mathbb{R}^{N \times N}$.
Taking the inner product $\langle \mathbf{w}, \cdot \rangle$ 
in the first equation from~(\ref{eq:RRE_1}), and using 
the fact that $\langle A^{\top} \mathbf{w}, \mathbf{x} \rangle
 = x_1$, one obtains:
\begin{align}
\frac{\mathrm{d}}{\mathrm{d}  t} \langle \mathbf{w}, \mathbf{x} \rangle & = 
  \frac{|A_{\boldsymbol{\alpha}_{\cdot,1} \to \boldsymbol{\alpha}_{\cdot,0}}|}{|A|} 
+ x_1 + \gamma_2 \frac{|A_{1,1}|}{|A|} y_1. \label{eq:RRE_1_aux}
\end{align}
Using the fact that $|A_{1,1}|/|A| < 0$, equations~(\ref{eq:RRE_1}) 
and~(\ref{eq:RRE_1_aux}) imply that
\begin{align}
\frac{\mathrm{d}}{\mathrm{d}  t} \left(- \langle \mathbf{w}, \mathbf{x} \rangle
+  \gamma_{1}^{-1} \left(y_2 - y_1\right) \right)& = 
  \left(\frac{|A_{\boldsymbol{\alpha}_{\cdot,1} \to - \boldsymbol{\alpha}_{\cdot,0}}|}{|A|} 
 - \frac{\beta_0}{\gamma_1} \right) - \gamma_2 \frac{|A_{1,1}|}{|A|} y_1
\ge
 \left(\frac{|A_{\boldsymbol{\alpha}_{\cdot,1} \to 
- \boldsymbol{\alpha}_{\cdot,0}}|}{|A|}  - \frac{\beta_0}{\gamma_1} \right).
\label{eq:RRE_1_final}
\end{align}
By Theorem~\ref{theorem:equilibrium}(i), a nonnegative equilibrium does not exist if and only if 
$\frac{\beta_0}{\gamma_1} \le \frac{|A_{\boldsymbol{\alpha}_{\cdot,1} \to -\boldsymbol{\alpha}_{\cdot,0}}|}{|A|}$; 
excluding the bifurcation point $\frac{\beta_0}{\gamma_1} = \frac{|A_{\boldsymbol{\alpha}_{\cdot,1}
 \to -\boldsymbol{\alpha}_{\cdot,0}}|}{|A|}$, it follows from~(\ref{eq:RRE_1_final})
that the linear combination $(-\langle \mathbf{w}, \mathbf{x} \rangle
+  \gamma_{1}^{-1} \left(y_2 - y_1\right))$, and hence an underlying 
concentration, is a monotonically increasing function of time for all nonnegative initial conditions, 
i.e. the output network displays a deterministic NEC. Identical argument 
implies that then the output network displays a stochastic NEC as well. 
Deterministic and stochastic NECs for cases (ii) and (iii) 
from Theorem~\ref{theorem:equilibrium} are established analogously. 
\end{proof}

\section{Appendix: Robust control of bimolecular input networks} \label{app:biproof} 
Consider an arbitrary (unimolecular, bimolecular, or any higher-molecular) mass-action input network 
$\mathcal{R}_{\alpha} = \mathcal{R}_{\alpha}(\mathcal{X})$, whose RREs are given by
\begin{align}
\frac{\mathrm{d} x_1}{\mathrm{d} t} & = f_1(\mathbf{x}; \, \boldsymbol{\alpha}), \nonumber \\ 
\frac{\mathrm{d} x_2}{\mathrm{d} t} & = f_2(\mathbf{x}; \, \boldsymbol{\alpha}), \nonumber \\
\vdots \nonumber \\
\frac{\mathrm{d} x_N}{\mathrm{d} t} & = f_N(\mathbf{x}; \, \boldsymbol{\alpha}). 
\label{eq:bimolecular_input}
\end{align}
Furthermore, let $\mathcal{R}_{\beta,\gamma}^{\pm}(X_i, Y_{1,i}, Y_{2,i}) \equiv 
\left(\mathcal{R}_{\beta}(Y_{1,i}, Y_{2,i})
\cup \mathcal{R}_{\gamma}^0(Y_{2,i}; \, X_i)
\cup \mathcal{R}_{\gamma}^{+}(X_{i}; \, Y_{1,i})
\cup \mathcal{R}_{\gamma}^{-}(X_{i}; \, Y_{2,i}) \right)$
be the controller given
\begin{align}
\mathcal{R}_{\beta}(Y_{1,i}, Y_{2,i}): \;
& & \varnothing & \xrightarrow[]{\beta_{0,i}} Y_{1,i}, \nonumber \\
& & Y_{1,i} + Y_{2,i} & \xrightarrow[]{\beta_{1,i}} \varnothing, \nonumber \\
\mathcal{R}_{\gamma}^{0}(Y_{2,i}; \, X_i): \;
& & X_i & \xrightarrow[]{\gamma_{1,i}} X_i + Y_{2,i}, \nonumber \\
\mathcal{R}_{\gamma}^{+}(X_i; \, Y_{1,i}): \;
& & Y_{1,i} & \xrightarrow[]{\gamma_{2,i}} X_i + Y_{1,i}, \nonumber \\
\mathcal{R}_{\gamma}^{-}(X_i; \, Y_{2,i}): \;
& & X_i + Y_{2,i} & \xrightarrow[]{\gamma_{3,i}} Y_{2,i}.
\label{eq:IFCnetapp_i}
\end{align}

In what follows, we let $\mathbf{y} = (y_{1,1}, y_{2,1}, y_{1,2}, y_{2,2}, \ldots, y_{1,N}, y_{2,N})
 \in \mathbb{R}_{\ge}^{2 N}$ denote the concentration vector for the species $\mathcal{Y} = \{Y_{1,1}, Y_{2,1}, 
Y_{1,2}, Y_{2,2}, \ldots, Y_{1,N}, Y_{2,N} \}$.
\begin{theorem}{\rm (Positive equilibrium)}\label{theorem:equilibrium_bimolecular}
Let $\mathcal{R}_{\alpha} = \mathcal{R}_{\alpha}(\mathcal{X})$, with species $\mathcal{X} = \{X_1, X_2, \ldots, X_N\}$, 
be an arbitrary mass-action input network with the {\rm RRE}s~{\rm(\ref{eq:bimolecular_input})}.
Let $\bigcup_{i = 1}^N \mathcal{R}_{\beta,\gamma}^{\pm}(X_i, Y_{1,i}, Y_{2,i})$ 
be the controller with $\mathcal{R}_{\beta,\gamma}^{\pm}(X_i, Y_{1,i}, Y_{2,i})$
given by~{\rm(\ref{eq:IFCnetapp_i})}. Then, for any choice of the rate coefficients
$\boldsymbol{\alpha}$, $\boldsymbol{\beta}$ and $\boldsymbol{\gamma}$, 
the output network $\mathcal{R}_{\alpha,\beta,\gamma}(\mathcal{X}, \mathcal{Y}) = 
\mathcal{R}_{\alpha}(\mathcal{X}) \bigcup_{i = 1}^N 
\mathcal{R}_{\beta,\gamma}^{\pm}(X_i, Y_{1,i}, Y_{2,i})$ has a positive equilibrium 
$(\mathbf{x}^*, \mathbf{y}^*) \in \mathbb{R}_{>}^{3 N}$, with $x_i^* = (\beta_{0,i}/\gamma_{1,i}) > 0$ 
for all $i \in \{1, 2, \ldots, N\}$.
\end{theorem}

\begin{proof}
The RREs of the output network 
$\mathcal{R}_{\alpha}(\mathcal{X}) \bigcup_{i = 1}^N 
\mathcal{R}_{\beta,\gamma}^{\pm}(X_i, Y_{1,i}, Y_{2,i})$ are given by
\begin{align}
\frac{\mathrm{d} x_1}{\mathrm{d} t} & = f_1(\mathbf{x}; \, \boldsymbol{\alpha})   
+ \gamma_{2,1} y_{1,1} - \gamma_{3,1} x_1 y_{2,1}, 
\hspace{0.3cm} 
\frac{\mathrm{d} y_{1,1}}{\mathrm{d} t} = \beta_{0,1} - \beta_{1,1} y_{1,1} y_{2,1},
\hspace{0.3cm} 
\frac{\mathrm{d} y_{2,1}}{\mathrm{d} t} =  \gamma_{1,1} x_1 - \beta_{1,1} y_{1,1} y_{2,1},
\nonumber \\
\frac{\mathrm{d} x_2}{\mathrm{d} t} & = f_2(\mathbf{x}; \, \boldsymbol{\alpha})   
+ \gamma_{2,2} y_{1,2} - \gamma_{3,2} x_2 y_{2,2}, 
\hspace{0.3cm} 
\frac{\mathrm{d} y_{1,2}}{\mathrm{d} t} = \beta_{0,2} - \beta_{1,2} y_{1,2} y_{2,2},
\hspace{0.3cm} 
\frac{\mathrm{d} y_{2,2}}{\mathrm{d} t} =  \gamma_{1,2} x_2 - \beta_{1,2} y_{1,2} y_{2,2},
\nonumber \\
\vdots \nonumber \\
\frac{\mathrm{d} x_N}{\mathrm{d} t} & = f_N(\mathbf{x}; \, \boldsymbol{\alpha})   
+ \gamma_{2,N} y_{1,N} - \gamma_{3,N} x_N y_{2,N}, 
\hspace{0.1cm} 
\frac{\mathrm{d} y_{1,N}}{\mathrm{d} t} = \beta_{0,N} - \beta_{1,N} y_{1,N} y_{2,N},
\hspace{0.1cm} 
\frac{\mathrm{d} y_{2,N}}{\mathrm{d} t} =  \gamma_{1,N} x_N - \beta_{1,N} y_{1,N} y_{2,N},
\label{eq:bimolecular_output}
\end{align}
and the statement of the theorem follows using the same 
argument as in~(\ref{eq:RREs_debar3}) from Section~\ref{sec:nonlinear}.
\end{proof}

\label{lastpage}


\begin{thebibliography}{9}
\bibitem{SynthBio1} Endy D., 2005. 
Foundations for engineering biology. 
\textit{Nature}, \textbf{484}: 449--453.

\bibitem{Control1} Del Vecchio, D., Dy, A. J., Qian, Y., 2016. 
Control theory meets synthetic biology.
\textit{Journal of the Royal Society Interface}, \textbf{13}(\textbf{120}): 3--43.

\bibitem{Toggle}
Gardner, T. S., Cantor, C. R., Collins, J. J., 2000.
 Construction of a genetic toggle switch in Escherichia coli. 
\textit{Nature}, \textbf{403}: 339--342.

\bibitem{Repressilator}
Elowitz, M. B., Leibler, S., 2000 
A synthetic oscillatory network of transcriptional regulators. 
\textit{Nature}, \textbf{403}: 335--338.

\bibitem{PreTranscription} Chappell, J., Takahashi, M. K., Lucks, J. B., 2015. 
Creating small transcription activating RNAs. 
\textit{Nature chemical biology}, \textbf{11}(\textbf{3}): 214--220.

\bibitem{PostTranscription} Isaacs, F. J., Dwyer, D. J., Ding, C., Pervouchine, D. D., Cantor, C. R., Collins, J. J., 2004.
Engineered riboregulators enable post-transcriptional control of gene expression. 
\textit{Nature biotechnology}, \textbf{22}(\textbf{7}): 841--847. 

\bibitem{Adaptation} Drengstig, T., Ueda, H. R., Ruoff, P., 2008.
Predicting perfect adaptation motifs in reaction kinetic networks.
\textit{Journal of Physical Chemistry B}, \textbf{112}(\textbf{51}): 16752--16758. 

\bibitem{CellSignal} Ferrell, J. E., 2016.
Perfect and near-perfect adaptation in cell signaling.
\textit{Cell Systems}, \textbf{2}(\textbf{2}): 62--67.

\bibitem{Glycolic} Chandra, F. A., Buzi, G., Doyle, J. C., 2011.
 Glycolytic oscillations and limits on robust efficiency.
\textit{Science}, \textbf{333}(\textbf{6039}): 187--192. 

\bibitem{Chemotaxis1} Barkai, N., Leibler, S., 1997.
Robustness in simple biochemical networks.
\textit{Nature}, \textbf{387}: 913--917.

\bibitem{Chemotaxis2} Spiro, P., Parkinson, J., Othmer, H. G., 1997.
A model of excitation and adaptation in bacterial chemotaxis.
\textit{PNAS}, USA, \textbf{94}: 7263--7268.

\bibitem{Chemotaxis3} Yi, T. M., Huang, Y., Simon, M. I., Doyle, J., 2000.
Robust perfect adaptation in bacterial chemotaxis through integral feedback control.
\textit{PNAS}, \textbf{97} (\textbf{9}): 4649--4653. 

\bibitem{Feinberg} Feinberg, M. \textit{Lectures on Chemical Reaction Networks}, 
Delivered at the Mathematics Research Center, U. of Wisconsin, 1979.

\bibitem{Kurtz} Kurtz, T. G., 1972. 
The relationship between stochastic and deterministic models for chemical reactions. 
\textit{Journal of Chemical Physics}, \textbf{57}: 2976--2978. 

\bibitem{RadekBook} Erban, R., Chapman, J. 
\emph{Stochastic Modelling of Reaction-Diffusion Processes}. 
Cambridge Texts in Applied Mathematics, Cambridge University Press, 2019. 

\bibitem{CellCycle} Kar S., Baumann W. T., Paul M. R., Tyson J. J., 2009. 
Exploring the roles of noise in the eukaryotic cell cycle.
\textit{Proceedings of the National Academy of Sciences} USA, \textbf{106}: 6471--6476.

\bibitem{Circadian} Vilar, J. M. G., Kueh, H. Y., Barkai, N., Leibler, S., 2002.
Mechanisms of noise-resistance in genetic oscillators.
\textit{PNAS}, USA, \textbf{99} (\textbf{9}): 5988--5992.

\bibitem{Control_theory}
\r{A}str\"{o}m, K. J., and H\"{a}gglund, T. (1995). 
\emph{PID Controllers: Theory, Design, and Tuning}. 
Instrument Society of America, 1995.

\bibitem{Me_Homoclinic}  Plesa, T., Vejchodsk\'{y}, T., and Erban, R., 2016. 
Chemical Reaction Systems with a Homoclinic Bifurcation: An Inverse Problem.
 \textit{Journal of Mathematical Chemistry}, \textbf{54}(\textbf{10}): 1884--1915.

\bibitem{Biochemical_IFC} Oishi, K., and Klavins, E., 2011.
Biomolecular implementation of linear I/O systems.
\textit{IET Systems Biology}, Volume 5, Issue 4: 252--260.

\bibitem{Khammash} Briat, C., Gupta, A., Khammash, M., 2016. 
Antithetic integral feedback ensures robust perfect adaptation in noisy bimolecular networks. 
\textit{Cell Systems}, \textbf{2}(\textbf{1}): 15--26.

\bibitem{Khammash2} Aoki, S.K., Lillacci, G., Gupta, A., Baumschlager, A., Schweingruber, D., and Khammash, M., 2019. 
A universal biomolecular integral feedback controller for robust perfect adaptation. 
\emph{Nature}, \textbf{570}: 533--537.

\bibitem{AIFC_1} Olsman, N., Baetica, A. A., Xiao, F., Leong, Y.P., Doyle, J., and Murray, R., 2019.
Hard limits and performance tradeoffs in a class of antithetic integral feedback networks.
\textit{Cell Systems}, \textbf{9}(\textbf{1}): 49--63. 

\bibitem{AIFC_2} Olsman, N., Xiao, F., Doyle, J., 2019.
Architectural principles for characterizing the performance of antithetic integral feedback networks.
\textit{ISince}, \textbf{14}: 277--291. 

\bibitem{Burden} Boo, A., Ellis, T., Stan, G. B., 2019. 
 Host-aware synthetic biology.
\textit{Current Opinion in Systems Biology}, \textbf{14}: 66--72.

\bibitem{Me_Limitcycles} Plesa, T., Vejchodsk\'{y}, T., and Erban, R., 2017. Test Models for Statistical Inference: 
Two-Dimensional Reaction Systems Displaying Limit Cycle Bifurcations and Bistability. 
\textit{Stochastic Dynamical Systems, Multiscale Modeling, 
Asymptotics and Numerical Methods for Computational Cellular Biology}, 2017. 

\bibitem{Sigma}  Sharma, U. K., Chatterji, D., 2008. 
Differential mechanisms of binding of anti-sigma factors \emph{Escherichia coli} 
Rsd and bacteriophage T4 AsiA to \emph{E. coli} RNA polymerase lead to diverse physiological consequences. 
\emph{Journal of Bacteriology}, \textbf{190}: 3434--3443.

\bibitem{Khammash3} Gupta, A., and Khammash, M., 2019.
An antithetic integral rein controller for bio-molecular networks. 
\textit{IEEE 58th Conference on Decision and Control (CDC)}, Nice, France: 2808--2813.

\bibitem{GillespieSSA} Gillespie, D.T., 1977.
Exact stochastic simulation of coupled chemical reactions.
\textit{Journal of Physical Chemistry}, \textbf{81}(\textbf{25}): 2340--2361.

\bibitem{Pavliotis} Pavliotis, G.~A., Stuart, A.~M.
\textit{Multiscale Methods: Averaging and Homogenization}. 
Springer, New York, 2008.

\bibitem{LinearPositive} Dines, L., 1926. 
On Positive Solutions of a System of Linear Equations. 
\textit{Annals of Mathematics}, \textbf{28}(1/4), second series: 386--392.


\bibitem{AG} Cox, D., Little, J., and O'Shea, D.
\textit{Using algebraic geometry}. Springer, second edition, 2005.

\bibitem{Mmatrix} Fiedler, M., Pt\'ak, V., 1962. 
On matrices with non-positive off-diagonal elements and positive principal minors.
\textit{Czechoslovak Mathematical Journal}, \textbf{12}(\textbf{3}): 382--400.

\bibitem{PositiveSystems} Farina, L., Rinaldi, S.
\textit{Positive linear systems: Theory and applications}.
 John Wiley \& Sons, Inc, 10.1002/9781118033029, 2000.

\bibitem{VanKampen} Van Kampen, N. G. 
\textit{Stochastic processes in physics and chemistry}. 
Elsevier, 2007.

\bibitem{Me_Morphing} Plesa, T., Stan, G. B., Ouldridge, T. E., and Bae., W., 2021. 
Quasi-robust control of biochemical reaction networks via stochastic morphing.
Available as https://arxiv.org/abs/1908.10779.

\bibitem{Tikhonov} Klonowski, W., 1983. 
Simplifying principles for chemical and enzyme reaction kinetics. 
\textit{Biophys. Chem.} \textbf{18}(\textbf{3}): 73--87. 

\end{thebibliography}
\end{document}